\documentclass[12pt,english]{article}
\usepackage{lmodern}
\usepackage[T1]{fontenc}
\usepackage[latin9]{inputenc}
\usepackage{geometry}
\usepackage{color}
\usepackage{babel}
\usepackage{array}
\usepackage{float}
\usepackage{mathrsfs}
\usepackage{mathtools}
\usepackage{multirow}
\usepackage{amsmath}
\usepackage{amsthm}
\usepackage{amssymb}
\usepackage{graphicx}
\usepackage{setspace}
\usepackage[authoryear]{natbib}
\usepackage[unicode=true,
 bookmarks=true,bookmarksnumbered=false,bookmarksopen=false,
 breaklinks=false,pdfborder={0 0 0},pdfborderstyle={},backref=false,colorlinks=true]
 {hyperref}
\hypersetup{pdftitle={Player-Compatible Learning and Player-Compatible Equilibrium},
 pdfpagelayout=OneColumn, pdfnewwindow=true, pdfstartview=XYZ, plainpages=false, urlcolor=[rgb]{0.0430 ,0, 0.5}, linkcolor=[rgb]{0.0430 ,0, 0.5}, citecolor=[rgb]{0.0430 ,0, 0.5}, hypertexnames=false}

\makeatletter

\providecommand{\tabularnewline}{\\}
\theoremstyle{definition}
\newtheorem{defn}{\protect\definitionname}
\theoremstyle{plain}
\newtheorem{prop}{\protect\propositionname}
\theoremstyle{plain}
\newtheorem{thm}{\protect\theoremname}
\theoremstyle{definition}
 \newtheorem{example}{\protect\examplename}
\theoremstyle{plain}
\newtheorem{lem}{\protect\lemmaname}
\theoremstyle{plain}
\newtheorem{fact}{\protect\factname}

\usepackage{chngcntr}
\setcitestyle{round}
\usepackage{mathtools}

\makeatother

\providecommand{\definitionname}{Definition}
\providecommand{\examplename}{Example}
\providecommand{\factname}{Fact}
\providecommand{\lemmaname}{Lemma}
\providecommand{\propositionname}{Proposition}
\providecommand{\theoremname}{Theorem}

\begin{document}
\title{Player-Compatible Learning and Player-Compatible Equilibrium\thanks{We thank Alessandro Bonatti, Dan Clark, Glenn Ellison, Ben Golub,
Shengwu Li, Dave Rand, Alex Wolitzky, and Muhamet Yildiz for valuable
conversations and comments. We thank National Science Foundation grant
SES 1643517 for financial support, and Giacomo Lanzani for research
assistance.}}
\author{Drew Fudenberg\thanks{Department of Economics, MIT. Email: \texttt{\protect\href{mailto:drew.fudenberg\%40gmail.com}{drew.fudenberg@gmail.com}}}
\and Kevin He\thanks{California Institute of Technology and University of Pennsylvania.
Email: \texttt{\protect\href{mailto:hesichao\%40gmail.com}{hesichao@gmail.com}}}}
\date{{\normalsize{}}%
\begin{tabular}{rl}
First version: & September 23, 2017\tabularnewline
This version: & May 27, 2020\tabularnewline
\end{tabular}}
\maketitle
\begin{abstract}
\emph{Player-Compatible Equilibrium} (PCE) imposes cross-player restrictions
on the magnitudes of the players' \textquotedblleft trembles\textquotedblright{}
onto different strategies. These restrictions capture the idea that
trembles correspond to deliberate experiments by agents who are unsure
of the prevailing distribution of play. PCE selects intuitive equilibria
in a number of examples where trembling-hand perfect equilibrium \citep{selten1975reexamination}
and proper equilibrium \citep{myerson1978refinements} have no bite.
We show that rational learning and weighted fictitious play imply
our compatibility restrictions in a steady-state setting.
\end{abstract}

\thispagestyle{empty}
\setcounter{page}{0}
\begin{flushleft}
{\small{}\newpage}{\small\par}
\par\end{flushleft}

\section{Introduction}

Starting with \citet{selten1975reexamination}, a number of papers
have used the device of vanishingly small ``trembles'' to refine
the set of Nash equilibria. This paper introduces \emph{player-compatible
equilibrium }(PCE), which extends the tremble-based approach. PCE
requires that the trembles used to support an equilibrium satisfy
\emph{player compatibility}, which imposes restrictions on how one
player's trembles compare to those of another. We say player $i$
is more player-compatible with strategy $s_{i}^{*}$ than player $j$
is with strategy $s_{j}^{*}$ if whenever $s_{j}^{*}$ is optimal
for $j$ against some correlated profile $\sigma$, $s_{i}^{*}$ is
optimal for $i$ against any profile $\hat{\sigma}$ matching $\sigma$
in terms of the strategies of players other than $i$ and $j$. PCE
is invariant to the utility representations of players' preferences
over game outcomes, and provides a link between tremble-based refinements
and learning-in-games. As we will explain, PCE interprets ``trembles''
not as errors, but as players' deliberate experiments to learn how
others play. Its cross-player tremble restrictions derive from an
analysis of the relative frequencies of experiments that different
players choose to undertake over time under a number of commonly used
learning policies.

Section \ref{sec:PCE} defines player compatibility and PCE, studies
their basic properties, and proves that PCE exist in all finite games.
The compatibility relation is easiest to satisfy when $i$ and $j$
are ``non-interacting,'' meaning that their payoffs do not depend
on each other's play. But PCE can have bite even when all players
interact with each other, provided that the interactions are not too
strong. Moreover, as shown by the examples in Section \ref{sec:Examples-of-PCE},
PCE can rule out seemingly implausible equilibria that other tremble-based
refinements such as trembling-hand perfect equilibrium \citep{selten1975reexamination}
and proper equilibrium \citep{myerson1978refinements} cannot eliminate.

One of these examples is a ``link-formation game,'' where players
are split into two sides, and each player decides whether or not to
pay a cost to be \textbf{Active} and form links with all of the active
players on the other side. Players with lower costs are more compatible
with \textbf{Active} and so experiment with it more. In the ``anti-monotonic''
version of the game, players who incur a higher private cost of link
formation give lower benefits to their linked partners; in the ``co-monotonic''
version, higher cost players give others higher benefits. In the anti-monotonic
version the only PCE outcome is for all players to choose \textbf{Active},
because the experimentation of the low-cost players induces all players
on the other side to be \textbf{Active} as well. On the other hand,
both ``all \textbf{Active}'' and ``all \textbf{Inactive''} are
PCE outcomes in the co-monotonic case. In contrast to PCE making different
predictions in the two versions of the game, other equilibrium refinements
make the same predictions whether payoffs are anti-monotonic or co-monotonic.

To provide motivation for player-compatible trembles, we study a learning
framework where agents are born into different player roles and repeatedly
play the stage game. They face some time-invariant distribution of
opponents' play, as they would in a steady state of a model where
a continuum of anonymous agents are randomly matched each period.
We compare the experimentation behavior of agents in different player
roles who follow ``index learning policies.'' These policies assign
a numerical index to each strategy that only depends on data from
periods when that strategy was used, and play the strategy with the
highest index. We formulate an \emph{index compatibility} condition
for index policies, and use a coupling argument to show that any index
policies for $i$ and $j$ satisfying this index-compatibility condition
for strategies $s_{i}^{*}$ and $s_{j}^{*}$ will lead to $i$ experimenting
relatively more with $s_{i}^{*}$ than $j$ with $s_{j}^{*}$ over
their lifetimes against any distribution of opponents' play.

Index compatibility provides a general condition for $i$ to choose
$s_{i}^{*}$ more often than $j$ chooses $s_{j}^{*}$. This condition
applies across a range of observation structures and (not necessarily
optimal) learning policies. We link player compatibility with index
compatibility for two canonical learning policies in a class of ``factorable
games.'' In these games, playing a strategy $s_{i}$ reveals how
opponents played at all the information sets that are relevant for
$i's$ payoff when they play $s_{i},$ but gives no information about
the payoffs of $i's$ other strategies. We show that player compatibility
implies index compatibility for the rational learning policy given
by the Gittins index, and for the weighted fictitious play heuristic
\citep{cheung1997individual}. Interpreting trembles as play frequencies
during a learning process, our analysis provides a learning foundation
for PCE's cross-player tremble restrictions --- in the link-formation
game, for example, it justifies the idea that low-cost agents ``tremble
onto'' \textbf{Active} more frequently than high-cost ones do.

\subsection{Related Work}

\subsubsection{Tremble-Based Refinements}

Tremble-based solution concepts date back to \citet{selten1975reexamination},
who thanks Harsanyi for suggesting them. These solution concepts consider
totally mixed strategy profiles where players do not play an exact
best reply to the strategies of others, but may assign positive probability
to some or all strategies that are not best replies. Different solution
concepts in this class consider different kinds of ``trembles,''
but they all make predictions based on the limits of these non-equilibrium
strategy profiles as the probability of trembling tends to zero. Since
we compare PCE to these refinements below, we summarize them here
for the reader's convenience.

An \emph{$\epsilon$-perfect equilibrium} is a totally mixed strategy
profile where every non-best reply has weight less than $\epsilon$.
A limit of $\epsilon_{t}$-perfect equilibria where $\epsilon_{t}\to0$
is called a \emph{trembling-hand perfect equilibrium}. An \emph{$\epsilon$-proper
equilibrium} is a totally mixed strategy profile $\sigma$ where for
every player $i$ and strategies $s_{i}$ and $s_{i}'$ if $u_{i}(s_{i},\sigma_{-i})<u_{i}(s_{i}',\sigma_{-i})$
then $\sigma_{i}(s_{i})<\epsilon\cdot\sigma_{i}(s_{i}^{'})$. A limit
of $\epsilon_{t}$-proper equilibria where $\epsilon_{t}\to0$ is
called a \emph{proper equilibrium; }in this limit a more costly tremble
is infinitely less likely than a less costly one, regardless of the
cost difference.\emph{ Approachable equilibrium} \citep{van1987stability}
is also based on the idea that strategies with worse payoffs are played
less often. It too is the limit of $\epsilon_{t}$-perfect equilibria,
but where the players pay control costs to reduce their tremble probabilities.
When these costs are ``regular,'' all of the trembles are of the
same order. Because PCE does not require that the less likely trembles
are infinitely less likely than more likely ones, it is closer to
approachable equilibrium than to proper equilibrium. The \emph{strategic
stability }concept of \citet{kohlberg_strategic_1986} is also defined
using trembles, but applies to components of Nash equilibria as opposed
to single strategy profiles.

Unlike PCE, proper equilibrium and approachable equilibrium do not
impose cross-player restrictions on the relative probabilities of
various trembles. For this reason, these equilibrium concepts reduce
to perfect Bayesian equilibrium in signaling games with two possible
signals, such as the beer-quiche game of \citet{cho_signaling_1987},
when each type of the sender is viewed as a different player. They
do impose restrictions when applied to the ex-ante form of the game,
i.e., at the stage before the sender has learned their type. However,
as \citet{cho_signaling_1987} point out, evaluating the cost of mistakes
at the ex-ante stage of a signaling game means that the interim losses
are weighted by the prior distribution over sender types, so that
less likely types are more likely to tremble. In addition, applying
a different positive linear rescaling to each type's utility function
preserves every type's preference over lotteries on outcomes, but
changes the sets of proper and approachable equilibria, while such
utility rescalings have no effect on the set of PCE. In light of these
issues, we always apply tremble-based refinements at the interim stage
in Bayesian games.

Like PCE, \emph{extended proper equilibrium} \citep{milgrom_mollner_EP}
places restrictions on the relative probabilities of tremble by different
players, but it does so in a different way: An extended proper equilibrium
is the limit of $(\boldsymbol{\beta}$,$\epsilon_{t})-$proper equilibria,
where $\boldsymbol{\beta}=(\beta_{1},...\beta_{I})$ is a strictly
positive vector of utility re-scaling, and $\sigma_{i}(s_{i})<\epsilon_{t}\cdot\sigma_{j}(s_{j})$
if player $i$'s rescaled loss from $s_{i}$ (compared to the best
response) is less than $j$'s loss from $s_{j}$. In a signaling game
with only two possible signals, every Nash equilibrium where each
sender type strictly prefers not to deviate from their equilibrium
signal is an extended proper equilibrium at the interim stage, because
suitable utility rescalings for the types can lead to any ranking
of their utility costs of deviating to the off-path signal. By contrast,
Proposition \ref{prop:compare_CC} shows every PCE must satisfy the
compatibility criterion of \citet{fudenberg_he_2017}, which has bite
even in binary signaling games such as the beer-quiche example of
\citet{cho_signaling_1987}. So an extended proper equilibrium need
not be a PCE, a fact that Examples \ref{exa: restaurant} and \ref{exa:link4}
further demonstrate. Conversely, because extended proper equilibrium
makes some trembles infinitely less likely than others, it can eliminate
some PCE.\footnote{Example available on request.}

\subsubsection{The Learning Foundations of Equilibrium}

This paper builds on the work of \citet{fudenberg_steady_1993} and
\citet{fudenberg_learning_1995,fudenbergKreps1994learning} on learning
foundations for self-confirming and Nash equilibrium. It is also related
to recent work that that provides explicit learning foundations for
various equilibrium concepts that reflect ambiguity aversion, misspecified
priors, or model uncertainty, such as \citet*{battigalli2016analysis},
\citet*{battigalli2017LRBiases}, \citet{esponda_berknash_2016},
\citet*{fudenberg2020limits} and \citet{lehrer2012partially}. Unlike
those papers, we focus on characterizing the relative rates with which
different players experiment with strategies that are not myopically
optimal. For this reason our analysis of learning is closer to \citet{fudenberg_superstition_2006}
and \citet{fudenberg_he_2017}.

Our investigation of learning dynamics significantly expands on that
of \citet{fudenberg_he_2017}, which focused on a particular learning
policy (rational Bayesians) in a restricted set of games (signaling
games). In contrast, our analysis applies more broadly to any index
policies that satisfy an \emph{index compatibility} condition. We
show that two strategies of $i$ and $j$ ranked by player compatibility
lead to index-compatible learning policies in the class of ``factorable
games'' defined in Section \ref{sec:Factorability-and-Isomorphic},
under both rational learning and weighted fictitious play. We develop
new tools to deal with new issues that arise in these more general
games. For instance, \citet{fudenberg_he_2017} compare the Gittins
indices of different sender types in signaling games using the fact
that any stopping time (for the auxiliary optimal-stopping problem
defining the index) of the less-compatible type is also feasible for
the more-compatible type. But our general setting allows player roles
to interact, so it is not always valid to exchange the stopping times
of two different roles. A feasible stopping time for $i$ in the auxiliary
problem only conditions on past observations of $-i$'s play, but
the optimal stopping time for $j\ne i$ may condition on past observations
of $i$'s play in environments where $i$ and $j$ interact.

We deal with this problem by showing how $i$ can nevertheless construct
a feasible stopping time that \emph{mimics} an infeasible one of $j.$
Moreover, when a player faces more than one opponent, their optimal
experimentation policy may lead them to observe a correlated distribution
of opponents' play, even though the opponents do no actually play
correlated strategies. We discuss this issue of \emph{endogenous correlation}
in Section \ref{subsec:The-Heuristic-Gittins}; it is the reason we
define PCE in terms of correlated play.

In methodology the paper is related to other work on active learning
and experimentation. In single-agent settings, these include \citet{doval2018whether},
\citet{francetich2020choosing1,francetich2020choosing2}, and \citet{fryer2017two}.
In multi-agent settings additional issues arise such as free-riding
and encouraging others to learn, see e.g., \citet{bolton1999strategic},
\citet{keller2005strategic}, \citet{klein2011negatively}, \citet*{heidhues2015strategic},
\citet{frick2015innovation}, \citet*{halac2016optimal}, \citet{strulovici2010learning},
and the survey by \citet{horner2016learning}. Unlike most models
of multi-agent bandit problems, our agents only learn from personal
histories, not from the actions or histories of others. Our focus
is the comparison of experimentation policies under different payoff
parameters, which is central to PCE's cross-player tremble restrictions.

\section{\label{sec:PCE}Player-Compatible Equilibrium\protect 
}In this section, we develop a concept of the relative ``compatibility''
between two player-strategy pairs and discuss its properties. We then
introduce PCE, which builds cross-player tremble restrictions based
on this compatibility relation into an equilibrium concept.

Like proper equilibrium, PCE is defined on the strategic form of a
game. Of course many extensive forms can have the same strategic form,
and the learning motivation for PCE and player-compatible trembles
does depend on the underlying extensive form and the feedback structure,
but we postpone these issues until Section \ref{sec:Indexable-Learning-Rules}.

\subsection{Player Compatibility}

Consider a game in its strategic form with a finite set of players
$i\in\mathbb{I}$, finite strategy sets $|\mathbb{S}_{i}|\ge2$,\footnote{If $\mathbb{S}_{i}=\{s_{i}^{*}\}$ is a singleton, we would have $s_{i}^{*}\succsim s_{j}$
and $s_{j}\succsim s_{i}^{*}$ for any strategy $s_{j}$ of any player
$j$ if we follow the convention that the maximum over an empty set
is $-\infty$.} and utility functions $u_{i}:\mathbb{S}\to\mathbb{R}$, where $\mathbb{S}:=\times_{i}\mathbb{S}_{i}$.
For each $i,$ let $\Delta(\mathbb{S}_{i})$ denote the set of mixed
strategies for $i$. Let $\Delta^{\circ}(\mathbb{S})$ represent the
interior of $\Delta(\mathbb{S})$, the set of full-support correlated
strategy profiles. We now define an incomplete or partial order on
strategy-player pairs.
\begin{defn}
\label{def:compatible_with}For player $i\ne j$ and strategies $s_{i}^{*}\in\mathbb{S}_{i}$,
$s_{j}^{*}\in\mathbb{S}_{j}$, $i$ \emph{is more player-compatible
with $s_{i}^{*}$ than $j$ is with $s_{j}^{*}$,} abbreviated as
$s_{i}^{*}\succsim s_{j}^{*}$,\footnote{This notation is unambiguous provided $i$ and $j$ have disjoint
strategy sets. In the event that $i$ and $j$ share some strategies,
we will clarify this notation by attaching player subscripts.} if for every totally mixed correlated strategy profiles $\sigma\in\Delta^{\circ}(\mathbb{S})$
with 
\[
\sum_{s\in\mathbb{S}}u_{j}(s_{j}^{*},s_{-j})\cdot\sigma(s)=\max_{s_{j}^{'}\in\mathbb{S}_{j}}\sum_{s\in\mathbb{S}}u_{j}(s_{j}^{'},s_{-j})\cdot\sigma(s),
\]
we get 
\[
\sum_{s\in\mathbb{S}}u_{i}(s_{i}^{*},s_{-i})\cdot\tilde{\sigma}(s)>\max_{s_{i}^{''}\in\mathbb{S}_{i}\backslash\{s_{i}^{*}\}}\sum_{s\in\mathbb{S}}u_{i}(s_{i}^{''},s_{-i})\cdot\tilde{\sigma}(s)
\]
 for every totally mixed correlated strategy profile $\tilde{\sigma}\in\Delta^{\circ}(\mathbb{S})$
satisfying $\text{marg}_{-ij}(\sigma)=\text{marg}_{-ij}(\tilde{\sigma})$.
\end{defn}
In words, if $s_{j}^{*}$ is weakly optimal for the less-compatible
$j$ against $\sigma$, then $s_{i}^{*}$ is strictly optimal for
the more-compatible $i$ against any $\tilde{\sigma}$ whose marginal
on $-ij$'s play agrees with the marginal of $\sigma$. The compatibility
condition does not depend on the particular expected utility functions
used to represent the players' preferences over probability distributions
on $\mathbb{S}$.

We argue below in Theorem \ref{thm:crossplayer_tremble_foundation}
that play in the learning model is constrained by the compatibility
relation, and that result is stronger when the compatibility relation
is more complete. Since $\Delta^{\circ}(\mathbb{S})\subseteq\Delta(\mathbb{S})$,
the compatibility relation is more complete than an alternative definition
that replaces totally mixed strategy profiles with any correlated
strategy profile, and Theorem \ref{thm:crossplayer_tremble_foundation}
would continue to hold with this alternative definition. We restrict
to totally mixed strategies here to get a sharper conclusion. The
restriction fits with our assumptions in the learning model that all
agents have full-support prior beliefs about opponents' strategies
(for rational Bayesians) or strictly positive initial counts (for
weighted fictitious play). Conversely, since $\times_{i}\Delta^{\circ}(\mathbb{S}_{i})\subseteq\Delta^{\circ}(\mathbb{S}),$
our definition of compatibility ranks fewer strategy-player pairs
than an alternative definition that only considers mixed strategy
profiles with independent mixing between different opponents.\footnote{Formally, this alternative definition would replace ``totally mixed
correlated strategy profiles'' with ``independently and totally
mixed strategy profiles'' in the definition of $s_{i}^{*}\succsim s_{j}^{*}$.} We need to use the more stringent definition to match the microfoundations
of our compatibility-based cross-player restrictions: the definition
that only considers independent mixing imposes restrictions that the
learning model does not imply.\footnote{One form of our microfoundation for player-compatible trembles considers
rational learners who choose strategies based on their Gittins index.
Even for learners who hold independent beliefs about opponents' play
at different information sets, a strategy's Gittins index need not
be its expected payoff against independent randomizations by the opponents,
but we show that the index is always the expected payoff against some
correlated strategy profile.}

The compatibility relation is transitive, as the next proposition
shows.
\begin{prop}
\label{prop:transitive} Suppose $s_{i}^{*}\succsim s_{j}^{*}\succsim s_{k}^{*}$
where $s_{i}^{*},s_{j}^{*},s_{k}^{*}$ are strategies of $i,j,k$.
Then $s_{i}^{*}\succsim s_{k}^{*}$.
\end{prop}
The compatibility relation is also asymmetric, except in some ``corner
cases.'' Say that a strategy is \emph{strictly interior dominant}
if it is strictly better than any other strategy versus any totally
mixed strategy profile of the opponents, and similarly say that it
is \emph{strictly interior dominated}\footnote{Recall that a strategy can be strictly dominated even though it is
not strictly dominated by any pure strategy.}\emph{ }if it is strictly dominated versus totally mixed opponent
strategy profiles.
\begin{prop}
\label{prop:asymm} If $s_{i}^{*}\succsim s_{j}^{*}$, then at least
one of the following is true: (i) $s_{j}^{*}\not\succsim s_{i}^{*}$;
(ii) $s_{i}^{*}$ is strictly interior dominated for $i$ and $s_{j}^{*}$
is strictly interior dominated for $j$; (iii) $s_{i}^{*}$ is strictly
interior dominant for $i$ and $s_{j}^{*}$ is strictly interior dominant
for $j.$\footnote{The converse of this statement is not true since the relation $\succsim$
is not in general complete: we could have neither $s_{i}^{*}\succsim s_{j}^{*}$
nor $s_{j}^{*}\succsim s_{i}^{*}$.}
\end{prop}
The proofs of Propositions \ref{prop:transitive} and \ref{prop:asymm}
are straightforward; they can be found in the Online Appendix.

The definition of player compatibility simplifies in the following
special case. A game has a \emph{multipartite structure} if the set
of players $\mathbb{I}$ can be divided into $C$ mutually exclusive
classes, $\mathbb{I}=\mathbb{I}_{1}\cup...\cup\mathbb{I}_{C}$, in
such a way that whenever $i$ and $j$ belong to the same class $i,j\in\mathbb{I}_{c}$,
(1) they are \emph{non-interacting}, meaning neither player's payoff
depends on the other's strategy; and (2) they have the same strategy
set, $\mathbb{S}_{i}=\mathbb{S}_{j}$. Every Bayesian game has a multipartite
structure when each type is viewed as a different player. As another
example, we will later use a complete-information game with a multipartite
structure, the link-formation game (Example \ref{exa:link4}), to
illustrate both PCE and the learning motivation for player-compatible
trembles.

In a game with multipartite structure with $i,j\in\mathbb{I}_{c}$,
we can write $u_{i}(s_{c},s_{-ij})$ without ambiguity for $s_{c}\in\mathbb{S}_{i}$,
since all augmentations of the strategy profile $s_{-ij}$ with a
strategy by player $j$ lead to the same payoff for $i$. For $s_{c}^{*}\in\mathbb{S}_{i}=\mathbb{S}_{j}$,
the definition of $s_{ic}^{*}\succsim s_{jc}^{*}$ reduces\footnote{We use $s_{ic}^{*}$ to refer to $i$'s copy of $s_{c}^{*}$ and $s_{jc}^{*}$
to refer to $j$'s copy.} to: For every totally mixed correlated $\sigma$ with $\sigma_{-ij}\in\Delta^{\circ}(\mathbb{S}_{-ij})$,
\[
\sum_{s\in\mathbb{S}}u_{j}(s_{jc}^{*},s_{-ij})\cdot\sigma(s)=\max_{s_{j}^{'}\in\mathbb{S}_{j}}\sum_{s\in\mathbb{S}}u_{j}(s_{j}^{'},s_{-ij})\cdot\sigma(s)
\]
 implies 
\[
\sum_{s\in\mathbb{S}}u_{i}(s_{ic}^{*},s_{-ij})\cdot\sigma(s)>\max_{s_{i}^{''}\in\mathbb{S}_{i}\backslash\{s_{ic}^{*}\}}\sum_{s\in\mathbb{S}}u_{i}(s_{i}^{''},s_{-ij})\cdot\sigma(s).
\]

\subsection{Player-Compatible Trembles and PCE}

PCE is a tremble-based solution concept. It builds on and modifies
\citet{selten1975reexamination}'s definition of trembling-hand perfect
equilibrium (in the strategic form) as the limit of equilibria of
perturbed games in which agents are constrained to tremble, so we
begin by defining our notation for the trembles and the associated
constrained equilibria.
\begin{defn}
A \emph{tremble profile} $\boldsymbol{\epsilon}$ assigns a positive
number $\boldsymbol{\epsilon}(s_{i})>0$ to every player $i$ and
every pure strategy $s_{i}\in\mathbb{S}_{i}$. Given a tremble profile
$\boldsymbol{\epsilon}$, write $\Pi_{i}^{\boldsymbol{\epsilon}}$
for the set of\emph{ $\boldsymbol{\epsilon}$-strategies} of player
$i$, namely: 
\[
\Pi_{i}^{\boldsymbol{\epsilon}}\coloneqq\left\{ \sigma_{i}\in\Delta(\mathbb{S}_{i})\text{ s.t. }\sigma_{i}(s_{i})\ge\boldsymbol{\epsilon}(s_{i})\ \forall s_{i}\in\mathbb{S}_{i}\right\} .
\]
We call $\sigma^{\circ}$ an $\boldsymbol{\epsilon}$\emph{-equilibrium}
if for each $i$, 
\[
\sigma_{i}^{\circ}\in\underset{\sigma_{i}\in\Pi_{i}^{\boldsymbol{\epsilon}}}{\arg\max}\ u_{i}(\sigma_{i},\sigma_{-i}^{\circ}).
\]
\end{defn}
Note that $\Pi_{i}^{\boldsymbol{\epsilon}}$ is compact and convex.
It is also non-empty when $\boldsymbol{\epsilon}$ is close enough
to $\boldsymbol{0}$. By standard results, whenever $\boldsymbol{\epsilon}$
is small enough so that $\Pi_{i}^{\boldsymbol{\epsilon}}$ is non-empty
for each $i$, an $\boldsymbol{\epsilon}$-equilibrium exists.

The key building block for PCE is \emph{$\boldsymbol{\epsilon}$-}PCE,
which is an $\boldsymbol{\epsilon}$-equilibrium where the tremble
profile is ``co-monotonic'' with $\succsim$ in the following sense:
\begin{defn}
Tremble profile $\boldsymbol{\epsilon}$ is \emph{player-compatible}
if $\boldsymbol{\epsilon}(s_{i}^{*})\ge\boldsymbol{\epsilon}(s_{j}^{*})$
for all $i,j,s_{i}^{*},s_{j}^{*}$ such that $s_{i}^{*}\succsim s_{j}^{*}$.
An $\boldsymbol{\epsilon}$-equilibrium where $\boldsymbol{\epsilon}$
is player-compatible is called a \emph{player-compatible $\boldsymbol{\epsilon}$-equilibrium
}(or\emph{ $\boldsymbol{\epsilon}$-PCE}).
\end{defn}
The condition on $\boldsymbol{\epsilon}$ says the minimum weight
$i$ could assign to $s_{i}^{*}$ is no smaller than the minimum weight
$j$ could assign to $s_{j}^{*}$ in the constrained game,
\[
\min_{\sigma_{i}\in\Pi_{i}^{\boldsymbol{\epsilon}}}\sigma_{i}(s_{i}^{*})\ge\min_{\sigma_{j}\in\Pi_{j}^{\boldsymbol{\epsilon}}}\sigma_{j}(s_{j}^{*}).
\]
This is a ``cross-player tremble restriction,'' that is, a restriction
on the relative probabilities of trembles by different players. Note
that this restriction, like the player compatibility relation, depends
on the players' preferences over distributions on $\mathbb{S}$ but
not on the particular utility representation. This invariance property
distinguishes player-compatible trembles from other models of stochastic
behavior such as the stochastic terms in logit best responses. Our
learning foundation will interpret these trembles not as mistakes,
but as deliberate experiments by agents trying to learn how others
play.

As is usual for tremble-based equilibrium refinements, we now define
PCE as the limit of a sequence of $\boldsymbol{\epsilon}$-PCE where
$\boldsymbol{\epsilon}\to\boldsymbol{0}$.
\begin{defn}
A strategy profile $\sigma^{*}$ is a\emph{ player-compatible equilibrium
(PCE)} if there exists a sequence of player-compatible tremble profiles
$\boldsymbol{\epsilon}^{(t)}\to\boldsymbol{0}$ and an associated
sequence of strategy profiles $\sigma^{(t)},$ where each $\sigma^{(t)}$
is an $\boldsymbol{\epsilon}^{(t)}$-PCE, such that $\sigma^{(t)}\to\sigma^{*}$.
\end{defn}
We think of PCE as primarily a solution concept for games with three
or more players, where the relative tremble probabilities of $i\ne j$
affect some third party's best response. In two-player games, $s_{i}^{*}\succsim s_{j}^{*}$
only when $s_{j}^{*}$ is strictly interior dominated or $s_{i}^{*}$
is strictly interior dominant.

The cross-player restrictions embodied in player-compatible trembles
translate into analogous restrictions on PCE, as shown in the next
result.
\begin{prop}
\label{prop:PCE_compatible} For any PCE $\sigma^{*}$, player $k$,
and strategy $\bar{s}_{k}$ such that $\sigma_{k}^{*}(\bar{s}_{k})>0,$
there exists a sequence of totally mixed strategy profiles $\sigma_{-k}^{(t)}\to\sigma_{-k}^{*}$
such that

(i) for every pair $i,j\ne k$ with $s_{i}^{*}\succsim s_{j}^{*}$,

\[
\liminf_{t\to\infty}\frac{\sigma_{i}^{(t)}(s_{i}^{*})}{\sigma_{j}^{(t)}(s_{j}^{*})}\ge1;
\]
 and (ii) $\bar{s}_{k}$ is a best response for $k$ against every
$\sigma_{-k}^{(t)}$ .
\end{prop}
The proof is in the Appendix, as are the proofs of subsequent results
except where otherwise stated.

Treating each $\sigma_{-k}^{(t)}$ as a totally mixed approximation
to $\sigma_{-k}^{*},$ in a PCE each player $k$ essentially best
responds to totally mixed opponent play that respects player compatibility.

It is easy to show that every $\boldsymbol{\epsilon}$-PCE respects
player compatibility up to the ``adding up constraint'' that probabilities
on different strategies must sum up to 1 and $i$ must place probability
no smaller than $\boldsymbol{\epsilon}(s_{i}^{'})$ on strategies
$s_{i}^{'}\ne s_{i}^{*}$. The ``up to'' qualification disappears
in the $\boldsymbol{\epsilon}^{(t)}\to\boldsymbol{0}$ limit because
the required probabilities on $s_{i}^{'}\ne s_{i}^{*}$ tend to 0.

Since PCE is defined as the limit of $\boldsymbol{\epsilon}$-equilibria
for a restricted class of trembles, the set of PCE form a subset of
trembling-hand perfect equilibria; the next result shows this subset
is not empty. It uses the fact that the tremble profiles with the
same lower bound on the probability of each action satisfy the compatibility
condition in any game.
\begin{thm}
\label{thm:PCE_existence} A PCE exists in every finite strategic-form
game.
\end{thm}

\subsection{Learning and Player-Compatible Trembles\protect 
}Sections \ref{sec:Indexable-Learning-Rules} and \ref{sec:Factorability-and-Isomorphic}
provide a microfoundation for player-compatible trembles, which form
the core innovation of PCE. To preview the results, Sections \ref{sec:Indexable-Learning-Rules}
presents a general sufficient condition for $i$ to experiment more
with $s_{i}^{*}$ than $j$ does with $s_{j}^{*}$ over their lifetimes
that is applicable across a range of learning environments and learning
policies. Sections \ref{sec:Factorability-and-Isomorphic} completes
the story by showing that in a class of games that includes our Section
\ref{sec:Examples-of-PCE} examples, the player-compatibility condition
$s_{i}^{*}\succsim s_{j}^{*}$ implies Sections \ref{sec:Indexable-Learning-Rules}'s
sufficient condition for the rational learning policy and for weighted
fictitious play. For analyzing rational behavior, we consider agents
who start with the same priors over the play of their opponents. We
believe we could extend this conclusion to agents with slightly different
priors using a stronger notion of player compatibility, but we do
not pursue this result here.\footnote{To do this, we would measure the ``strength'' of the compatibility
ranking by saying that $i$ \emph{is $\lambda$ more player-compatible
with $s_{i}^{*}$ than $j$ is with $s_{j}^{*}$ }if the inequality
in the definition $s_{i}^{*}\succsim s_{j}^{*}$ holds for all $\tilde{\sigma}\in\Delta^{\circ}(\mathbb{S})$
satisfying $||\text{marg}_{-ij}(\sigma)-\text{marg}_{-ij}(\tilde{\sigma})||\leq\lambda$.
We believe that our learning foundation would extend to cases where
the agents' priors are sufficiently close compared to $\text{\ensuremath{\lambda}.}$}

PCE incorporates cross-player tremble restrictions derived from learning
into an equilibrium framework. It also imposes some extra restrictions
that we do not microfound, so the map from learning models to PCE
is inexact, but PCE does capture some novel implications of learning.
Like any game-theoretic equilibrium concept, it provides a reduced
form that allows analysts to study comparative statics in various
applications without needing to solve the dynamic learning problem
anew in each of them.

In Section \ref{sec:Replication-Invariance-of-PCE}, we expand the
game to include duplicate copies of some of the original strategies,
where two strategies are duplicates if they provide exactly the same
payoff and exactly the same information.\footnote{Two strategies with the same payoffs that give different information
about opponents' play are not equivalent in our learning model.} If $s_{i}^{*}\succsim s_{j}^{*}$ in the original game, then in the
expanded game we impose the cross-player tremble restriction that
the probability of $i$ trembling onto the set of copies of $s_{i}^{*}$
is larger than the probability of $j$ trembling onto the set of copies
of $s_{j}^{*}$. The way we update our PCE definition in the presence
of duplicates fits our interpretation of trembles as experimentation
frequencies: As we show, the sum of $i$'s lifetime experimentation
frequencies with all duplicates of $s_{i}^{*}$ is independent of
the number of duplicates under both rational behavior and weighted
fictitious play. We show that the set of PCE in the expanded game
with these new tremble restrictions is the same as the set of PCE
in the original game.

\section{\label{sec:Examples-of-PCE}Examples of PCE}

In this section, we study examples of games where PCE rules out unintuitive
Nash equilibria. We will also use these examples to distinguish PCE
from existing refinements.

\subsection{The Restaurant Game}

We start with a complete-information game where PCE differs from other
solution concepts.
\begin{example}
\label{exa: restaurant} There are three players in the game: a food
critic, a regular diner, and a restaurant. Simultaneously, the restaurant
decides between ordering high-quality (\textbf{H}) or low-quality
(\textbf{L}) ingredients, while the critic and the diner decide whether
to go eat at the restaurant (\textbf{R}) or order pizza (\textbf{Z})
and eat at home. The utility from \textbf{Z} is normalized to 0. If
both customers choose \textbf{Z}, the restaurant also gets 0 payoff.
Otherwise, the restaurant's payoff depends on the ingredient quality
and clientele. Choosing \textbf{L} yields a profit of +2 per customer
while choosing \textbf{H} yields a profit of +1 per customer. In addition,
if the food critic is present she will write a review based on ingredient
quality, which affects the restaurant's payoff by $\pm2.5.$ Each
customer gets a payoff of $x<-1$ from consuming food made with low-quality
ingredients and a payoff of $y>0.5$ from consuming food made with
high-quality ingredients, while the critic gets an additional +1 payoff
from going to the restaurant and writing a review (regardless of food
quality). Customers each incur a 0.5 congestion cost if they both
go to the restaurant. We depict this situation in the game tree below,
with $c$ and $d$ subscripts denoting strategies of the critic and
the diner. \begin{center}\includegraphics[scale=0.45]{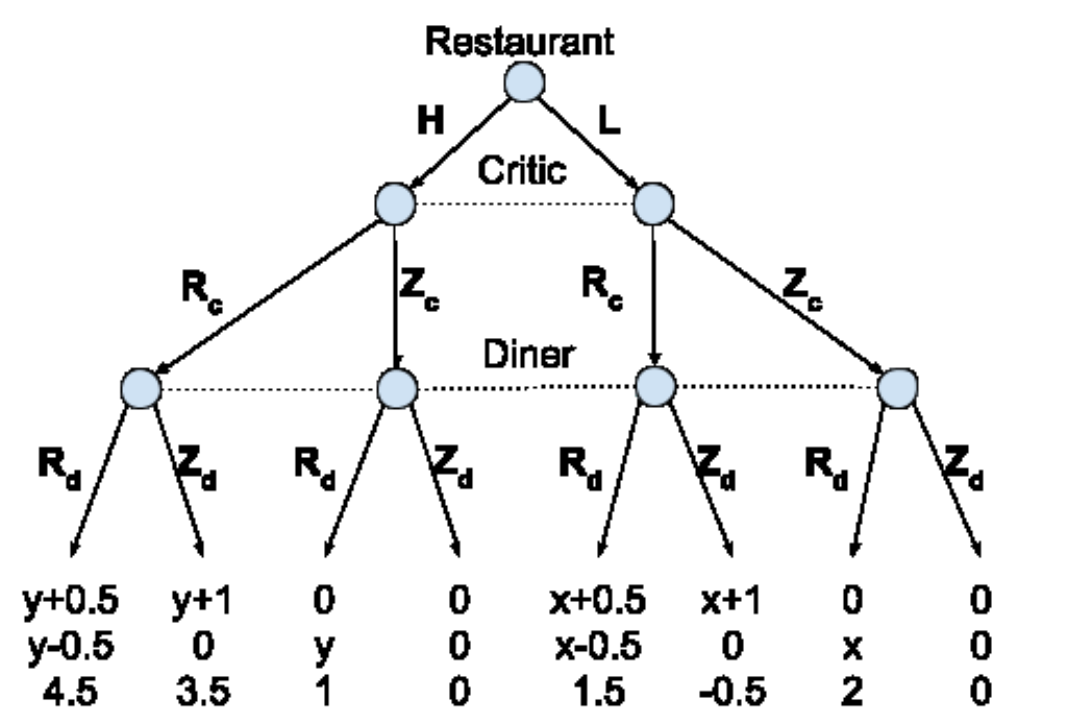}\end{center}

The strategies of the two customers affect each other's payoffs, so
the critic and the diner are not non-interacting players. In particular,
they cannot be mapped into two types of the same agent in a Bayesian
game.

The strategy profile (\textbf{$\boldsymbol{Z_{c}}$}, $\boldsymbol{Z_{d}}$,
\textbf{L}) is a proper equilibrium, sustained by the restaurant's
belief that when at least one customer plays \textbf{R}, it is far
more likely that the diner deviated to patronizing the restaurant
than the critic, even though the critic has a greater incentive to
go to the restaurant since she gets paid for writing reviews. It is
also an extended proper equilibrium.\footnote{(\textbf{$\boldsymbol{Z_{c}}$}, $\boldsymbol{Z_{d}}$, \textbf{L})
is an extended proper equilibrium, because scaling the critic's payoff
by a large positive constant makes it more costly for the critic to
deviate to \textbf{R1} than for the diner to deviate to \textbf{R2}.}

We claim that $\boldsymbol{R_{c}}\succsim\boldsymbol{R_{d}}$. Note
that for any profile $\sigma$ of totally mixed, correlated play that
makes the diner indifferent between $\boldsymbol{Z_{d}}$ and $\boldsymbol{R_{d}}$,
we must have $u_{1}(\boldsymbol{R_{c}},\tilde{\sigma}_{-c})\ge0.5$
for any profile $\tilde{\sigma}$ that agrees with $\sigma$ in terms
of the restaurant's play. The critic's utility from \textbf{$\boldsymbol{R_{c}}$}
is minimized when the diner chooses \textbf{$\boldsymbol{R_{d}}$}
with probability 1, but even then the critic gets 0.5 higher utility
from going to a crowded restaurant than the diner gets from going
to an empty restaurant, holding fixed food quality at the restaurant.
This shows $\boldsymbol{R_{c}}\succsim\boldsymbol{R_{d}}$.

Whenever $\sigma_{c}^{(t)}(\text{\ensuremath{\boldsymbol{R_{c}}})}/\sigma_{d}^{(t)}(\text{\ensuremath{\boldsymbol{R_{d}}})}>\frac{1}{4}$,
the restaurant strictly prefers \textbf{H} over \textbf{L}. Thus by
Proposition \ref{prop:PCE_compatible}, there is no PCE where the
restaurant plays \textbf{L} with positive probability. \hfill{}$\blacklozenge$
\end{example}

\subsection{\label{subsec:The-link-formation-game}The Link-Formation Game}

In the next example, PCE makes different predictions in two versions
of a game with different payoff parameters, while all other solution
concepts we know of make the same predictions in both versions.
\begin{example}
\label{exa:link4}There are 4 players in the game, split into two
sides: North and South. The players are named North-1, North-2, South-1,
and South-2, abbreviated as N1, N2, S1, and S2.

These players engage in a strategic link-formation game. Each player
simultaneously takes an action: either \textbf{Inactive} or \textbf{Active}.
An \textbf{Inactive} player forms no links. An \textbf{Active} player
forms a link with every \textbf{Active} player on the opposite side.
(Two players on the same side cannot form links.) For example, suppose
N1 plays \textbf{Active}, N2 plays \textbf{Active}, S1 plays \textbf{Inactive},
and S2 plays \textbf{Active}. Then N1 creates a link to S2, N2 creates
a link to S2, S1 creates no links, and S2 creates links to both N1
and N2.

\begin{center}\includegraphics[scale=0.4]{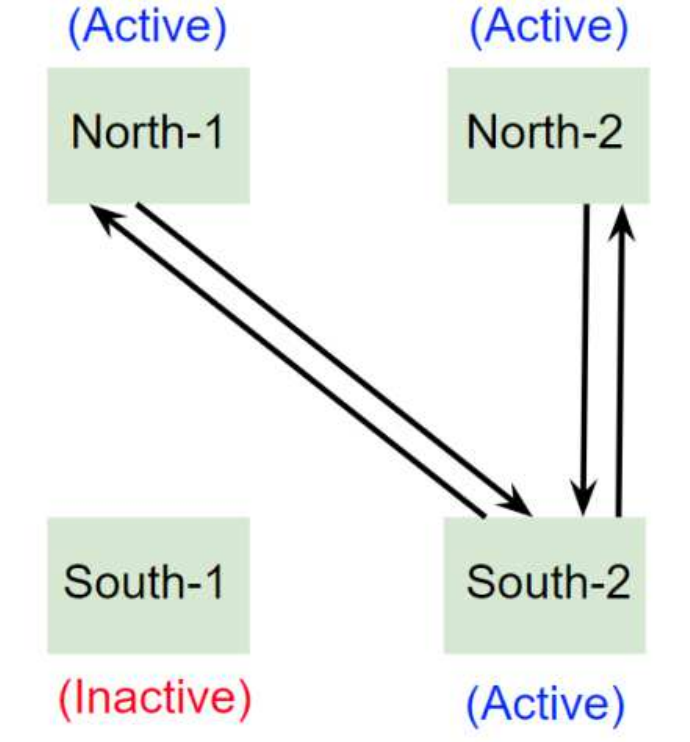}\end{center}

Each player $i$ is characterized by two parameters: cost ($c_{i}$)
and quality $(q_{i})$. Cost refers to the private cost that a player
pays for each link they create. Quality refers to the benefit that
a player provides to others when they link to her. A player who forms
no links gets a payoff of 0. In the above example, the payoff to North-1
is $q_{\text{S2}}-c_{\text{N1}}$ and the payoff to South-2 is $(q_{\text{N1}}-c_{\text{S2}})+(q_{\text{N2}}-c_{\text{S2}})$.

We consider two specifications of the payoff functions. In the \emph{anti-monotonic}
version on the left, players with a higher cost have a lower quality.
In the \emph{co-monotonic} version on the right, players with a higher
cost have a higher quality. There are two pure-strategy Nash outcomes
for each version: all links form or no links form. ``All links form''
is the unique PCE outcome in the anti-monotonic case, while both ``all
links'' and ``no links'' are PCE outcomes under co-monotonicity.

\begin{center}%
\begin{tabular}{|c|c|c|}
\hline 
\multicolumn{3}{|c|}{Anti-Monotonic}\tabularnewline
\hline 
\textbf{Player} & \textbf{Cost} & \textbf{Quality}\tabularnewline
\hline 
\hline 
North-1 & 14 & 30\tabularnewline
\hline 
North-2 & 19 & 10\tabularnewline
\hline 
\hline 
South-1 & 14 & 30\tabularnewline
\hline 
South-2 & 19 & 10\tabularnewline
\hline 
\end{tabular}$\quad$ %
\begin{tabular}{|c|c|c|}
\hline 
\multicolumn{3}{|c|}{Co-Monotonic}\tabularnewline
\hline 
\textbf{Player} & \textbf{Cost} & \textbf{Quality}\tabularnewline
\hline 
\hline 
North-1 & 14 & 10\tabularnewline
\hline 
North-2 & 19 & 30\tabularnewline
\hline 
\hline 
South-1 & 14 & 10\tabularnewline
\hline 
South-2 & 19 & 30\tabularnewline
\hline 
\end{tabular}\end{center}

PCE makes different predictions in these two versions of the game
because the compatibility structure with respect to own quality is
reversed. In both versions, $\text{Active}_{N1}\succsim\text{Active}_{N2},$
but N1 has high quality in the anti-monotonic version, and low quality
in the co-monotonic version. Thus, in the anti-monotonic version but
not in the co-monotonic version, player-compatible trembles lead to
the high-quality counterparty choosing \textbf{Active} at least as
often as the low-quality counterparty, which means \textbf{Active}
has a positive expected payoff even when one's own cost is high.

In contrast, the set of equilibria that satisfy extended proper equilibrium,
proper equilibrium, trembling-hand perfect equilibrium, $p$-dominance,
Pareto efficiency, and strategic stability do not depend on whether
payoffs are anti-monotonic or co-monotonic, as shown in the Online
Appendix. \hfill{}$\blacklozenge$
\end{example}

\subsection{Signaling Games}

Recall that a signaling game is a two-player Bayesian game, where
P1 is a sender who knows their own type $\theta,$ and P2 only knows
that P1's type is drawn according to the distribution $\lambda\in\Delta(\Theta)$
on a finite type space $\Theta$. After learning their type, the sender
sends a signal $s\in S$ to the receiver. Then, the receiver responds
with an action $a\in A$. Utilities depend on the sender's type $\theta$,
the signal $s$, and the action $a$.

\citet{fudenberg_he_2017}'s compatibility criterion is defined only
for signaling games. It does not use limits of games with trembles,
but instead restricts the beliefs that the receiver can have about
the sender's type. That sort of restriction does not seem easy to
generalize beyond games with observed actions, while using trembles
allows us to define PCE for general strategic-form games. As we will
see, the more general PCE definition implies the compatibility criterion
in signaling games.

With each sender type viewed as a different player, this game has
$|\Theta|+1$ players, $\mathbb{I}=\Theta\cup\{2\}$, where the strategy
set of each sender type $\theta$ is $\mathbb{S}_{\theta}=S$ while
the strategy set of the receiver is $\mathbb{S}_{2}=A^{S}$, the set
of signal-contingent plans. So a mixed strategy of $\theta$ is a
possibly mixed signal choice $\sigma_{1}(\cdot\mid\theta)\in\Delta(S),$
while a mixed strategy $\sigma_{2}\in\Delta(A^{S})$ of the receiver
is a mixed plan about how to respond to each signal.

\citet{fudenberg_he_2017} define\emph{ type compatibility} for signaling
games. A signal $s^{*}$ is more type-compatible with $\theta$ than
with $\theta^{'}$ if for every behavioral strategy $\sigma_{2}$,

\[
u_{1}(s^{*},\sigma_{2};\theta^{'})\ge\max_{s^{'}\ne s^{*}}u_{1}(s^{'},\sigma_{2};\theta^{'})
\]
 implies 
\[
u_{1}(s^{*},\sigma_{2};\theta)>\max_{s^{'}\ne s^{*}}u_{1}(s^{'},\sigma_{2};\theta).
\]
They also define the \emph{compatibility criterion}, which imposes
restrictions on off-path beliefs in signaling games. Consider a Nash
equilibrium $\sigma_{1}^{*},\sigma_{2}^{*}$. For any signal $s^{*}$
and receiver action $a$ with $\sigma_{2}^{*}(a\mid s^{*})>0$, the
compatibility criterion requires that $a$ best responds to some belief
$p\in\Delta(\Theta)$ about the sender's type such that, whenever
$s^{*}$ is more type-compatible with $\theta$ than with $\theta^{'}$
and $s^{*}$ is not equilibrium dominated\footnote{Signal $s^{*}$ is not equilibrium dominated for $\theta$ if $\max_{a\in A}u_{1}(s^{*},a;\theta)>u_{1}(s_{1},\sigma_{2}^{*};\theta)$
for every $s_{1}$ with $\sigma_{1}^{*}(s_{1}\mid\theta)>0.$} for $\theta$, $p$ satisfies $\frac{p(\theta^{'})}{p(\theta)}\le\frac{\lambda(\theta^{'})}{\lambda(\theta)}.$

Since every totally mixed strategy of the receiver is payoff-equivalent
to a behavioral strategy, it is easy to see that type compatibility
implies $s_{\theta}^{*}\succsim s_{\theta^{'}}^{*}$.\footnote{The converse does not hold. We defined type compatibility to require
testing against all receiver strategies and not just the totally mixed
ones, so it is possible that $s_{\theta}^{*}\succsim s_{\theta^{'}}^{*}$
but $s^{*}$ is not more type-compatible with $\theta$ than with
$\theta^{'},$ so type-compatibility is harder to satisfy than player
compatibility. We now realize that we could have restricted type compatibility
to only consider totally mixed strategies, and all of the results
of \citet{fudenberg_he_2017} would still hold.} The next result shows that when specialized to signaling games, all
PCE pass the compatibility criterion.
\begin{prop}
\label{prop:compare_CC}In a signaling game, every PCE $\sigma^{*}$
is a Nash equilibrium satisfying the compatibility criterion of \citet{fudenberg_he_2017}.
\end{prop}
This proposition in particular implies that in the beer-quiche game
of \citet{cho_signaling_1987}, the quiche-pooling equilibrium is
not a PCE, as it does not satisfy the compatibility criterion.

\section{\label{sec:Indexable-Learning-Rules}Index Learning Policies and
Index Compatibility}

This section characterizes a general class of ``index learning policies''
that lead $i$ to experiment more with $s_{i}^{*}$ than $j$ does
with $s_{j}^{*}$. The next section shows that optimal learning behavior
and weighted fictitious play belong to this class in ``factorable''
games, when $s_{i}^{*}\succsim s_{j}^{*}$. Together, these sections
link the player-compatibility relation with agents' learning behavior
under various learning policies, providing a learning foundation for
the tremble restrictions central to PCE.

The learning problem the players face depends on what they observe
about the play of others, which in turn depends on the extensive form
of the game, denoted by $\Gamma$. This game has a set of players
$i\in\mathbb{I}$ and also a player 0 that we will use to model Nature's
moves. The collection of information sets of player $i\in\mathbb{I}$
is written as $\mathcal{H}_{i}$. At each $h\in\mathcal{H}_{i}$,
player $i$ chooses an action $a_{h}$ from the finite set of possible
actions $A_{h}$. An extensive-form pure strategy of $i$ specifies
an action at each information set $h\in\mathcal{H}_{i}$. We denote
by $\mathbb{S}_{i}$ the set of all such strategies. Let $Z$ be the
set of terminal vertices of $\Gamma$. Also, let $Z(s)$ denote the
terminal vertex reached under the pure strategy profile $s\in\times_{i}\mathbb{S}_{i}.$

We consider an agent born into player role $i$ who maintains this
role throughout their life. They have has a geometrically distributed
lifetime with $0\le\gamma<1$ probability of survival between periods.
All agents share the same survival chance. Each period, the agent
plays the stage game $\Gamma$, choosing a strategy $s_{i}\in\mathbb{S}_{i}$.
Then, with probability $\gamma$, they continue into the next period
and play the stage game again. With complementary probability they
exit the system.

Each player is equipped with a finite set of \emph{observations} $\mathbb{O}_{i}$
and a \emph{feedback function} $\mathfrak{o}_{i}:Z\to\mathbb{O}_{i}$
that maps the terminal node reached to an observation. We assume each
player has perfect recall and remembers their chosen strategy.
\begin{defn}
The set of all finite \emph{histories} of all lengths for $i$ is
$Y_{i}:=\cup_{t\ge0}(\mathbb{S}_{i}\times\mathbb{O}_{i})^{t}$. For
a history $y_{i}\in Y_{i}$ and $s_{i}\in\mathbb{S}_{i}$, the \emph{subhistory}
$y_{i,s_{i}}$ is the (possibly empty) subsequence of $y_{i}$ where
the agent played $s_{i}$.
\end{defn}
To compare players $i$ and $j$'s relative experimentation probabilities,
we need the agents in these two player roles to face ``equivalent''
learning environments. The next definition formalizes what this means;
it is a joint restriction on the game tree $\Gamma$ and the feedback
structures $(\mathbb{O}_{i},\mathfrak{o}_{i})$, $(\mathbb{O}_{j},\mathfrak{o}_{j})$
of the two player roles.
\begin{defn}
Agents $i$ and $j$ face \emph{isomorphic learning problems} if:
\begin{itemize}
\item There exists an isomorphism $\varphi:\mathbb{S}_{i}\to\mathbb{S}_{j}$
\item For each $s_{i}\in\mathbb{S}_{i},$ the union $(\{s_{i}\}\times\mathbb{O}_{i})\cup(\{s_{j}\}\times\mathbb{O}_{j})$
of possible observations of $i$ after $s_{i}$ and possible observations
of $j$ after $s_{j}$ can be partitioned into a set of equivalence
classes $\mathbb{E}_{s_{i},s_{j}},$ where elements $E\in\mathbb{E}_{s_{i},s_{j}}$
are disjoint subsets.
\item For each pure strategy profile $\tilde{s}$ and $s_{i}\in\mathbb{S}_{i}$
with $s_{j}=\varphi(s_{i}),$ $(s_{i},\mathfrak{o}_{i}(Z(s_{i},\tilde{s}_{-i})))$
and $(s_{j},\mathfrak{o}_{j}(Z(s_{j},\tilde{s}_{-j})))$ belong to
the same equivalence class in $\mathbb{E}_{s_{i},s_{j}}.$
\end{itemize}
\end{defn}
Whenever we say $i$ and $j$ face isomorphic learning problems, it
is always with respect to some isomorphism $\varphi$ and some collection
of equivalence classes $\mathbb{E}=(\mathbb{E}_{s_{i},\varphi(s_{i})})_{s_{i}\in S_{i}}$.
To make this explicit, sometimes we will say $i$ and $j$ face \emph{$(\varphi,\mathbb{E})$-isomorphic
learning problems}.

If $i$ and $j$ face isomorphic learning problems, there is a bijection
$\varphi$ between the strategies of the two agents and a notion of
equivalence between the one-period histories they might observe after
choosing $s_{i}$ and $s_{j}$, with $s_{j}=\varphi(s_{i}).$ This
equivalence puts $i$'s observation after $(s_{i},\tilde{s}_{-i})$
in the same equivalence class as $j$'s observation after $(s_{j},\tilde{s}_{-j}),$
for any profile $\tilde{s}$.

Consider Example \ref{exa: restaurant} when the Critic and the Diner
observe all other players' actions if they choose \textbf{R}, but
observe nothing if they choose \textbf{Z}. That is, 
\[
\mathbb{O}_{C}=\mathbb{O}_{D}=\{(L,R),(L,Z),(H,R),(H,Z),\varnothing\}.
\]
 Consider the natural isomorphism $\varphi(R_{c})=R_{d}$ and $\varphi(Z_{c})=Z_{d}$,
and partition histories after $R_{c}$ and $R_{d}$ as 
\begin{align*}
\mathbb{E}_{R_{c},R_{d}}= & \left\{ \left\{ (R_{c},(L,R)),(R_{d},(L,R)),(R_{c},(L,Z)),(R_{d},(L,Z))\right\} ,\right.\\
 & \left.\left\{ (R_{c},(H,R)),(R_{d},(H,R)),(R_{c},(H,Z)),(R_{d},(H,Z))\right\} \right\} .
\end{align*}
The two equivalence classes in $\mathbb{E}_{R_{c},R_{d}}$ represent
whether the restaurant is observed to play \textbf{L} or \textbf{H}.
Also let $\mathbb{E}_{Z_{c},Z_{d}}=\{\{(Z_{c},\varnothing),(Z_{d},\varnothing)\}\}$
contain just one equivalence class, which corresponds to no observations.
It is clear that given any pure strategy profile $s$, $(R_{c},s_{-c})$
and $(R_{d},s_{-d})$ lead to the same histories, up to equivalence
classes (i.e., the same observation of the restaurant's play.)

We extend the notion of equivalence to histories with more than one
period in the natural way.
\begin{defn}
Suppose $i$ and $j$ face $(\varphi,\mathbb{E})$-isomorphic learning
problems with $\varphi(s_{i})=s_{j}.$ Say $i$'s subhistory $y_{i,s_{i}}$
is\emph{ equivalent} to $j$'s subhistory $y_{j,s_{j}}$, written
as $y_{i,s_{i}}\sim y_{j,s_{j}}$, if they are equivalent period by
period according to the equivalence classes $\mathbb{E}_{s_{i},\varphi(s_{i})}.$
\end{defn}
Equivalence of $y_{i,s_{i}}$ and $y_{j,s_{j}}$ says $i$ has played
$s_{i}$ as many times as $j$ has played $s_{j},$ and that the sequence
of observations that $i$ encountered from experimenting with $s_{i}$
are the ``same'' as those that $j$ encountered from experimenting
with $s_{j}$.

The following histories for the critic and the diner of the restaurant
game are equivalent for the strategy $R$ (under the $\varphi,\mathbb{E}$
previously given). This equivalence arises because the subhistories
$y_{c,R}$ and $y_{d,R}$ contain the same sequences of the restaurant's
play (even though the two agents have different observations in terms
of how often the other patron goes to the restaurant).

\begin{table}[H]
\begin{centering}
\begin{tabular}{|c|c|c|c|c|c|c|}
\hline 
 & period & 1 & 2 & 3 & 4 & 5\tabularnewline
\hline 
\hline 
\multirow{2}{*}{$y_{c}$:} & own strategy & $R$ & $Z$ & $Z$ & $Z$ & $R$\tabularnewline
\cline{2-7} \cline{3-7} \cline{4-7} \cline{5-7} \cline{6-7} \cline{7-7} 
 & observation & $(L,Z)$ & $\varnothing$ & $\varnothing$ & $\varnothing$ & $(H,Z)$\tabularnewline
\hline 
\multirow{2}{*}{$y_{d}$:} & own strategy & $Z$ & $R$ & $Z$ & $R$ & \tabularnewline
\cline{2-7} \cline{3-7} \cline{4-7} \cline{5-7} \cline{6-7} \cline{7-7} 
 & observation & $\varnothing$ & $(L,R)$ & $\varnothing$ & $(H,Z)$ & \tabularnewline
\hline 
\end{tabular}
\par\end{centering}
\caption{The two histories $y_{c}$ (for the critic, with length 5) and $y_{d}$
(for the diner, with length 4) have equivalent subhistories for $R$.}
\end{table}

The agent decides which strategy to use in each period based on their
history so far. We assume that this \emph{learning policy }is a deterministic
map (which is w.l.o.g. for expected-utility maximizers), and denote
it $r_{i}:Y_{i}\to\mathbb{S}_{i}$.
\begin{defn}
A learning policy is an \emph{index policy} if there are\emph{ index
functions} $(\iota_{s_{i}})_{s_{i}\in\mathbb{S}_{i}}$ with each $\iota_{s_{i}}$
mapping subhistories of $s_{i}$ to real numbers, such that $r_{i}(y_{i})\in\underset{s_{i}\in\mathbb{S}_{i}}{\arg\max}\left\{ \iota_{s_{i}}(y_{i,s_{i}})\right\} $.
\end{defn}
If an agent uses an index policy, we can think of their behavior in
the following way. At each history, they compute an index for each
strategy $s_{i}\in\mathbb{S}_{i}$ based on the subhistory of those
periods where they chose $s_{i},$ and play a strategy with the highest
index.\footnote{To handle possible ties, we can introduce a strict order over each
agent's strategy set, and specify that if two strategies have the
same index the agent plays the one that is higher ranked.} The best-known example of an index policy is the Gittins index \citep{gittins1979bandit}.
Some heuristics for learning problems, such as weighted fictitious
play \citep{cheung1997individual}, are also index policies. The key
restriction in an index policy is that each strategy's index depends
only on the observations when that strategy was played. Note that
index policies are deterministic, unlike some heuristics such as Thompson
sampling \citep{thompson1933likelihood}.

Finally, we define a notion of the relative compatibility of index
policies $r_{i}$ and $r_{j}$ with various strategies.
\begin{defn}
Suppose $i$ and $j$ face $(\varphi,\mathbb{E})$-isomorphic learning
problems with $\varphi(s_{i}^{*})=s_{j}^{*}$. For two index policies
$r_{i}$ and $r_{j},$ $r_{i}$ is \emph{more index-compatible} with
$s_{i}^{*}$ than $r_{j}$ is with $s_{j}^{*}$ if for any histories
$y_{i},y_{j}$ and any strategy $s_{i}^{'}\in\mathbb{S}_{i},$ $s_{i}^{'}\ne s_{i}^{*}$
satisfying
\begin{itemize}
\item $y_{i,s_{i}^{*}}\sim y_{j,s_{j}^{*}}$ and $y_{i,s_{i}^{'}}\sim y_{j,\varphi(s_{i}^{'})}$
\item $s_{j}^{*}$ has weakly the highest index for $j$,
\end{itemize}
we have that $s_{i}^{'}$ does not have the weakly highest index for
$i$.
\end{defn}
The key result of this section, Proposition \ref{prop:index}, shows
that index compatibility is a sufficient condition for average agents
in the role of player $i$ to play $s_{i}^{*}$ more frequently than
those in the $j$-role play $s_{j}^{*}$. This result is not obvious
because the index-compatibility relation only applies when two agents
have equivalent histories, which typically does not hold during the
dynamic process of experimentation. To formalize the result, suppose
that each agent is randomly matched with opponents in other player
roles, and that the play of these opponents corresponds to a fixed
mixed strategy profile $\sigma$, as it would in a steady state. We
call $\sigma$ the \emph{social distribution}. It, together with the
agent's learning policy, generates a stochastic process $X_{i}^{t}$
describing $i$'s strategy in period $t$; denote its distribution
by $\mathbb{P}_{r_{i},\sigma}$.
\begin{defn}
Let $X_{i}^{t}$ be the $\mathbb{S}_{i}$-valued random variable representing
$i$'s play in period $t$. Player $i$'s\emph{ discounted lifetime
play under }the social distribution $\sigma$ and learning policy
$r_{i}$ is $\phi_{i}(\cdot;r_{i},\sigma):\mathbb{S}_{i}\to[0,1]$,
where for each $s_{i}\in\mathbb{S}_{i}$ 
\[
\phi_{i}(s_{i};r_{i},\sigma):=(1-\gamma)\sum_{t=1}^{\infty}\gamma^{t-1}\cdot\mathbb{P}_{r_{i},\sigma}\{X_{i}^{t}=s_{i}\}.
\]
\end{defn}
Each newborn player $i$ expects to play each of $s_{i}$ a share
$\phi_{i}(s_{i};r_{i},\sigma)$ of their lifetime.
\begin{prop}
\label{prop:index} Suppose $i$ and $j$ face $(\varphi,\mathbb{E})$-isomorphic
learning problems with $\varphi(s_{i}^{*})=s_{j}^{*}$, and $r_{i},r_{j}$
are index policies, with $r_{i}$ more index-compatible with $s_{i}^{*}$
than $r_{j}$ is with $s_{j}^{*}$. Then $\phi_{i}(s_{i}^{*};r_{i},\sigma_{-i})\ge\phi_{j}(s_{j}^{*};r_{j},\sigma_{-j})$
for any $0\le\gamma<1$ and $\sigma\in\times_{k}\Delta(\mathbb{S}_{k})$.
\end{prop}
Note that Proposition \ref{prop:index} holds whenever there exists
any $(\varphi,\mathbb{E})$ such that $i,j$ face $(\varphi,\mathbb{E})$-isomorphic
learning problems and $r_{i},r_{j}$ satisfy index-compatibility with
respect to $(\varphi,\mathbb{E}).$

The proof extends the coupling argument in the proof of \citet{fudenberg_he_2017}'s
Lemma 2, which only applies to the Gittins index in signaling games,
and also fills in a missing step (Lemma \ref{lem:response_path})
that the earlier proof implicitly assumed. To deal with the issue
that $i$ and $j$ learn from endogenous data that diverge as they
undertake different experiments, we couple the learning problems of
$i$ and $j$ using what we call\emph{ response paths} $\mathfrak{S}\in((\mathbb{S})^{N})^{\infty}$
where $N=\max_{i}|\mathbb{S}_{i}|.$ We can think of $\mathfrak{S}$
as a two-dimensional array of strategy profiles, $\mathfrak{S}=((\mathfrak{S}_{1,1},\mathfrak{S}_{1,2},...,\mathfrak{S}_{1,N}),(\mathfrak{S}_{2,1},\mathfrak{S}_{2,2},...,\mathfrak{S}_{2,N}),...)$,
where $\mathfrak{S}_{t,n_{i}}\in\mathbb{S}$ for every $t\ge1,$ $1\le n_{i}\le N.$
We may enumerate each player's strategy set $\mathbb{S}_{i}$ and
interchangeably refer to each strategy $s_{i}\in\mathbb{S}_{i}$ with
its assigned number $n_{s_{i}}\in\{1,...,N\}.$ For a given path and
learning policy $r_{i}$ for player $i,$ imagine running the policy
against the data-generating process where the $t$-th time $i$ plays
the $n_{i}$-th strategy in $\mathbb{S}_{i}$, $i$ is matched up
with opponents who play the strategies $\mathfrak{S}_{t,n_{i}}$.
Given a learning policy $r_{i}$, each $\mathfrak{S}$ induces a deterministic
infinite history of $i$'s strategies $y_{i}(\mathfrak{S},r_{i})\in(\mathbb{S}_{i})^{\infty}.$
(For $n_{i}>|\mathbb{S}_{i}|,$ the values of $(\mathfrak{S}_{t,n_{i}})_{t\ge1}$
do not matter for the induced history.) We show that under the hypothesis
that $r_{i}$ is more index-compatible with $s_{i}^{*}$ than $r_{j}$
is with $s_{j}^{*}$, the weighted lifetime frequency of $s_{i}^{*}$
in $y_{i}(\mathfrak{S},r_{i})$ is larger than the frequency of $s_{j}^{*}$
in $y_{j}(\mathfrak{S},r_{j})$ for every $\mathfrak{S}$, where play
in different periods of the infinite histories $y_{i}(\mathfrak{S},r_{i}),y_{j}(\mathfrak{S},r_{j})$
are weighted by the probabilities of surviving into these periods,
just as in the definition of discounted lifetime play.

Lemma \ref{lem:response_path} in the Appendix shows that when $i$
and $j$ face i.i.d. draws of opponents' play from a fixed learning
environment $\sigma,$ the discounted lifetime play are the same as
if they each faced a random response path $\mathfrak{S}$ drawn at
birth according to the (infinite) product measure over $((\mathbb{S})^{N})^{\infty}$
whose marginals correspond to $\sigma^{|\mathbb{S}_{i}|}$.

\section{\label{sec:Factorability-and-Isomorphic}Index Compatibility and
Player Compatibility in Factorable Games}

Section \ref{sec:Indexable-Learning-Rules} proves that whenever index-strategy
pairs $(r_{i},s_{i}^{*})$ and $(r_{j},s_{j}^{*})$ satisfy index
compatibility, and $i$ and $j$ face isomorphic learning problems,
index policy $r_{i}$ uses $s_{i}^{*}$ more often than $r_{j}$ uses
$s_{j}^{*}$ against any social distribution $\sigma.$ Index compatibility
is a joint restriction on the agents' learning policy and the game's
feedback structure $(\mathbb{O},\mathfrak{o})$, which gives the domain
that the learning policies are defined on. This section shows that
player compatibility implies index compatibility for rational behavior
and weighted fictitious play in a class of \emph{factorable} games.
Factorability applies to the examples discussed in Section \ref{sec:Examples-of-PCE}
for players ranked by compatibility.

\subsection{Factorability and Isomorphic Factoring}

In factorable games, agent $i$'s observation is just their utility:
$\mathfrak{o}_{i}(s_{i},s_{-i})=u_{i}(s_{i},s_{-i}).$ In general,
$i$'s payoff $u_{i}(s_{i},s_{-i})$ needs not reveal the actions
that others' strategies $s_{-i}$ pick at all $-i$ information sets
in the game tree. The definition of factorability puts restrictions
on the extensive-form game tree $\Gamma$ to discipline what $i$
can learn from own payoffs.

Factorability says that playing strategy $s_{i}\in\mathbb{S}_{i}$
against any strategy profile of $-i$ identifies all of the opponents'
actions that can be payoff-relevant for $s_{i}$, but does not reveal
any information about the payoff consequences of choosing any other
strategy $s'_{i}\ne s_{i}$ against the social distribution. From
$i$'s perspective, it is as if the game tree can be ``factored''
into disjoint parts based on the elements of $\mathbb{S}_{i},$ and
playing each $s_{i}\in\mathbb{S}_{i}$ lets $i$ learn how $s_{-i}$
play at all payoff-relevant $-i$ information sets in the $s_{i}$-part
of the game tree, but provides no information about $s_{-i}$ in any
other part of the tree. We now make this idea formal.

For an information set $h$ of $j$ with $j\ne i$, write $P_{h}$
for the partition on $\mathbb{S}_{-i}$ where two strategy profiles
$s_{-i},s'_{-i}$ are in the same element of the partition if they
prescribe the same play on $h.$ That is, the partition elements in
$P_{h}$ are $\{s_{-i}\in\mathbb{S}_{-i}:s_{-i}(h)=a_{h}\}$ for $a_{h}\in A_{h}.$
Thus partition $P_{h}$ contains perfect information about play on
$h$, but no other information.
\begin{defn}
\label{def:factorable}For each player $i$ and strategy $s_{i}\in\mathbb{S}_{i}$,
let $\Pi_{i}[s_{i}]$ be the coarsest partition of $\mathbb{S}_{-i}$
that makes $s_{-i}\mapsto u_{i}(s_{i},s_{-i})$ measurable. The game
$\Gamma$ is \emph{factorable for i} if:
\end{defn}
\begin{enumerate}
\item For each $s_{i}\in\mathbb{S}_{i}$ there exists a (possibly empty)
collection of $-i$'s information sets $F_{i}[s_{i}]\subseteq\mathcal{H}_{-i}$
so that $\Pi_{i}[s_{i}]=\bigvee_{h\in F_{i}[s_{i}]}P_{h}$. (The meet
across an empty collection is the coarsest possible partition on $\mathbb{S}_{-i}$,
i.e. no information).
\item For two strategies $s_{i}\ne s'_{i}$, $F_{i}[s_{i}]\cap F_{i}[s'_{i}]=\emptyset.$
\end{enumerate}
When $\Gamma$ is factorable for $i$, we refer to $F_{i}[s_{i}]$
as the \emph{$s_{i}$-relevant information sets}, a terminology we
now justify. In general, $i$'s payoff from playing $s_{i}$ can depend
on the profile of $-i$'s actions at all opponent information sets.
Condition (1) implies that only opponents' actions on $F_{i}[s_{i}]$
matter for $i$'s payoff after choosing $s_{i}$, and furthermore
this dependence is one-to-one. That is,

\[
u_{i}\left(s_{i},s_{-i}\right)=u_{i}\left(s_{i},s'_{-i}\right)\Leftrightarrow\left(\forall h\in F_{i}\left[s_{i}\right]\quad s_{-i}\left(h\right)=s'_{-i}\left(h\right)\right).
\]
 The one-to-one mapping from $s_{-i}$ to $i$'s payoff implies that
$i$'s learning cannot be blocked by another player: By choosing $s_{i}$,
$i$ can always use their own payoffs to identify actions on $F_{i}[s_{i}]$
regardless of what happens elsewhere in the game tree.\footnote{It is easy but expositionally costly to extend this to the case where
several actions on $A_{h}$ lead to the same payoff for $i$.} It also shows that if $\Gamma$ is factorable for $i,$ then $F_{i}[s_{i}]$
are uniquely defined for all $s_{i}$. Suppose there were two collections
$(F_{i}[s_{i}])_{s_{i}\in\mathbb{S}_{i}}$ and $(\tilde{F}_{i}[s_{i}])_{s_{i}\in\mathbb{S}_{i}}$
with $F_{i}[s_{i}]\backslash\tilde{F}_{i}[s_{i}]\ne\varnothing$ for
some $s_{i}\in\mathbb{S}_{i}$ that both satisfy Condition (1) of
Definition \ref{def:factorable}. Then there are two $-i$ profiles
$s_{-i},s'_{-i}$ that match on $\tilde{F}_{i}[s_{i}]$ but not on
$F_{i}[s_{i}]$. But then we get both $u_{i}(s_{i},s_{-i})=u_{i}(s_{i},s'_{-i})$
and $u_{i}(s_{i},s_{-i})\ne u_{i}(s_{i},s'_{-i})$, a contradiction.

Condition (2) implies that $i$ cannot extrapolate the payoff consequence
of a different strategy $s'_{i}\ne s_{i}$ through playing $s_{i}$
(provided $i$'s prior is independent about opponents' play on different
information sets). This is because there is no intersection between
the $s_{i}$-relevant information sets and the $s'_{i}$-relevant
ones. In particular this means that player $i$ cannot ``free ride''
on others' experiments and learn about the consequences of various
risky strategies while playing a safe one that is myopically optimal.

In short, Condition (1) ensures $i$ gets information about play in
the same part of the game tree every time she plays $s_{i}$ (instead
of learning about play in two different parts of the tree depending
on someone else's strategy), while Condition (2) guarantees that there
is no interaction between learning about different strategies.

If $F_{i}[s_{i}]$ is empty, then $s_{i}$ is a kind of ``opt out''
action for $i$. After choosing $s_{i}$, $i$ receives the same utility
from every reachable terminal node and gets no information about the
payoff consequences of any of their other strategies.

\subsubsection{Examples of Factorable Games}

We now illustrate factorability using the examples from Section \ref{sec:Examples-of-PCE}
and some other general classes of games.

\paragraph{\label{subsubsec:restaurant_with_partition}The Restaurant Game}

Consider the restaurant game from Example \ref{exa: restaurant}.
Since $x<-1$ and $y>0.5,$ we have $x\ne y$ and $x\ne y+0.5.$ By
choosing \textbf{R}, the customer's payoff perfectly reveals others'
play. By choosing \textbf{Z}, the customer always gets 0 payoff (these
nodes are colored in the diagram below) and so cannot infer anyone
else's play.

\begin{center}\includegraphics[scale=0.45]{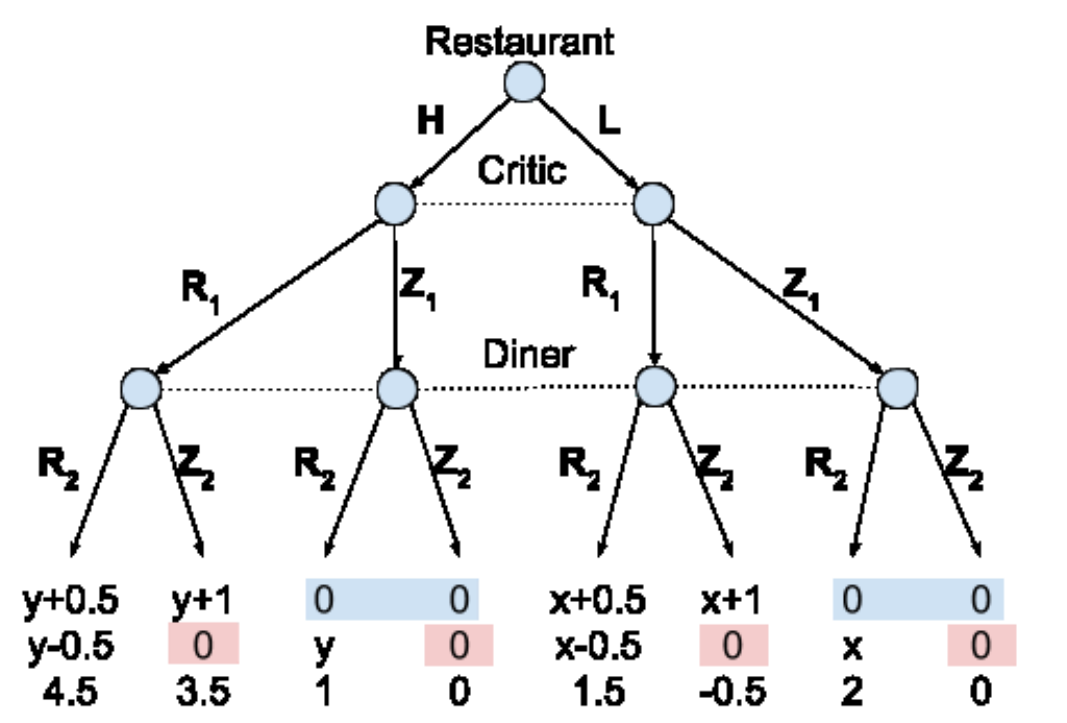}\end{center}

The restaurant game is factorable for the Critic and the Diner. Let
$F_{i}[\text{\textbf{R}}_{i}]$ consist of the two information sets
of $-i$ and let $F_{i}[\text{\textbf{Z}}_{i}]$ be the empty set
for each $i\in\{1,2\}$. It is easy to verify that the two conditions
of factorability are satisfied.

It is important for factorability that a customer who takes the ``outside
option'' of ordering pizza gets the same payoff regardless of the
restaurant's play, and does not observe the restaurant's quality choice
even if the other customer patronizes the restaurant. Factorability
rules out this sort of ``free information,'' so that when we analyze
the non-equilibrium learning problem we know that each agent can only
learn a strategy's payoff consequences by playing it themselves.

\paragraph{\label{subsubsec:link_with_partition}The Link-Formation Game}

Consider the link-formation game from Example \ref{exa:link4}. The
payoff for a player choosing \textbf{Inactive} is always 0, whereas
the payoff for a player choosing \textbf{Active} exactly identifies
the play of the two players on the opposite side. We can let $F_{i}[\mathbf{Active}_{i}]$
consist of the information sets of the other two agents on the other
side of $i$ and let $F_{i}[\mathbf{Inactive}_{i}]$ be empty. This
specification of the $s_{i}$-relevant information sets shows the
stage game is factorable for every player.

\paragraph{Binary Participation Games}

More generally, $\Gamma$ is factorable for $i$ whenever it is a
\emph{binary participation game} for $i.$
\begin{defn}
\label{def:binary_participation}$\Gamma$ is a \emph{binary participation
game} for $i$ if the following are satisfied.
\begin{enumerate}
\item $i$ has a unique information set with two actions, without loss labeled
\textbf{In} and \textbf{Out.}
\item All paths of play in $\Gamma$ pass through $i$'s information set.
\item All paths of play where $i$ plays \textbf{In} pass through the same
information sets.
\item Terminal vertices associated with $i$ playing \textbf{Out} all have
the same payoff for $i$.
\item Terminal vertices associated with $i$ playing\textbf{ In} all have
different payoffs for $i$.
\end{enumerate}
\end{defn}
Action \textbf{Out} is an outside option for $i$ that leads to a
constant payoff regardless of others' play. We are implicitly assuming
in part (5) of the definition that the game has generic payoffs for
$i$ after choosing \textbf{In}, in the sense that changing the action
at any one information set on the path of play will change $i$'s
payoff.

If $\Gamma$ is a binary participation game for $i,$ let $F_{i}[\mathbf{In}]$
be the collection of $-i$ information sets encountered in paths of
play where $i$ chooses \textbf{In}. Let $F_{i}[\mathbf{Out}]$ be
the empty set. We see that $\Gamma$ is factorable for $i.$ Clearly
$F_{i}[\mathbf{In}]\cap F_{i}[\mathbf{Out}]=\emptyset$, so Condition
(2) of factorability is satisfied. When $i$ chooses the strategy
\textbf{In}, the tree structure of $\Gamma$ implies different profiles
of play on $F_{i}[\mathbf{In}]$ must lead to different terminal nodes,
and the generic payoff condition means Condition (1) of factorability
is satisfied for strategy \textbf{In}. When $i$ plays \textbf{Out},
$i$ gets the same payoff regardless of the others' play, so Condition
(1) of factorability is satisfied for strategy \textbf{Out}.

The restaurant game is a binary participation game for the critic
and the diner, where ordering pizza is the outside option. The link-formation
game is a binary participation game for every player, where \textbf{Inactive}
is the outside option.

\paragraph{\label{subsubsec:signal_multiple_audiences}Signaling to Multiple
Audiences}

To give a different class of examples of factorable games, consider
a game of signaling to one or more audiences. To be precise, Nature
moves first and chooses a type for the sender, drawn according to
some known distribution over a finite set of types, $\Theta$. The
sender then chooses a signal $s\in S$, observed by all receivers
$r_{1},...,r_{n_{r}}$. Each receiver then simultaneously chooses
an action. The profile of receiver actions, together with the sender's
type and signal, determine payoffs for all players. Viewing different
types of senders as different players, this game is factorable for
all sender types, provided payoffs are generic. This factorability
arises because for each type $i$, $F_{i}[s]$ is the set of $n_{r}$
information sets for the receivers after seeing signal $s.$

\subsubsection{Examples of Non-Factorable Games}

The next result gives a necessary condition for factorability, which
we then use to provide examples of non-factorable games. Suppose $H$
is an information set of player $j\ne i.$ Player $i$'s payoff is
\emph{independent} of $h$ if $u_{i}(a_{h},a_{-h})=u_{i}(a'_{h},a_{-h})$
for all $a_{h},a'_{h},a_{-h}$, where $a_{h},a'_{h}$ are actions
on information set $h,$ and $a_{-h}$ is a profile of actions on
all other information sets in the game tree. In games where each player
moves at most once along any path of play, if $i$'s payoff is not
independent of the action taken at some information set $h$, then
$i$ can always put $h$ onto the path of play via a unilateral deviation
at one of their information sets.
\begin{prop}
\label{prop:one_step} Suppose each player moves at most once along
any path of play in $\Gamma$, the game is factorable for $i$, and
$i$'s payoff is not independent of $h^{*}$. For any strategy profile,
either $h^{*}$ is on the path of play, or $i$ has a deviation at
one of their information sets that puts $h^{*}$ onto the path of
play.
\end{prop}
This result follows from two lemmas.
\begin{lem}
\label{lem:one_step_lem1} For any game that is factorable for $i$
and any information set $h^{*}$ for player $j\ne i$ where $j$ has
at least two different actions, if $h^{*}\in F_{i}[s_{i}]$ for some
extensive-form strategy $s_{i}\in\mathbb{S}_{i},$ then $h^{*}$ is
always on the path of play when $i$ chooses $s_{i}$.
\end{lem}
\begin{lem}
\label{lem:one_step_lem2} For any game that is factorable for $i$
and any information set $h^{*}$ of player $j\ne i$, suppose $i$'s
payoff is not independent of $h^{*}$. Then 1) $j$ has at least two
different actions on $h^{*}$; and (2) there exists some extensive-form
strategy $s_{i}\in\mathbb{S}_{i}$ so that $h^{*}\in F_{i}[s_{i}]$.
\end{lem}
We can combine these two lemmas to prove the proposition. Consider
the centipede game for three players below.
\begin{center}
\includegraphics[scale=0.2]{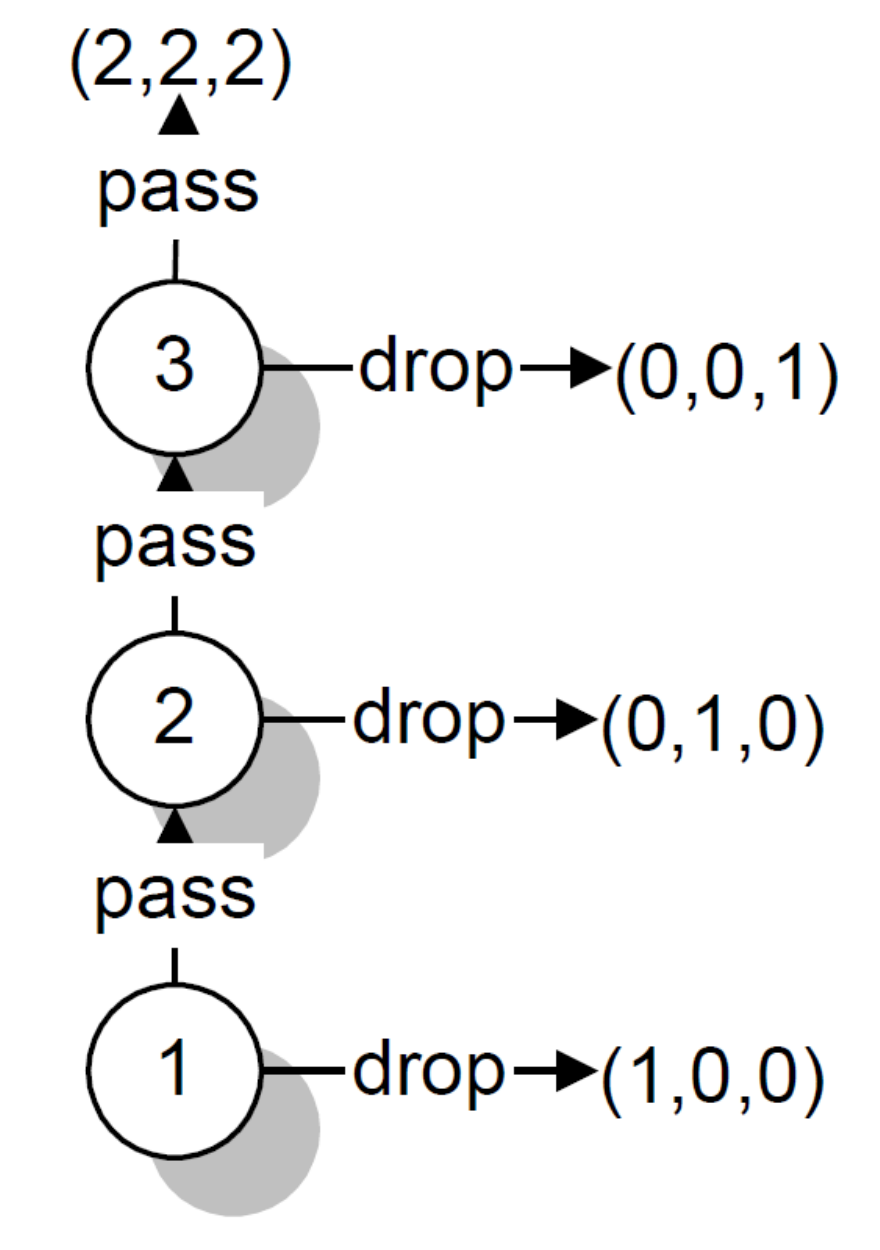}
\par\end{center}

Each player moves at most once on each path, and 1 and 2's payoffs
are not independent of the (unique) information set of player 3. But,
if both 1 and 2 choose ``drop'', then no one step deviation by either
1 or 2 can put the information set of 3 onto the path of play. Proposition
\ref{prop:one_step} thus implies the centipede game is not factorable
for either 1 or 2. Moreover, \citet{fudenberg_superstition_2006}
showed that in this game even very patient player 2's may not learn
to play a best response to player 3, so that the strategy profile
(drop, drop, pass) can persist even though it is not trembling-hand
perfect. Intuitively, if the player 1's only play pass as experiments,
then when the fraction of new players is very small, the player 2's
may not get to play often enough to make experimentation with pass
worthwhile.
\begin{center}
\includegraphics[scale=0.4]{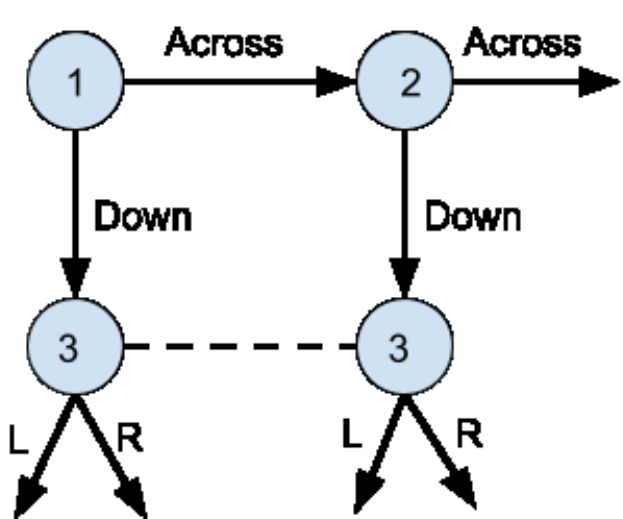}
\par\end{center}

As another example, the Selten's horse game displayed above is not
factorable for 1 or 2 if the payoffs are generic, even though the
conclusion of Proposition \ref{prop:one_step} is satisfied. The information
set of 3 must belong to both $F_{1}[\text{Down}]$ and $F_{1}[\text{Across}]$
because 3's play can affect 1's payoff even if 1 chooses Across, since
2 could choose Down. This violates the factorability requirement that
$F_{1}[\text{Down}]\cap F_{1}[\text{Across}]=\varnothing$. The same
argument shows the information set of 3 must belong to both $F_{2}[\text{Down}]$
and $F_{2}[\text{Across}]$, since when 1 chooses Down the play of
3 affects 2's payoff regardless of 2's play. So, again, $F_{2}[\text{Down}]\cap F_{2}[\text{Across}]=\varnothing$
is violated.

Condition (2) of factorability also rules out games where $i$ has
two strategies that give the same information, but one strategy always
has a worse payoff under all profiles of opponents' play. In this
case, we can think of the worse strategy as an informationally equivalent
but more costly experiment than the better strategy. Reasonable learning
policies (including rational learning) will not use such strategies,
but we do not capture this feature in the general definition of PCE
because our setup there only considers abstract strategy spaces $\mathbb{S}_{i}$
and not an extensive-form game tree.\footnote{It would be interesting to try to refine the definition of PCE to
capture this, perhaps using the ``signal function'' approach of
\citet{battigalli1997conjectural} and \citet{rubinstein1994rationalizable}.}

\subsubsection{Isomorphic Factoring}

In order to compare the learning behavior of agents $i$ and $j$,
it is not enough that the game is factorable for each of them. We
must also ensure that they face isomorphic learning problems in order
to apply Proposition \ref{prop:index}. To do this we define the notion
of\emph{ isomorphic factoring}.
\begin{defn}
When $\Gamma$ is factorable for both $i$ and $j$, the factoring
is \emph{isomorphic} for $i$ and $j$ if there exists a bijection
$\varphi:\mathbb{S}_{i}\to\mathbb{S}_{j}$ such that $F_{i}[s_{i}]\cap\mathcal{H}_{-ij}=F_{j}[\varphi(s_{i})]\cap\mathcal{H}_{-ij}$
for every $s_{i}\in\mathbb{S}_{i}$.
\end{defn}
This says the $s_{i}$-relevant information sets (for $i$) are the
same as the $\varphi(s_{i})$-relevant information sets (for $j$),
insofar as the actions of $-ij$ are concerned. For example, the restaurant
game is isomorphically factorable for the critic and the diner (under
the isomorphism $\varphi(\text{\textbf{R1})=\textbf{R2}}$, $\varphi(\text{\textbf{Z1})=\text{\textbf{Z2}}}$)
because $F_{1}[\text{\textbf{R1}}]\cap\mathcal{H}_{3}=F_{2}[\text{\textbf{R2}}]\cap\mathcal{H}_{3}=$
the singleton set containing the unique information set of the restaurant.
As another example, all signaling games (with possibly many receivers
as in Section \ref{subsubsec:signal_multiple_audiences}) are isomorphically
factorable for the different types of the sender. Similarly, the link-formation
game is isomorphically factorable for pairs (N1, N2), and (S1, S2),
but note that it is not isomorphically factorable for (N1, S1).

Factorability and isomorphic factoring let us construct $(\varphi,\mathbb{E})$
so that $i$ and $j$ face $(\varphi,\mathbb{E})$-isomorphic learning
problems. The equivalence classes $\mathbb{E}_{s_{i},\varphi(s_{i})}$
for $s_{i}\in\mathbb{S}_{i}$ are such that $(s_{i},u_{i}(s_{i},\tilde{s}_{-i}))\sim(\varphi(s_{i}),u_{j}(\varphi(s_{i}),\hat{s}_{-j})$
if and only if $\tilde{s}_{-i}|_{F_{i}[s_{i}]\cap\mathcal{H}_{-ij}}=\hat{s}_{-j}|_{F_{j}[\varphi(s_{i})]\cap\mathcal{H}_{-ij}}.$

\subsection{Rational Learning in Factorable Games}

We first consider rational agents who maximize expected discounted
payoffs. This learning rule requires two additional elements: a Bayesian
prior belief over others' play and a discount factor. We assume that
each agent $i$ starts with a \emph{regular independent prior:}
\begin{defn}
Agent $i$ has a \emph{regular independent prior} if their beliefs
$g_{i}$ on $\times_{h\in\mathcal{H}_{-i}}\Delta(A_{h})$ can be written
as the product of full-support marginal densities on $\Delta(A_{h})$
across different $h\in\mathcal{H}_{-i}$, so that $g_{i}((\alpha_{h})_{h\in\mathcal{H}_{-i}})=\prod_{h\in\mathcal{H}_{-i}}g_{i}^{h}(\alpha_{h})$
with $g_{i}^{h}(\alpha_{h})>0$ for all $\alpha_{h}\in\Delta^{\circ}(A_{h})$.
\end{defn}
Thus, the agent thinks the social distribution assigns an unknown
mixed action at each $-i$ information set, and thinks actions at
different $-i$ information sets are generated independently from
these underlying mixed actions, whether the information sets belong
to the same player or to different players. Furthermore, the agent
holds independent beliefs about the mixed actions at different information
sets.\footnote{We assume that agents do not know Nature's mixed actions, which must
be learned just as the play of other players. If agents know Nature's
move, then a regular independent prior would be a density $g_{i}$
on $\times_{h\in\mathcal{H}_{\mathbb{I}\backslash\{i\}}}\Delta(A_{h})$,
so that $g_{i}((\alpha_{h})_{\mathcal{H}_{\mathbb{I}\backslash\{i\}}})=\prod_{h\in\mathcal{H}_{\mathbb{I}\backslash\{i\}}}g_{i}^{h}(\alpha_{h})$
with $g_{i}^{h}(\alpha_{h})>0$ for all $\alpha_{h}\in\Delta^{\circ}(A_{h})$.
As \citet{fudenberg1993learning} point out, an agent who believes
two opponents are randomizing independently may nevertheless have
subjective correlation in their uncertainty about the randomizing
probabilities of these opponents. Here we study the natural special
case where the agents' prior beliefs about the opponents are independent,
i.e., a product measure.} The agent updates $g_{i}$ by applying Bayes rule to their history
$y_{i}$. If the stage game is a signaling game, for example, this
independence assumption means that the senders only update their beliefs
about the receiver's response to a given signal based on the responses
received to that signal, and that the senders' beliefs about this
response do not depend on the responses they have observed to other
signals.

In addition to the survival chance $0\le\gamma<1$ between periods,
the agent further discounts future payoffs according to their patience
$0\le\delta<1,$ so their overall effective discount factor is $0\le\delta\gamma<1$.
All agents have the same patience level.

Given a belief about the distribution of play at each opponent information
set, we can calculate the Gittins index of each strategy $s_{i}\in\mathbb{S}_{i}$.
Let $\nu_{s_{i}}\in\times_{h\in F_{i}[s_{i}]}\Delta(\Delta(A_{h}))$
be a belief over opponents' mixed actions at the information sets
in $F_{i}[s_{i}]$. The Gittins index of $s_{i}$ under belief $\nu_{s_{i}}$
is given by the maximum value of the following auxiliary optimization
problem:

\[
\sup_{\tau\ge1}\dfrac{\mathbb{E}_{\nu_{s_{i}}}\left\{ \sum_{t=1}^{\tau}(\delta\gamma)^{t-1}\cdot u_{i}(s_{i},(a_{h}(t))_{h\in F_{i}[s_{i}]})\right\} }{\mathbb{E}_{\nu_{s_{i}}}\left\{ \sum_{t=1}^{\tau}(\delta\gamma)^{t-1}\right\} },
\]
where the supremum is taken over all positive-valued stopping times
$\tau\ge1$. Here $(a_{h}(t)){}_{h\in F_{i}[s_{i}]}$ means the profile
of actions that $-i$ plays on $F_{i}[s_{i}]$ the $t$-th time that
$i$ uses $s_{i}$ --- by assumption about factorable games, only
these actions and not actions elsewhere in the game tree determine
$i$'s payoff from playing $s_{i}$, and $i$ can always infer these
actions from their own payoffs. The distribution over the infinite
sequence of profiles $(a_{h}(t))_{t=1}^{\infty}$ is given by $i$'s
belief $\nu_{s_{i}}$, that is, there is some fixed mixed action in
$\times_{h\in F_{i}[s_{i}]}\Delta(A_{h})$ that generates profiles
$(a_{h}(t))$ i.i.d. across periods $t.$ The event $\{\tau=T\}$
for $T\ge1$ corresponds to using $s_{i}$ for $T$ periods, observing
the first $T$ elements $(a_{h}(t))_{t=1}^{T}$, then stopping.

A learning policy that chooses a strategy $s_{i}$ with the highest
Gittins index after each history $y_{i}$ solves the rational agent's
dynamic optimization problem. We denote any such policy as $\text{OPT}_{i},$
suppressing its dependence on $\delta$ and $g_{i}.$

\subsection{Weighted Fictitious Play in Factorable Games}

Next we consider the weighted fictitious play heuristic, a generalization
of \citet{brown1951iterative}'s fictitious play.\footnote{This heuristic was first estimated on lab data by \citet{cheung1997individual}.
It was generalized by \citet{camerer_experience-weighted_1999} and
later analyzed by \citet*{benaim2009learning}.} Agent $i$ keeps track of \emph{counts} for actions at the opponent
information sets in the game tree, 
\[
\{N_{h}^{a_{h}}\in\mathbb{R}_{++}:h\in\mathcal{H}_{-i},a_{h}\in A_{h}\}.
\]
 The $N_{h}^{a_{h}}$ values of a newcomer agent start at some\emph{
initial counts}, $N_{h}^{a_{h}}(0)>0,$ and the counts update as $i$
learns.

After history $y_{i}$ of $i$ where $s_{i}$ has been used $T\ge0$
times, $i$'s subhistory for $s_{i}$ can be viewed as $y_{i,s_{i}}=(s_{i},s_{-i}^{(t)}(h)_{h\in F_{-i}[s_{i}]})_{t=1}^{T}$
where $s_{-i}^{(t)}(h)_{h\in F_{-i}[s_{i}]}$ is the observed $-i$'s
play on $F_{i}[s_{i}]$ the $t$-th time that $s_{i}$ was used. (This
is because there is a one-to-one relationship between $s_{-i}$'s
play on $F_{i}[s_{i}]$ and $u_{i}(s_{i},s_{-i})$.) The updated count
on $(h,a_{h})$ for $h\in F_{i}[s_{i}]$ and $a_{h}\in A_{h}$ is

\[
N_{h}^{a_{h}}(y_{i})=\sum_{t=1}^{T}\boldsymbol{1}(s_{-i}^{(t)}(h)=a_{h})\cdot\rho^{(T-t)}+\rho^{T}N_{h}^{a_{h}}(0)
\]
 for some $\rho\in(0,1].$ That is, $i$ calculates a weighted sum
for the total number of times that $-i$ have played $a_{h}$ in the
history $y_{i}$, where past observations on $F_{i}[s_{i}]$ are discounted
at a rate $\rho$ between successive uses of the strategy $s_{i}$.
All agents share the same weight factor $\rho.$

Following history $y_{i},$ $i$ assigns an index to $s_{i}$ equal
to its expected payoff when opponents play the mixed action $\alpha_{h}(a_{h};y_{i})=\frac{N_{h}^{a_{h}}(y_{i})}{\sum_{a_{h}^{'}\in A_{h}}N_{h}^{a_{h}^{'}}(y_{i})}$
on information sets $h\in F_{i}[s_{i}].$ Write $\text{WFP}_{i}$
for a learning policy that chooses a strategy with the highest weighted
fictitious play index after every history (suppressing its dependence
on $\rho$ and the initial counts $\{N_{h}^{a_{h}}(0)\})$.

When $\rho=1$, the counts are updated according to the unweighted
fictitious play, and the limit of $\rho\to0$ corresponds to myopically
best replying to the observed play when each strategy was most recently
used. The special case of the Gittins index where the prior $g_{i}$
marginalized to each $\Delta(A_{h})$ is a Dirichlet distribution
and $\delta=0$ is equivalent to the special case of unweighted fictitious
play (i.e., $\rho=1)$ with some initial counts that depend on the
Dirichlet priors' parameters. In general $\text{OPT}_{i}$ differs
from $\text{WFP}_{i}$ outside of these corner cases.

\subsection{Player-Compatibility Implies Index-Compatibility of $\text{OPT}$
and $\text{WFP}$ under Isomorphic Factoring}

The main result of this paper, Theorem \ref{thm:crossplayer_tremble_foundation},
shows that if $s_{i}^{*}\succsim s_{j}^{*}$ in a game isomorphically
factorable for $i$ and $j$ with $\varphi(s_{i}^{*})=s_{j}^{*}$,
then $i$ uses $s_{i}^{*}$ more frequently than $j$ uses $s_{i}^{*}$
both under rational experimentation and under weighted fictitious
play. This comparison holds under the hypothesis that $i$ and $j$
start their learning processes with the same ``initial conditions.''
For $\text{OPT}$, this means $i$'s prior $g_{i}$ marginalized to
the $s_{i}$-relevant $-ij$ information sets equals to $j$'s prior
$g_{j}$ marginalized to the $\varphi(s_{i})$-relevant $-ij$ information
sets for every $s_{i}\in\mathbb{S}_{i}$. For $\text{WFP}$, this
means $i$ and $j$ start with the same initial counts about $-ij$'s
actions.
\begin{thm}
\label{thm:crossplayer_tremble_foundation} Suppose $s_{i}^{*}\succsim s_{j}^{*}$
and the game is isomorphically factorable for $i$ and $j$ with $\varphi(s_{i}^{*})=s_{j}^{*}$.
For any common survival chance $0\le\gamma<1$ and any social distribution
$\sigma$, we have $\phi_{i}(s_{i}^{*};r_{i},\sigma_{-i})\ge\phi_{j}(s_{j}^{*};r_{j},\sigma_{-j})$
under either of the following conditions:
\end{thm}
\begin{itemize}
\item $r_{i}=\text{OPT}_{i},$ $r_{j}=\text{OPT}_{j}$, and priors $g_{i},g_{j}$
are regular and equivalent\footnote{The theorem easily generalizes to the case where $i$ starts with
one of $L\ge2$ possible priors $g_{i}^{(1)},...,g_{i}^{(L)}$ with
probabilities $p_{1},...,p_{L}$ and $j$ starts with priors $g_{j}^{(1)},...,g_{j}^{(L)}$
with the same probabilities, and each $g_{i}^{(l)},g_{j}^{(l)}$ is
a pair of equivalent regular priors for $1\le l\le L$.}: that is, they satisfy $g_{i}|_{\Delta(A_{h}):h\in F_{i}[s_{i}]\cap\mathcal{H}_{-ij}}=g_{j}|_{\Delta(A_{h}):h\in F_{j}[\varphi(s_{i})]\cap\mathcal{H}_{-ij}}$
for every $s_{i}\in\mathbb{S}_{i}$.
\item $r_{i}=\text{WFP}_{i},$ $r_{j}=\text{WFP}_{j},$ and $i$ and $j$
have the same initial counts $N_{h}^{a_{h}}(0)$ for every $s_{i}\in\mathbb{S}_{i},$
$h\in F_{i}[s_{i}]\cap\mathcal{H}_{-ij}$, and $a_{h}\in A_{h}.$
\end{itemize}
The proof works by showing that if $s_{i}^{*}\succsim s_{j}^{*},$
then both rational learning and weighted fictitious play are $(\varphi,\mathbb{E})$-index
compatible. This then lets us apply Proposition \ref{prop:index}'s
general conclusion about index compatible learning policies.

\section{\label{sec:Replication-Invariance-of-PCE}Replication Invariance
of PCE}

Fix a\emph{ base game} where each $i$ has a finite strategy set $\mathbb{S}_{i}$
and utility function $u_{i}:\mathbb{S}\to\mathbb{R}$.
\begin{defn}
An\emph{ extended game with duplicates} is a game with the same players
as the base game, where $\bar{\mathbb{S}}_{i}$ the set of strategies
of $i$ is a finite subset of $\mathbb{S}_{i}\times\mathbb{N}$ such
that for all $(s_{i},n_{i})\in\bar{\mathbb{S}}_{i}$ and $(s_{j},n_{j})_{j\ne i}\in\bar{\mathbb{S}}_{-i}$,
the payoff in the new game is $\bar{u}_{i}((s_{i},n_{i}),(s_{j},n_{j})_{j\ne i})=u_{i}(s_{i},s_{j})$.
For every $s_{i}\in\mathbb{S}_{i}$, there exists some $n_{i}\in\mathbb{N}$
so that $(s_{i},n_{i})\in\bar{\mathbb{S}}_{i}$.
\end{defn}
The interpretation is that $i$ has some copies of every strategy
they had in the base game, and could have different numbers of copies
of different strategies, where duplicate copies of the same strategy
have the same payoff consequences. For the learning foundation we
will now assume that the extended game is factorable, which implies
that each duplicate provides the same information. As an example,
suppose that in the Restaurant Game the critic can arrive at the restaurant
by taking the red bus or the blue bus, and the color of the bus is
not observed by other players, does not change anyone's payoffs, and
does not change what the diner observes. We can then replace $R_{c}$
with two actions $R_{c}^{red},R_{c}^{blue}$ at every node in the
diner's information set, letting $R_{c}^{red}$ and $R_{c}^{blue}$
both have the same payoff consequences as $R_{c}$ in the original
game. This modified game is an extended game with duplicates for the
original game.

Subsection \ref{subsec:PCE_duplicates} defines player-compatible
trembles and PCE in extended games with duplicates. Using the compatibility
relation $\succsim$ from the base game, a tremble profile in the
extended game with duplicates is player compatible if the \emph{sum}
of tremble probabilities assigned to all copies of $s_{i}^{*}$ exceeds
the sum assigned to all copies of $s_{j}^{*},$ whenever $s_{i}^{*}\succsim s_{j}^{*}$.
PCE is then defined using this restriction on trembles. We show that
the set of PCE in the base game coincides with the set of PCE in the
extended game with duplicates.

This definition of player-compatible trembles in extended games with
duplicates fits with our interpretation of trembles as experimentation
frequencies and an analysis of how learning dynamics in the extended
game compare with those in the base game. The idea is that if all
copies of a strategy $s_{i}$ give $i$ the same information about
others' play, then $i$ should be exactly indifferent between all
such copies after all histories in the learning process. Holding fixed
initial beliefs and the social distribution, $i$'s weighted lifetime
average play of $s_{i}$ in the base game should then equal the sum
of their weighted lifetime average plays of all copies of $s_{i}$
in the extended game with duplicates. Thus, any comparisons that hold
between the ``tremble'' probabilities of $i$ onto $s_{i}^{*}$
and $j$ onto $s_{j}^{*}$ in the base game must also hold between
the ``tremble'' probabilities of $i$ onto the copies of $s_{i}^{*}$
and $j$ onto the copies of $s_{j}^{*}$ in the extended game. We
formalize this intuition in binary participation games in Subsection
\ref{subsec:foundation_duplicates} for rational learning and weighted
fictitious play.

\subsection{\label{subsec:PCE_duplicates}PCE in Extended Games with Duplicates}

A \emph{tremble profile} \emph{of the extended game} $\bar{\boldsymbol{\epsilon}}$
assigns a positive number $\bar{\boldsymbol{\epsilon}}(s_{i},n_{i})>0$
to every player $i$ and every pure strategy $(s_{i},n_{i})\in\bar{\mathbb{S}}_{i}$.
We define \emph{$\bar{\boldsymbol{\epsilon}}$-strategies} of $i$
and $\bar{\boldsymbol{\epsilon}}$\emph{-equilibrium} of the extended
game in the usual way, relative to the strategy sets $\bar{\mathbb{S}}_{i}$.
\begin{defn}
Tremble profile $\boldsymbol{\bar{\epsilon}}$ is \emph{player-compatible
in the extended game} if $\sum_{n_{i}}\boldsymbol{\bar{\epsilon}}(s_{i}^{*},n_{i})\ge\sum_{n_{j}}\bar{\boldsymbol{\epsilon}}(s_{j}^{*},n_{j})$
for all $i,j,s_{i}^{*},s_{j}^{*}$ such that $s_{i}^{*}\succsim s_{j}^{*}$.
An $\bar{\boldsymbol{\epsilon}}$-equilibrium where $\bar{\boldsymbol{\epsilon}}$
is player-compatible is called a \emph{player-compatible $\bar{\boldsymbol{\epsilon}}$-equilibrium
}(or\emph{ $\bar{\boldsymbol{\epsilon}}$-PCE}).
\end{defn}
We now relate $\bar{\boldsymbol{\epsilon}}$-equilibria in the extended
game to $\boldsymbol{\epsilon}$-equilibria in the base game. Recall
the following constrained optimality condition that applies to both
the extended game and the base game:
\begin{fact}
A feasible mixed strategy of $i$ is \textbf{not} a constrained best
response to a $-i$ profile if and only if it assigns more than the
required weight to a non-optimal response.
\end{fact}
We associate with a strategy profile $\bar{\sigma}\in\times_{i\in\mathbb{I}}\Delta(\bar{\mathbb{S}}_{i})$
in the extended game a \emph{consolidated strategy profile} $\mathfrak{\mathscr{C}}(\bar{\sigma})\in\times_{i\in\mathbb{I}}\Delta(\mathbb{S}_{i})$
in the base game, given by adding up the probabilities assigned to
all copies of each base-game strategy. More precisely, $\mathfrak{\mathscr{C}}(\bar{\sigma})_{i}(s_{i}):=\sum_{n_{i}}\bar{\sigma}_{i}(s_{i},n_{i})$.
Similarly, $\mathfrak{\mathscr{C}}(\boldsymbol{\bar{\epsilon}})$
is the consolidated tremble profile, given by $\mathfrak{\mathscr{C}}(\boldsymbol{\bar{\epsilon}})(s_{i}):=\sum_{n_{i}}\bar{\boldsymbol{\epsilon}}(s_{i},n_{i})$.

Conversely, given a strategy profile $\sigma\in\times_{i\in\mathbb{I}}\Delta(\mathbb{S}_{i})$
in the base game, the extended strategy profile $\mathfrak{\mathscr{E}}(\sigma)\in\times_{i\in\mathbb{I}}\Delta(\bar{\mathbb{S}}_{i})$
is defined by $\mathfrak{\mathscr{E}}(\sigma)_{i}(s_{i},n_{i}):=\sigma_{i}(s_{i})/N(s_{i})$
for each $i,(s_{i},n_{i})\in\bar{\mathbb{S}}_{i}$, where $N(s_{i})$
is the number of copies of $s_{i}$ that $\bar{\mathbb{S}}_{i}$ contains.
Similarly, $\mathfrak{\mathscr{E}}(\boldsymbol{\epsilon})$ is the
extended tremble profile, given by $\mathfrak{\mathscr{E}}(\boldsymbol{\epsilon})(s_{i},n_{i}):=\boldsymbol{\epsilon}(s_{i})/N(s_{i})$.
\begin{lem}
\label{lem:epsilon_eqm_translate} If $\bar{\sigma}$ is an $\bar{\boldsymbol{\epsilon}}$-equilibrium
in the extended game, then $\mathfrak{\mathscr{C}}(\bar{\sigma})$
is an $\mathfrak{\mathscr{C}}(\boldsymbol{\bar{\epsilon}})$-equilibrium
in the base game. If $\sigma$ is an $\boldsymbol{\epsilon}$-equilibrium
in the base game, then $\mathfrak{\mathscr{E}}(\sigma)$ is an $\mathfrak{\mathscr{E}}(\boldsymbol{\epsilon})$-equilibrium
in the extended game.
\end{lem}
The proof of results in this section can be found in the Online Appendix.

PCE is defined as usual in the extended game.
\begin{defn}
A strategy profile $\bar{\sigma}^{*}$ is a\emph{ player-compatible
equilibrium (PCE) in the extended game} if there exists a sequence
of player-compatible tremble profiles $\boldsymbol{\bar{\epsilon}}^{(t)}\to\boldsymbol{0}$
and an associated sequence of strategy profiles $\bar{\sigma}^{(t)},$
where each $\bar{\sigma}^{(t)}$ is an $\boldsymbol{\bar{\epsilon}}^{(t)}$-PCE,
such that $\bar{\sigma}^{(t)}\to\bar{\sigma}^{*}$.
\end{defn}
These PCE correspond exactly to PCE of the base game.
\begin{prop}
\label{prop:PCE_translate}If $\bar{\sigma}^{*}$ is a PCE in the
extended game, then $\mathfrak{\mathscr{C}}(\bar{\sigma}^{*})$ is
a PCE in the base game. If $\sigma^{*}$ is a PCE in the base game,
then $\mathfrak{\mathscr{E}}(\sigma^{*})$ is a PCE in the extended
game.
\end{prop}
In fact, starting from a PCE $\sigma^{*}$ of the base game, we can
construct more PCE of the extended game than $\mathfrak{\mathscr{E}}(\sigma^{*})$
by shifting around the probabilities assigned to different copies
of the same base-game strategy, but all these profiles essentially
correspond to the same outcome.

\subsection{\label{subsec:foundation_duplicates}Learning and Trembles in Binary
Participation Games with Duplicates}

We give the simplest illustration of how learning dynamics in extended
games with duplicates relate to those in the base game, using binary
participation games. These results can also be developed for other
factorable games, but at the cost of more complicated notation.

Consider a binary participation game for $i$ (Definition \ref{def:binary_participation})
as the base game and create an extended game with duplicates by adding
an extra copy of the \textbf{In} strategy for $i$ to the game tree,
called \textbf{In-d}. We show that when $r_{i}$ is an optimal learning
policy for $i$ or the weighted fictitious play heuristic, the discounted
lifetime play $\phi_{i}(\text{\textbf{In}};r_{i},\sigma_{-i})$ for
the base game is equal to the sum $\phi_{i}(\text{\textbf{In}};r_{i},\sigma_{-i})+\phi_{i}(\text{\textbf{In-d}};r_{i},\sigma_{-i})$
in the new game, for the same social distribution $\sigma.$

We modify the original game tree $\Gamma$ and information sets $\mathcal{H}$
to arrive at a new game tree $\bar{\Gamma}$ with information sets
$\bar{\mathcal{H}}.$ The basic idea is that \textbf{In-d} gives the
same payoffs and information to $i,$ and $-i$ cannot tell which
one $i$ chose.

By the definition of a binary participation game for $i$, let $h_{i}$
be $i$'s unique information set in $\mathcal{H}.$ Enumerate the
vertices in $h_{i}$ as $h_{i}=\{v_{1},...,v_{n}\}$. Playing $\text{\textbf{In}}$
at vertex $v_{k}$ in the original tree leads to some subtree $\Gamma^{(k)}\subseteq\Gamma$.
Start with $\bar{\Gamma}=\Gamma$ and add a new move, \textbf{In-d},
to every $v_{k}\in h_{i}.$ Append a new subtree $\hat{\Gamma}^{(k)}$
to $\bar{\Gamma}$ for every $v_{k}\in h_{i}$, such that $\hat{\Gamma}^{(k)}$
is a copy of $\Gamma^{(k)}$ (including payoffs at terminal vertices)
and playing \textbf{In-d} at $v_{k}$ leads to $\hat{\Gamma}^{(k)}$.
Now we give a procedure to construct the information sets $\bar{\mathcal{H}}$
to capture the idea that \textbf{In }and\textbf{ In-d }are indistinguishable
to others. Start with $\bar{\mathcal{H}}=\mathcal{H}$ and let $V^{(k)}$
be the set of vertices in $\Gamma^{(k)}.$ For every $1\le k\le n$
and $v\in V^{(k)},$ find the information set $h\in\bar{\mathcal{H}}$
with $v\in h,$ then put $h:=h\cup\{\tilde{v}\}$, where $\tilde{v}$
is the copy of $v$ in $\hat{\Gamma}^{(k)}.$ That is, each vertex
reachable after $i$ chooses\textbf{ In-d} is indistinguishable to
others from its ``twin'' reachable when $i$ chooses \textbf{In}.

As discussed before, the Restaurant Game is a binary participation
game for the critic and the diner, with going to the restaurant as
\textbf{In} and ordering pizza as \textbf{Out}. We illustrate adding
a duplicate copy of $R_{c}$ for the critic to the game, labeled $R_{c}-d$.
The critic's unique information set contains two vertices, and the
new game tree adds two new subtrees to the original game, highlighted
in red.
\begin{center}
\includegraphics[scale=0.7]{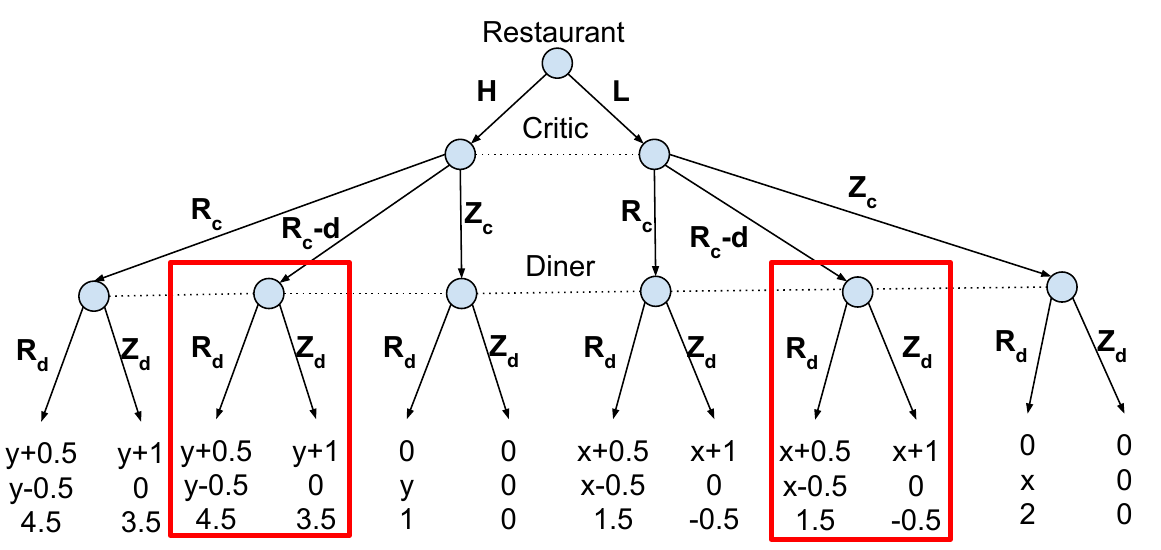}
\par\end{center}

The set of histories in the learning framework for $i$ with the extended
stage game is $\tilde{Y}_{i}=\cup_{t\ge0}(\{\text{\textbf{In}},\text{\textbf{In-d}},\text{\textbf{Out}}\}\times\mathbb{R})^{t}.$
We now define a notion of equivalence between a \emph{stochastic}
learning policy in the extended game $\tilde{r}_{i}:\tilde{Y}_{i}\to\Delta(\{\text{\textbf{In}},\text{\textbf{In-d}},\text{\textbf{Out}}\}\}$
and a (deterministic) learning policy in the original game, $r_{i}:Y_{i}\to\{\text{\textbf{In}},\text{\textbf{Out}}\}$.
Basically, $\tilde{r}_{i}$ behaves just like $r_{i}$ except it can
randomize between $\text{\textbf{In}}$ and $\text{\textbf{In-d}}$.
\begin{defn}
Let $\zeta:\tilde{Y}_{i}\to Y_{i}$ be such that for $\tilde{y}_{i}\in\tilde{Y}_{i}$,
$\zeta(\tilde{y}_{i})\in Y_{i}$ replaces every instance of instance
of $\textbf{In-d}$ with $\textbf{In}$. Learning policies $\tilde{r}_{i}:\tilde{Y}_{i}\to\Delta(\{\text{\textbf{In}},\text{\textbf{In-d}},\text{\textbf{Out}}\}\}$
and $r_{i}:Y_{i}\to\{\text{\textbf{In}},\text{\textbf{Out}}\}$ are
\emph{equivalent up to duplicates }if for every $\tilde{y}_{i}\in\tilde{Y}_{i},$
if $r_{i}(\zeta(\tilde{y}_{i}))=\text{\textbf{Out}}$, then also $\tilde{r}_{i}(\tilde{y}_{i})(\text{\textbf{Out}})=1.$
If $r_{i}(\zeta(\tilde{y}_{i}))=\text{\textbf{In}}$, then $\tilde{r}_{i}(\tilde{y}_{i})(\text{\textbf{In}})+\tilde{r}_{i}(\tilde{y}_{i})(\text{\textbf{In-d}})=1$.
\end{defn}
The main result of this section shows that rational learning and weighted
fictitious play lead to learning policies that are equivalent up to
duplicates in the base game and the extended game. Furthermore, any
pair of such equivalent policies in the two settings lead to the same
lifetime discounted frequencies of playing \textbf{In} for the original
game as playing \textbf{In} and \textbf{In-d} for the extended game
against the same social distributions of $-i.$

Technically, extensive-form strategies in $(\Gamma,\mathcal{H})$
and $(\bar{\Gamma},\bar{\mathcal{H}})$ are defined over two different
domains. To make sense of $i$ facing the ``same'' social distribution
of $-i$'s play in the two settings, let $\psi:\bar{\mathcal{H}}\to\mathcal{H}$
be the natural isomorphism between the two collections of information
sets. Each information set $\tilde{h}$ in the modified game is either
equal to an information set $h\in\mathcal{H},$ or it is an old information
set with some extra vertices added, that is there is some (unique)
$h$ with $\tilde{h}\supsetneq h$. Let $\psi(\tilde{h}):=h.$ Two
strategy profiles $\sigma,\tilde{\sigma}$ for $(\Gamma,\mathcal{H})$
and $(\bar{\Gamma},\bar{\mathcal{H}})$ are \emph{$-i$ equivalent}
if $\tilde{\sigma}(\tilde{h})=\sigma(\psi(\tilde{h}))$ for all $\tilde{h}\in\tilde{\mathcal{H}}_{-i}.$
\begin{prop}
\label{prop:learning_with_duplicates} Suppose stochastic learning
policy $\tilde{r}_{i}$ in the extended game is equivalent up to duplicates
with the learning policy $r_{i}$ in the base game.
\begin{itemize}
\item For a fixed patience parameter $0\le\delta<1$ and regular prior $g_{i}$
over others' play,\footnote{The prior is over $\times_{h\in\mathcal{H}_{-i}}\Delta(A_{h})$ in
the original stage game and over $\times_{\tilde{h}\in\tilde{\mathcal{H}}_{-i}}\Delta(A_{\tilde{h}})$
in the extended game, but we identify $\Delta(A_{\tilde{h}})$ with
$\Delta(A_{\psi(\tilde{h})})$ for each $\tilde{h}\in\tilde{H}_{-i}.$
The same identification applies for the initial counts in the original
and extended games.} $r_{i}$ is $\text{OPT}_{i}$ if and only if $\tilde{r}_{i}$ is
an optimal learning policy with the extended game.
\item For a fixed decay parameter $0\le\rho<1$ and initial counts $N_{h}^{a_{h}}(0),$
$r_{i}$ is $\text{WFP}_{i}$ if and only if after every $\tilde{y}_{i}\in\tilde{Y}_{i},$
$\tilde{r}_{i}(\tilde{y}_{i})$ is supported on strategies that maximize
payoffs under the weighted fictitious play conjecture of $-i$'s play.
\item For $-i$ equivalent social distributions $\sigma,\tilde{\sigma}$
for the base game and extended games, $\phi_{i}(\text{\textbf{In}};r_{i},\sigma_{-i})=\phi_{i}(\text{\textbf{In}};\tilde{r}_{i},\tilde{\sigma}_{-i})+\phi_{i}(\text{\textbf{In-d}};\tilde{r}_{i},\tilde{\sigma}_{-i})$.
\end{itemize}
\end{prop}

\section{Concluding Discussion}

PCE makes two key contributions. First, it generates new and sensible
restrictions on equilibrium play by imposing cross-player restrictions
on the relative probabilities that different players assign to certain
strategies --- namely, those strategy pairs $s_{i},s_{j}$ ranked
by the player-compatibility relation $s_{i}\succsim s_{j}$. As we
have shown through examples, these cross-player restrictions distinguish
PCE from other refinement concepts and allows us to make comparative
statics predictions in some games where other equilibrium refinements
do not.

Second, PCE shows how restricted ``trembles'' can capture some of
the implications of non-equilibrium learning. PCE's cross-player restrictions
arise endogenously for a general class of index learning policies,
which includes both the standard model of Bayesian agents maximizing
their expected discounted lifetime utility, and computationally tractable
heuristics like weighted fictitious play. We conjecture that the result
that $i$ is more likely to experiment with $s_{i}$ than $j$ is
with $s_{j}$ when $s_{i}\succsim s_{j}$ applies in other natural
models of learning or dynamic adjustment, such as those considered
by \citet{francetich2020choosing1,francetich2020choosing2}, and that
it may be possible to provide foundations for PCE in other and perhaps
larger classes of games.

The strength of the PCE refinement depends on the completeness of
the compatibility order $\succsim$, since $\boldsymbol{\epsilon}$-PCE
imposes restrictions on $i$ and $j$'s play only when the relation
$s_{i}\succsim s_{j}$ holds. Our player compatibility definition
supposes that player $i$ thinks all mixed strategies of other players
are possible, as it considers the set of all totally mixed correlated
strategies $\sigma_{-i}\in\Delta^{\circ}(\mathbb{S}_{-i}).$ If the
players have some prior knowledge about their opponents' utility functions,
player $i$ might deduce \emph{a priori} that the other players will
only play strategies in some subset of $\Delta^{\circ}(\mathbb{S}_{-i})$.
As we show in \citet{FudenbergHe2017TCE}, in signaling games imposing
this kind of prior knowledge leads to a more complete version of the
compatibility order. It may similarly lead to a more refined version
of PCE.

PCE is defined for general strategic forms. We have only provided
learning foundations for player-compatible trembles in factorable
games. Moreover, even in factorable games, PCE imposes some extra
restrictions that we do not microfound, but we view this as a first
step in connecting together tremble-based refinement concepts with
learning-in-games. As we have shown through the link-formation game
and other examples, PCE is a convenient reduced form that generates
novel comparative statics predictions in various applications without
needing the analyst to solve the dynamic learning problem anew in
each of them.

\bibliographystyle{ecta}
\bibliography{Gittins_eqm}

\begin{thebibliography}{41}
\newcommand{\enquote}[1]{``#1''}
\expandafter\ifx\csname natexlab\endcsname\relax\def\natexlab#1{#1}\fi

\bibitem[\protect\citeauthoryear{Battigalli, Cerreia-Vioglio, Maccheroni, and
  Marinacci}{Battigalli et~al.}{2016}]{battigalli2016analysis}
\textsc{Battigalli, P., S.~Cerreia-Vioglio, F.~Maccheroni, and M.~Marinacci}
  (2016): \enquote{Analysis of information feedback and selfconfirming
  equilibrium,} \emph{Journal of Mathematical Economics}, 66, 40--51.

\bibitem[\protect\citeauthoryear{Battigalli, Francetich, Lanzani, and
  Marinacci}{Battigalli et~al.}{2019}]{battigalli2017LRBiases}
\textsc{Battigalli, P., A.~Francetich, G.~Lanzani, and M.~Marinacci} (2019):
  \enquote{Learning and self-confirming long-run biases,} \emph{Journal of
  Economic Theory}, 183, 740--785.

\bibitem[\protect\citeauthoryear{Battigalli and Guaitoli}{Battigalli and
  Guaitoli}{1997}]{battigalli1997conjectural}
\textsc{Battigalli, P. and D.~Guaitoli} (1997): \enquote{Conjectural equilibria
  and rationalizability in a game with incomplete information,} in
  \emph{Decisions, Games and Markets}, Springer, 97--124.

\bibitem[\protect\citeauthoryear{Bena{\"\i}m, Hofbauer, and
  Hopkins}{Bena{\"\i}m et~al.}{2009}]{benaim2009learning}
\textsc{Bena{\"\i}m, M., J.~Hofbauer, and E.~Hopkins} (2009): \enquote{Learning
  in games with unstable equilibria,} \emph{Journal of Economic Theory}, 144,
  1694--1709.

\bibitem[\protect\citeauthoryear{Bolton and Harris}{Bolton and
  Harris}{1999}]{bolton1999strategic}
\textsc{Bolton, P. and C.~Harris} (1999): \enquote{Strategic experimentation,}
  \emph{Econometrica}, 67, 349--374.

\bibitem[\protect\citeauthoryear{Brown}{Brown}{1951}]{brown1951iterative}
\textsc{Brown, G.~W.} (1951): \enquote{Iterative solution of games by
  fictitious play,} \emph{Activity Analysis of Production and Allocation}, 13,
  374--376.

\bibitem[\protect\citeauthoryear{Camerer and Ho}{Camerer and
  Ho}{1999}]{camerer_experience-weighted_1999}
\textsc{Camerer, C. and T.-H. Ho} (1999): \enquote{Experience-weighted
  Attraction Learning in Normal Form Games,} \emph{Econometrica}, 67, 827--874.

\bibitem[\protect\citeauthoryear{Cheung and Friedman}{Cheung and
  Friedman}{1997}]{cheung1997individual}
\textsc{Cheung, Y.-W. and D.~Friedman} (1997): \enquote{Individual learning in
  normal form games: Some laboratory results,} \emph{Games and Economic
  Behavior}, 19, 46--76.

\bibitem[\protect\citeauthoryear{Cho and Kreps}{Cho and
  Kreps}{1987}]{cho_signaling_1987}
\textsc{Cho, I.-K. and D.~M. Kreps} (1987): \enquote{Signaling {Games} and
  {Stable} {Equilibria},} \emph{Quarterly Journal of Economics}, 102, 179--221.

\bibitem[\protect\citeauthoryear{Doval}{Doval}{2018}]{doval2018whether}
\textsc{Doval, L.} (2018): \enquote{Whether or not to open Pandora's box,}
  \emph{Journal of Economic Theory}, 175, 127--158.

\bibitem[\protect\citeauthoryear{Esponda and Pouzo}{Esponda and
  Pouzo}{2016}]{esponda_berknash_2016}
\textsc{Esponda, I. and D.~Pouzo} (2016): \enquote{Berk-{Nash} {Equilibrium}:
  {A} {Framework} for {Modeling} {Agents} {With} {Misspecified} {Models},}
  \emph{Econometrica}, 84, 1093--1130.

\bibitem[\protect\citeauthoryear{Francetich and Kreps}{Francetich and
  Kreps}{2020{\natexlab{a}}}]{francetich2020choosing1}
\textsc{Francetich, A. and D.~Kreps} (2020{\natexlab{a}}): \enquote{Choosing a
  good toolkit, I: Prior-free heuristics,} \emph{Journal of Economic Dynamics
  and Control}, 111, 103813.

\bibitem[\protect\citeauthoryear{Francetich and Kreps}{Francetich and
  Kreps}{2020{\natexlab{b}}}]{francetich2020choosing2}
---\hspace{-.1pt}---\hspace{-.1pt}--- (2020{\natexlab{b}}): \enquote{Choosing a
  good toolkit, II: Bayes-rule based heuristics,} \emph{Journal of Economic
  Dynamics and Control}, 111, 103814.

\bibitem[\protect\citeauthoryear{Frick and Ishii}{Frick and
  Ishii}{2015}]{frick2015innovation}
\textsc{Frick, M. and Y.~Ishii} (2015): \enquote{Innovation adoption by
  forward-looking social learners,} \emph{Working Paper}.

\bibitem[\protect\citeauthoryear{Fryer and Harms}{Fryer and
  Harms}{2017}]{fryer2017two}
\textsc{Fryer, R. and P.~Harms} (2017): \enquote{Two-armed restless bandits
  with imperfect information: Stochastic control and indexability,}
  \emph{Mathematics of Operations Research}, 43, 399--427.

\bibitem[\protect\citeauthoryear{Fudenberg and He}{Fudenberg and
  He}{2018}]{fudenberg_he_2017}
\textsc{Fudenberg, D. and K.~He} (2018): \enquote{Learning and Type
  Compatibility in Signaling Games,} \emph{Econometrica}, 86, 1215--1255.

\bibitem[\protect\citeauthoryear{Fudenberg and He}{Fudenberg and
  He}{2020}]{FudenbergHe2017TCE}
---\hspace{-.1pt}---\hspace{-.1pt}--- (2020): \enquote{Payoff Information and
  Learning in Signaling Games,} \emph{Games and Economic Behavior}, 120,
  96--120.

\bibitem[\protect\citeauthoryear{Fudenberg and Kreps}{Fudenberg and
  Kreps}{1993}]{fudenberg1993learning}
\textsc{Fudenberg, D. and D.~M. Kreps} (1993): \enquote{Learning {Mixed}
  {Equilibria},} \emph{Games and Economic Behavior}, 5, 320--367.

\bibitem[\protect\citeauthoryear{Fudenberg and Kreps}{Fudenberg and
  Kreps}{1994}]{fudenbergKreps1994learning}
---\hspace{-.1pt}---\hspace{-.1pt}--- (1994): \enquote{Learning in
  Extensive-Form Games, II: Experimentation and Nash Equilibrium,}
  \emph{Working Paper}.

\bibitem[\protect\citeauthoryear{Fudenberg and Kreps}{Fudenberg and
  Kreps}{1995}]{fudenberg_learning_1995}
---\hspace{-.1pt}---\hspace{-.1pt}--- (1995): \enquote{Learning in
  {Extensive-Form} {Games} {I}. {Self}-{Confirming} {Equilibria},} \emph{Games
  and Economic Behavior}, 8, 20--55.

\bibitem[\protect\citeauthoryear{Fudenberg, Lanzani, and Strack}{Fudenberg
  et~al.}{2020}]{fudenberg2020limits}
\textsc{Fudenberg, D., G.~Lanzani, and P.~Strack} (2020): \enquote{Limits
  Points of Endogenous Misspecified Learning,} \emph{Working Paper}.

\bibitem[\protect\citeauthoryear{Fudenberg and Levine}{Fudenberg and
  Levine}{1993}]{fudenberg_steady_1993}
\textsc{Fudenberg, D. and D.~K. Levine} (1993): \enquote{Steady {State}
  {Learning} and {Nash} {Equilibrium},} \emph{Econometrica}, 61, 547--573.

\bibitem[\protect\citeauthoryear{Fudenberg and Levine}{Fudenberg and
  Levine}{2006}]{fudenberg_superstition_2006}
---\hspace{-.1pt}---\hspace{-.1pt}--- (2006): \enquote{Superstition and
  {Rational} {Learning},} \emph{American Economic Review}, 96, 630--651.

\bibitem[\protect\citeauthoryear{Gittins}{Gittins}{1979}]{gittins1979bandit}
\textsc{Gittins, J.~C.} (1979): \enquote{Bandit {Processes} and {Dynamic}
  {Allocation} {Indices},} \emph{Journal of the Royal Statistical Society.
  Series B (Methodological)}, 148--177.

\bibitem[\protect\citeauthoryear{Halac, Kartik, and Liu}{Halac
  et~al.}{2016}]{halac2016optimal}
\textsc{Halac, M., N.~Kartik, and Q.~Liu} (2016): \enquote{Optimal contracts
  for experimentation,} \emph{Review of Economic Studies}, 83, 1040--1091.

\bibitem[\protect\citeauthoryear{Heidhues, Rady, and Strack}{Heidhues
  et~al.}{2015}]{heidhues2015strategic}
\textsc{Heidhues, P., S.~Rady, and P.~Strack} (2015): \enquote{Strategic
  experimentation with private payoffs,} \emph{Journal of Economic Theory},
  159, 531--551.

\bibitem[\protect\citeauthoryear{H{\"o}rner and Skrzypacz}{H{\"o}rner and
  Skrzypacz}{2016}]{horner2016learning}
\textsc{H{\"o}rner, J. and A.~Skrzypacz} (2016): \enquote{Learning,
  experimentation and information design,} in \emph{Advances in Economics and
  Econometrics: Eleventh World Congress}, ed. by B.~Honore, A.~Pakes,
  M.~Piazzesi, and L.~Samuelson, Cambridge University Press, chap.~2, 63--97.

\bibitem[\protect\citeauthoryear{Jackson and Wolinsky}{Jackson and
  Wolinsky}{1996}]{jackson1996strategic}
\textsc{Jackson, M.~O. and A.~Wolinsky} (1996): \enquote{A strategic model of
  social and economic networks,} \emph{Journal of Economic Theory}, 71, 44--74.

\bibitem[\protect\citeauthoryear{Keller, Rady, and Cripps}{Keller
  et~al.}{2005}]{keller2005strategic}
\textsc{Keller, G., S.~Rady, and M.~Cripps} (2005): \enquote{Strategic
  experimentation with exponential bandits,} \emph{Econometrica}, 73, 39--68.

\bibitem[\protect\citeauthoryear{Klein and Rady}{Klein and
  Rady}{2011}]{klein2011negatively}
\textsc{Klein, N. and S.~Rady} (2011): \enquote{Negatively correlated bandits,}
  \emph{Review of Economic Studies}, 78, 693--732.

\bibitem[\protect\citeauthoryear{Kohlberg and Mertens}{Kohlberg and
  Mertens}{1986}]{kohlberg_strategic_1986}
\textsc{Kohlberg, E. and J.-F. Mertens} (1986): \enquote{On the {Strategic}
  {Stability} of {Equilibria},} \emph{Econometrica}, 54, 1003--1037.

\bibitem[\protect\citeauthoryear{Lehrer}{Lehrer}{2012}]{lehrer2012partially}
\textsc{Lehrer, E.} (2012): \enquote{Partially specified probabilities:
  decisions and games,} \emph{American Economic Journal: Microeconomics}, 4,
  70--100.

\bibitem[\protect\citeauthoryear{Milgrom and Mollner}{Milgrom and
  Mollner}{2019}]{milgrom_mollner_EP}
\textsc{Milgrom, P. and J.~Mollner} (2019): \enquote{Extended Proper
  Equilibrium,} \emph{Working Paper}.

\bibitem[\protect\citeauthoryear{Monderer and Shapley}{Monderer and
  Shapley}{1996}]{monderer1996potential}
\textsc{Monderer, D. and L.~S. Shapley} (1996): \enquote{Potential games,}
  \emph{Games and Economic Behavior}, 14, 124--143.

\bibitem[\protect\citeauthoryear{Myerson}{Myerson}{1978}]{myerson1978refinements}
\textsc{Myerson, R.~B.} (1978): \enquote{Refinements of the Nash equilibrium
  concept,} \emph{International Journal of Game Theory}, 7, 73--80.

\bibitem[\protect\citeauthoryear{Pearce}{Pearce}{1984}]{pearce1984rationalizable}
\textsc{Pearce, D.~G.} (1984): \enquote{Rationalizable strategic behavior and
  the problem of perfection,} \emph{Econometrica}, 52, 1029--1050.

\bibitem[\protect\citeauthoryear{Rubinstein and Wolinsky}{Rubinstein and
  Wolinsky}{1994}]{rubinstein1994rationalizable}
\textsc{Rubinstein, A. and A.~Wolinsky} (1994): \enquote{Rationalizable
  conjectural equilibrium: between Nash and rationalizability,} \emph{Games and
  Economic Behavior}, 6, 299--311.

\bibitem[\protect\citeauthoryear{Selten}{Selten}{1975}]{selten1975reexamination}
\textsc{Selten, R.} (1975): \enquote{Reexamination of the perfectness concept
  for equilibrium points in extensive games,} \emph{International Journal of
  Game Theory}, 4, 25--55.

\bibitem[\protect\citeauthoryear{Strulovici}{Strulovici}{2010}]{strulovici2010learning}
\textsc{Strulovici, B.} (2010): \enquote{Learning while voting: Determinants of
  collective experimentation,} \emph{Econometrica}, 78, 933--971.

\bibitem[\protect\citeauthoryear{Thompson}{Thompson}{1933}]{thompson1933likelihood}
\textsc{Thompson, W.~R.} (1933): \enquote{On the likelihood that one unknown
  probability exceeds another in view of the evidence of two samples,}
  \emph{Biometrika}, 25, 285--294.

\bibitem[\protect\citeauthoryear{Van~Damme}{Van~Damme}{1987}]{van1987stability}
\textsc{Van~Damme, E.} (1987): \emph{Stability and Perfection of Nash
  Equilibria}, Springer-Verlag.

\end{thebibliography}

\newpage
\begin{center}
{\LARGE{}Appendix}{\LARGE\par}
\par\end{center}

\section{\label{sec:Omitted-Proofs-From-Main}Proofs of Results Stated in
the Main Text}

\subsection{Proof of Proposition \ref{prop:PCE_compatible}}

We first state an auxiliary lemma.
\begin{lem}
\emph{\label{lem:epsilon_PCE_compatible}If $\sigma^{\circ}$ is an
$\boldsymbol{\epsilon}$-PCE and $s_{i}^{*}\succsim s_{j}^{*}$, then
\[
\sigma_{i}^{\circ}(s_{i}^{*})\ge\min\left[\sigma_{j}^{\circ}(s_{j}^{*}),1-\sum_{s_{i}^{'}\ne s_{i}^{*}}\boldsymbol{\epsilon}(s_{i}^{'})\right].
\]
}
\end{lem}
\begin{proof}
Suppose $\boldsymbol{\epsilon}$ is player-compatible and let $\boldsymbol{\epsilon}$-equilibrium
$\sigma^{\circ}$ be given. For\emph{ $s_{i}^{*}\succsim s_{j}^{*}$},
suppose $\sigma_{j}^{\circ}(s_{j}^{*})=\boldsymbol{\epsilon}(s_{j}^{*})$.
Then $\sigma_{i}^{\circ}(s_{i}^{*})\ge\boldsymbol{\epsilon}(s_{i}^{*})\ge\boldsymbol{\epsilon}(s_{j}^{*})=\sigma_{j}^{\circ}(s_{j}^{*})$,
where the second inequality comes from $\boldsymbol{\epsilon}$ being
player-compatible. On the other hand, suppose $\sigma_{j}^{\circ}(s_{j}^{*})>\boldsymbol{\epsilon}(s_{j}^{*})$.
Since $\sigma^{\circ}$ is an $\boldsymbol{\epsilon}$-equilibrium,
the fact that $j$ puts more than the minimum required weight on $s_{j}^{*}$
implies $s_{j}^{*}$ is at least a weak best response for $j$ against
$\sigma^{\circ}$, with $\sigma^{\circ}$ totally mixed due to the
trembles.The definition of $s_{i}^{*}\succsim s_{j}^{*}$ then implies
that $s_{i}^{*}$ must be a strict best response for $i$ against
$\sigma^{\circ}$ as well. In the $\boldsymbol{\epsilon}$-equilibrium,
$i$ must assign as much weight to $s_{i}^{*}$ as possible, so that
$\sigma_{i}^{\circ}(s_{i}^{*})=1-\sum_{s_{i}^{'}\ne s_{i}^{*}}\boldsymbol{\epsilon}(s_{i}^{'})$.
Combining these two cases establishes the desired result.
\end{proof}
\textbf{Proposition} \textbf{\ref{prop:PCE_compatible}}: \emph{For
any PCE $\sigma^{*}$, player $k$, and strategy $\bar{s}_{k}$ such
that $\sigma_{k}^{*}(\bar{s}_{k})>0,$ there exists a sequence of
totally mixed strategy profiles $\sigma_{-k}^{(t)}\to\sigma_{-k}^{*}$
such that}

\emph{(i) for every pair $i,j\ne k$ with $s_{i}^{*}\succsim s_{j}^{*}$,
\[
\liminf_{t\to\infty}\frac{\sigma_{i}^{(t)}(s_{i}^{*})}{\sigma_{j}^{(t)}(s_{j}^{*})}\ge1;
\]
and (ii) $\bar{s}_{k}$ is a best response for $k$ against every
$\sigma_{-k}^{(t)}$ .}
\begin{proof}
By Lemma \ref{lem:epsilon_PCE_compatible}, for every $\boldsymbol{\epsilon}^{(t)}$-PCE
we get 
\begin{align*}
\frac{\sigma_{i}^{(t)}(s_{i}^{*})}{\sigma_{j}^{(t)}(s_{j}^{*})} & \ge\min\left[\frac{\sigma_{j}^{(t)}(s_{j}^{*})}{\sigma_{j}^{(t)}(s_{j}^{*})},\frac{1-\sum_{s_{i}^{'}\ne s_{i}^{*}}\boldsymbol{\epsilon}^{(t)}(s_{i}^{'}|i)}{\sigma_{j}^{(t)}(s_{j}^{*})}\right]\\
 & =\min\left[1,\frac{1-\sum_{s_{i}^{'}\ne s_{i}^{*}}\boldsymbol{\epsilon}^{(t)}(s_{i}^{'}|i)}{\sigma_{j}^{(t)}(s_{j}^{*})}\right]\ge1-\sum_{s_{i}^{'}\ne s_{i}^{*}}\boldsymbol{\epsilon}^{(t)}(s_{i}^{'}|i).
\end{align*}

This says 
\[
\inf_{t\ge T}\frac{\sigma_{i}^{(t)}(s_{i}^{*})}{\sigma_{j}^{(t)}(s_{j}^{*})}\ge1-\sup_{t\ge T}\sum_{s_{i}^{'}\ne s_{i}^{*}}\boldsymbol{\epsilon}^{(t)}(s_{i}^{'}|i).
\]
For any sequence of trembles such that $\boldsymbol{\epsilon}^{(t)}\to\boldsymbol{0},$
$\lim_{T\to\infty}\sup_{t\ge T}\sum_{s_{i}^{'}\ne s_{i}^{*}}\boldsymbol{\epsilon}^{(t)}(s_{i}^{'}|i)=0,$
so
\[
\liminf_{t\to\infty}\frac{\sigma_{i}^{(t)}(s_{i}^{*})}{\sigma_{j}^{(t)}(s_{j}^{*})}=\lim_{T\to\infty}\left\{ \inf_{t\ge T}\frac{\sigma_{i}^{(t)}(s_{i}^{*})}{\sigma_{j}^{(t)}(s_{j}^{*})}\right\} \ge1.
\]

This shows that if we fix a PCE $\sigma^{*}$ and consider a sequence
of player-compatible trembles $\boldsymbol{\epsilon}^{(t)}$ and $\boldsymbol{\epsilon}^{(t)}$-PCE
$\sigma^{(t)}\rightarrow\sigma^{*}$, then each $\sigma_{-k}^{(t)}$
satisfies $\text{\ensuremath{\liminf}}_{t\to\infty}\sigma_{i}^{(t)}(s_{i}^{*})/\sigma_{j}^{(t)}(s_{j}^{*})\ge1$
whenever $i,j\ne k$ and $s_{i}^{*}\succsim s_{j}^{*}$. Furthermore,
from $\sigma_{k}^{*}(\bar{s}_{k})>0$ and $\sigma_{k}^{(t)}\to\sigma_{k}^{*}$,
we know there is some $T_{1}\in\mathbb{N}$ so that $\sigma_{k}^{(t)}(\bar{s}_{k})>\sigma_{k}^{*}(\bar{s}_{k})/2$
for all $t\ge T_{1}$. We may also find $T_{2}\in\mathbb{N}$ so that
$\boldsymbol{\epsilon}^{(t)}(\bar{s}_{k}|k)<\sigma_{k}^{*}(\bar{s}_{k})/2$
for all $t\ge T_{2}$, since $\boldsymbol{\epsilon}^{(t)}\to\boldsymbol{0}$.
So when $t\ge\max(T_{1},T_{2})$, $\sigma_{k}^{(t)}$ places strictly
more than the required weight on $\bar{s}_{k},$ so $\bar{s}_{k}$
is at least a weak best response for $k$ against $\sigma_{-k}^{(t)}.$
Now the subsequence of opponent play $(\sigma_{-k}^{(t)})_{t\ge\max(T_{1},T_{2})}$
satisfies the requirement of this proposition.
\end{proof}

\subsection{Proof of Theorem \ref{thm:PCE_existence}}

\textbf{Theorem \ref{thm:PCE_existence}}: \emph{PCE exists in every
finite strategic-form game.}
\begin{proof}
Consider a sequence of tremble profiles with the same lower bound
on the probability of each strategy, that is $\boldsymbol{\epsilon}^{(t)}(s_{i}|i)=\epsilon^{(t)}$
for all $i$ and $s_{i}$, and with $\epsilon^{(t)}$ decreasing monotonically
to 0 in $t$. Each of these tremble profiles is player-compatible
(regardless of the compatibility structure $\succsim$) and there
is some finite $T$ large enough that $t\ge T$ implies an $\boldsymbol{\epsilon}^{(t)}$-equilibrium
exists, and some subsequence of these $\boldsymbol{\epsilon}^{(t)}$-equilibria
converges since the space of strategy profiles is compact. By definition
these $\boldsymbol{\epsilon}^{(t)}$-equilibria are also $\boldsymbol{\epsilon}^{(t)}$-PCE,
which establishes existence of PCE.
\end{proof}

\subsection{\label{subsec:Proof-of-CC} Proof of Proposition \ref{prop:compare_CC}}

\textbf{Proposition \ref{prop:compare_CC}}: \emph{In a signaling
game, every PCE $\sigma^{*}$ is a Nash equilibrium satisfying the
compatibility criterion of \citet{fudenberg_he_2017}.}
\begin{proof}
Since every PCE is a trembling-hand perfect equilibrium and since
this latter solution concept refines Nash, $\sigma^{*}$ is a Nash
equilibrium. To show that it satisfies the compatibility criterion,
we need to show that $\sigma_{2}^{*}$ assigns probability 0 to plans
in $A^{S}$ that, for some $s\in S$, do not best respond to an ``admissible''
belief at signal $s$ under profile $\sigma^{*}.$ For any plan assigned
positive probability under $\sigma_{2}^{*}$, by Proposition \ref{prop:PCE_compatible}
we may find a sequence of totally mixed signal profiles $\sigma_{1}^{(t)}$
of the sender, so that whenever $s_{\theta}\succsim s_{\theta^{'}}$
we have $\liminf_{t\to\infty}\sigma_{1}^{(t)}(s\mid\theta)/\sigma_{1}^{(t)}(s\mid\theta^{'})\ge1.$
Write $q^{(t)}(\cdot\mid s)$ as the Bayesian posterior belief about
the sender's type after signal $s$ under $\sigma_{1}^{(t)}$, which
is well defined because each $\sigma_{1}^{(t)}$ is totally mixed.
Whenever $s_{\theta}\succsim s_{\theta^{'}}$, this sequence of posterior
beliefs satisfies $\liminf_{t\to\infty}q^{(t)}(\theta\mid s)/q^{(t)}(\theta^{'}\mid s)\ge\lambda(\theta)/\lambda(\theta^{'})$,
so if the receiver's plan best responds to every element in the sequence,
it also best responds to an accumulation point $(q^{\infty}(\cdot\mid s))_{s\in S}$
with $q^{\infty}(\theta\mid s)/q^{\infty}(\theta^{'}\mid s)\ge\lambda(\theta)/\lambda(\theta^{'})$
whenever $s_{\theta}\succsim s_{\theta^{'}}$. Since the player compatibility
definition used in this paper is slightly easier to satisfy than the
type compatibility definition that the set $P(s^{'},\sigma^{*})$
is based on, the plan best responds to $P(s^{'},\sigma^{*})$ after
every signal $s^{'}$.
\end{proof}

\subsection{Proof of Proposition \ref{prop:index}}

Let $N=\max_{i}|\mathbb{S}_{i}|.$ We first show that $i'$s discounted
lifetime play is the same whether $i$ plays against strategy profiles
drawn i.i.d. in different periods from the social distribution $\sigma_{-i}$,
or against a response path drawn from a certain distribution $\eta$
at the start of $i$'s life. The next lemma constructs this $\eta$
from $\sigma$, which is the same for all agents, and does not depend
on their (possibly stochastic) learning policies.
\begin{lem}
\label{lem:response_path} In a factorable game, for each $\sigma\in\times_{k}\Delta(\mathbb{S}_{k}),$
there is a distribution $\eta$ over response paths, so that for any
player $i$, any possibly random policy $r_{i}:Y_{i}\to\Delta(\mathbb{S}_{i})$,
and any strategy $s_{i}\in\mathbb{S}_{i}$, we have 
\[
\phi_{i}(s_{i};r_{i},\sigma)=(1-\gamma)\mathbb{E}_{\mathfrak{S}\sim\eta}\left[\sum_{t=1}^{\infty}\gamma^{t-1}\cdot(y_{i}^{t}(\mathfrak{S},r_{i})=s_{i})\right],
\]
where $y_{i}^{t}(\mathfrak{S},r_{i})$ refers to the $t$-th period
history in $y_{i}(\mathfrak{S},r_{i})$.
\end{lem}
\begin{proof}
In fact, we will prove a stronger statement: we will show there is
such a distribution that induces the same distribution over period-$t$
histories for every $i,$ every learning policy $r_{i},$ and every
$t.$

Think of each response path $\mathfrak{S}$ as a two-dimensional array,
$\mathfrak{S}=(\mathfrak{S}_{t,n})_{t\in\mathbb{N},1\le n\le N}$.
For non-negative integers $(m_{n})_{n=1}^{N}$, each finite two-dimensional
array of strategy profiles $((s_{t,n})_{t=1}^{m_{n}})_{n=1}^{N}$
with each $s_{t,n}\in\mathbb{S}$ defines a ``cylinder set'' of
response paths with the form: 
\[
\{\mathfrak{S}:\mathfrak{S}_{t,n}=s_{t,n}\ \text{for each \ensuremath{1\le n\le N,1\le t\le m_{n}\}}}.
\]
That is, the cylinder set consists of those response paths whose first
$m_{n}$ elements for the $n$-th strategy match a given sequence
of strategy profiles, $(s_{t,n})_{t=1}^{m_{n}}.$ (If $m_{n}=0$,
then there is no restriction on $\mathfrak{S}_{t,n}$ for any $t.$)
We specify the distribution $\eta$ by specifying the probability
it assigns to these cylinder sets:

\[
\eta\left\{ ((s_{t,n})_{t=1}^{m_{n}})_{n=1}^{N}\right\} =\prod_{n=1}^{N}\prod_{t=1}^{m_{n}}\sigma(s_{t,n}),
\]
where we have abused notation to write $((s_{t,n})_{t=1}^{m_{n}})_{n=1}^{N}$
for the cylinder set satisfying this profile of sequences, and we
have used the convention that the empty product is defined to be 1.

We establish the claim by induction on $t$ for period-$t$ histories.
For $t\ge0,$ let $Y_{i}[t]\subseteq Y_{i}$ be the set of possible
period-$t$ histories of $i,$ that is $Y_{i}[t]:=(\mathbb{S}_{i}\times\mathbb{O}_{i})^{t}$.
In the base case of $t=1,$ we show playing against a response path
drawn according to $\eta$ and playing against a pure strategy\footnote{In the random matching model agents are facing a randomly drawn pure
strategy profile each period (and not a fixed behavior strategy):
they are matched with random opponents, who each play a pure strategy
in the game as a function of their personal history. From Kuhn's theorem,
this is equivalent to facing a fixed profile of behavior strategies.} drawn from $\sigma_{-i}\in\times_{k\ne i}\Delta(\mathbb{S}_{k})$
generate the same period-1 history. Fixing a learning policy $r_{i}:Y_{i}\to\mathbb{S}_{i}$
of $i,$ the probability of $i$ having the period-1 history $(s_{i}^{(1)},o^{(1)})\in Y_{i}[1]$
in the random-matching model is $\boldsymbol{1}(r_{i}(\emptyset)=s_{i}^{(1)})\cdot\sigma(s:\mathfrak{o}_{i}(Z(s_{i}^{(1)},s_{-i}))=o^{(1)})$.
That is, $i$'s policy must play $s_{i}^{(1)}$ in the first period
of $i$'s life. Then, $i$ must encounter such a pure strategy that
generates the required observation $o^{(1)}$, and this has probability
$\sigma(s:\mathfrak{o}_{i}(Z(s_{i}^{(1)},s_{-i}))=o^{(1)})$. The
probability of this happening against a response path drawn from $\eta$
is 
\begin{align*}
 & \boldsymbol{1}(r_{i}(\emptyset)=s_{i}^{(1)})\cdot\eta(\mathfrak{S}\mathfrak{:}\mathfrak{o}_{i}(Z(s_{i}^{(1)},s_{1,s_{i}^{(1)},-i}))=o^{(1)})\\
= & \boldsymbol{1}(r_{i}(\emptyset)=s_{i}^{(1)})\cdot\sigma(s:\mathfrak{o}_{i}(Z(s_{i}^{(1)},s_{-i}))=o^{(1)}),
\end{align*}
where the second line comes from the probability $\eta$ assigns to
cylinder sets.

We now proceed with the inductive step. By induction, suppose random
matching and the $\eta$-distributed response path induce the same
distribution over the set of period-$T$ histories, $Y_{i}[T]$, where
$T\ge1.$ Write this common distribution as $\phi_{i,T}^{RM}=\phi_{i,T}^{\eta}=\phi_{i,T}\in\Delta(Y_{i}[T]).$
We prove that they also generate the same distribution over length
$T+1$ histories.

Suppose random matching generates distribution $\phi_{i,T+1}^{RM}\in\Delta(Y_{i}[T+1])$
and the $\eta$-distributed response path generates distribution $\phi_{i,T+1}^{\eta}\in\Delta(Y_{i}[T+1]).$
Each length $T+1$ history $y_{i}[T+1]\in Y_{i}[T+1]$ may be written
as $(y_{i}[T],(s_{i}^{(T+1)},o^{(T+1)})),$ where $y_{i}[T]$ is a
length-$T$ history and $(s_{i}^{(T+1)},o^{(T+1)})$ is a one-period
history corresponding to what happens in period $T+1$. Therefore,
we may write for each $y_{i}[T+1],$ 
\[
\phi_{i,T+1}^{RM}(y_{i}[T+1])=\phi_{i,T}^{RM}(y_{i}[T])\cdot\phi_{i,T+1|T}^{RM}((s_{i}^{(T+1)},o^{(T+1)})|y_{i}[T]),
\]
 and 
\[
\phi_{i,T+1}^{\eta}(y_{i}[T+1])=\phi_{i,T}^{\eta}(y_{i}[T])\cdot\phi_{i,T+1|T}^{\eta}(((s_{i}^{(T+1)},o^{(T+1)})|y_{i}[T]),
\]
where $\phi_{i,T+1|T}^{RM}$ and $\phi_{i,T+1|T}^{\eta}$ are the
conditional probabilities of the form ``having history $(s_{i}^{(T+1)},o^{(T+1)})$
in period $T+1,$ conditional on having history $y_{i}[T]\in Y_{i}[T]$
in the first $T$ periods.'' If such conditional probabilities are
always the same for the random-matching model and the $\eta$-distributed
response path model, then from the hypothesis $\phi_{i,T}^{RM}=\phi_{i,T}^{\eta},$
we can conclude $\phi_{i,T+1}^{RM}=\phi_{i,T+1}^{\eta}.$

By argument exactly analogous to the base case, we have for the random-matching
model
\[
\phi_{i,T+1|T}^{RM}((s_{i}^{(T+1)},o^{(T+1)})|y_{i}[T])=\boldsymbol{1}(r_{i}(y_{i}(T))=s_{i}^{(T+1)})\cdot\sigma(s:\mathfrak{o}_{i}(Z(s_{i}^{(T+1)},s_{-i}))=o^{(T+1)}),
\]
since the matching is independent across periods. In the $\eta$-distributed
response path model, since a single response path is drawn once and
fixed, one must compute the conditional probability that the drawn
$\mathfrak{\mathfrak{S}}$ is such that the observation $o^{(T+1)}$will
be seen in period $T+1$, given the history $y_{i}[T]$ (which is
informative about which response path $i$ is facing).

For each $1\le n\le N,$ let the non-negative integer $m_{n}$ represent
the number of times $i$ has used the $n$-th strategy in $\mathbb{S}_{i}$
in the history $y_{i}[T].$ Let $(o_{t,n})_{1\le t\le m_{n}}$ represent
the sequence of observations seen after using the $n$-th strategy,
in chronological order. Consider the following finite union of cylinder
sets, $(s_{t,n}:\mathfrak{o}_{i}(Z(n,s_{t,n,-i}))=o_{t,n})_{1\le t\le m_{n},1\le n\le N}$.
This is the set of response sequences consistent with the observations
so far.

If $\mathfrak{S}$ is to produce the observation $o^{(T+1)}$ from
$i$'s next play of $s_{i}^{(T+1)}$, then $\mathfrak{S}$ must belong
to a more restrictive cylinder set that satisfies the additional restriction
$(s_{m_{s_{i}^{(T+1)}}+1,s_{i}^{(T+1)}}:\mathfrak{o}_{i}(Z(s_{i}^{(T+1)},s_{-i}))=o_{m_{s_{i}^{(T+1)}}+1,s_{i}^{(T+1)}})$.
The conditional probability of $\mathfrak{S}$ belonging to this more
restrictive cylinder set, given that it falls in $(s_{t,n}:\mathfrak{o}_{i}(Z(n,s_{t,n,-i}))=o_{t,n})_{1\le t\le m_{n},1\le n\le N},$
is given by the ratio of $\eta$-probabilities of these unions of
cylinder sets, which from the product structure of $\eta$ on cylinder
sets, must be $\sigma(s:\mathfrak{o}_{i}(Z(s_{i}^{(T+1)},s_{-i}))=o^{(T+1)}).$
\end{proof}
Thus, to prove that $\phi_{i}(s_{i}^{*};r_{i},\sigma_{-i})\ge\phi_{j}(s_{j}^{*};r_{j},\sigma_{-j}),$
it suffices to show that for every $\mathfrak{S}$, the period where
$s_{i}^{*}$ is played for the $k$-th time in induced history $y_{i}(\mathfrak{S},r_{i})$
happens earlier than the period where $s_{j}^{*}$ is played for the
$k$-th time in history $y_{j}(\mathfrak{S},r_{j}).$

Now we turn to the proof of Proposition \ref{prop:index}.
\begin{proof}
Let $0\le\gamma<1$ and the social distribution $\sigma$ be fixed.
Enumerate the strategy sets of $i$ and $j$ so that $s_{i}$ and
$\varphi(s_{i})$ are assigned the same number for every $s_{i}\in\mathbb{S}_{i}.$
Consider the product distribution $\eta$ on the space of response
paths, $((\mathbb{S})^{N})^{\infty}$, as in the proof of Lemma \ref{lem:response_path}.

By Lemma \ref{lem:response_path}, denote the period where $s_{i}^{*}$
appears in $y_{i}(\mathfrak{S},r_{i})$ for the $k$-th time as $T_{i}^{(k)}$,
the period where $s_{j}^{*}$ appears in $y_{j}(\mathfrak{S},r_{j})$
for the $k$-th time as $T_{j}^{(k)}.$ The quantities $T_{i}^{(k)},T_{j}^{(k)}$
are defined to be $\infty$ if the corresponding strategies do not
appear at least $k$ times in the infinite histories. Write $\#(s_{i}^{'};k)\in\mathbb{N}\cup\{\infty\}$
be the number of times $s_{i}^{'}\in\mathbb{S}_{i}$ is played in
the history $y_{i}(\mathfrak{S},r_{i})$ before $T_{i}^{(k)}$. Similarly,
$\#(s_{j}^{'};k)\in\mathbb{N}\cup\{\infty\}$ denotes the number of
times $s_{j}^{'}\in\mathbb{S}_{j}$ is played in the history $y_{j}(\mathfrak{S},r_{j})$
before $T_{j}^{(k)}$. Since $\varphi$ establishes a bijection between
$\mathbb{S}_{i}$ and $\mathbb{S}_{j}$, it suffices to show that
for every $k=1,2,3,...$ either $T_{j}^{(k)}=\infty$ or for all $s_{i}^{'}\ne s_{i}^{*}$,
$\#(s_{i}^{'};k)\le\#(s_{j}^{'};k)$ where $s_{j}^{'}=\varphi(s_{i}^{'}).$

We show this by induction on $k$. First we establish the base case
of $k=1$.

Suppose $T_{j}^{(1)}\ne\infty$, and, by way of contradiction, suppose
there is some $s_{i}^{'}\ne s_{i}^{*}$ such that $\#(s_{i}';1)>\#(\varphi(s_{i}^{'});1)$.
Find the subhistory $y_{i}$ of $y_{i}(\mathfrak{S},r_{i}).$ that
leads to $s_{i}^{'}$ being played for the $(\#(\varphi(s_{i}^{'});1)+1)$-th
time, and find the subhistory $y_{j}$ of $y_{j}(\mathfrak{S},r_{j})$
that leads to $j$ playing $s_{j}^{*}$ for the first time ($y_{j}$
is well-defined because $T_{j}^{(1)}\ne\infty$). Note that $y_{i,s_{i}^{*}}\sim y_{j,s_{j}^{*}}$
vacuously, since $i$ has never played $s_{i}^{*}$ in $y_{i}$ and
$j$ has never played $s_{j}^{*}$ in $y_{j}$.

Also, $y_{i,s_{i}^{'}}\sim y_{j,s_{j}^{'}}$. To see this, note that
$i$ has played $s_{i}^{'}$ for $\#(\varphi(s_{i}^{'});1)$ times
and $j$ has played $s_{j}^{'}$ for the same number of times. The
definition of response paths implies they faced the same sequence
of opponent strategy profiles, and the definition of isomorphic learning
problems implies they have gotten equivalent observations in all these
periods.

Since $r_{j}(y_{j})=s_{j}^{*}$ and $r_{j}$ is an index policy, $s_{j}^{*}$
must have weakly the highest index at $y_{j}.$ Since $r_{i}$ is
more compatible with $s_{i}^{*}$ than $r_{j}$ is with $s_{j}^{*}$,
$s_{i}^{'}$ must not have the weakly highest index at $y_{i}$. And
yet $r_{i}(y_{i})=s_{i}'$ contradiction.

Now suppose this statement holds for all $k\le K$ for some $K\ge1.$
We show it also holds for $k=K+1.$ If $T_{j}^{(K+1)}=\infty$ or
$T_{j}^{(K)}=\infty$, we are done. Otherwise, by way of contradiction,
suppose there is some $s_{i}^{'}\ne s_{i}^{*}$ so that $\#(s_{i}';K+1)>\#(\varphi(s_{i}^{'});K+1)$.
Find the subhistory $y_{i}$ of $y_{i}(\mathfrak{S},r_{i}).$ that
leads to $s_{i}^{'}$ being played for the $(\#(\varphi(s_{i}^{'});K+1)+1)$-th
time. Since $T_{j}^{(K)}\ne\infty$, from the inductive hypothesis
$T_{i}^{(K)}\ne\infty$ and $\#(s_{i}';K)\le\#(\varphi(s_{i}^{'});K)$.
That is, $i$ must have played $s_{i}^{'}$ no more than $\#(\varphi(s_{i}^{'});K)$
times before playing $s_{i}^{*}$ for the $K$-th time. Since $\#(\varphi(s_{i}^{'});K+1)+1>\#(\varphi(s_{i}^{'});K),$
the subhistory $y_{i}$ must extend beyond period $T_{i}^{(K)}$,
so it contains $K$ instances of $i$ playing $s_{i}^{*}$.

Next, find the subhistory $y_{j}$ of $y_{j}(\mathfrak{S},r_{j})$
that leads to $j$ playing $s_{j}^{*}$ for the $(K+1)$-th time.
(This is well-defined because $T_{j}^{(K+1)}\ne\infty$.) Note that
$y_{i,s_{i}^{*}}\sim y_{j,s_{j}^{*}}$, since $i$ and $j$ have played
$s_{i}^{*},s_{j}^{*}$ for $K$ times each, and they were facing the
same response paths. Also, $y_{i,s_{i}^{'}}\sim y_{j,s_{j}^{'}}$
since $i$ has played $s_{i}^{'}$ for $\#(\varphi(s_{i}^{'});K+1)$
times and $j$ has played $s_{j}^{'}$ for the same number of times.
Since $r_{j}(y_{j})=s_{j}^{*}$ and $r_{j}$ is an index policy, $s_{j}^{*}$
must have weakly the highest index at $y_{j}.$ Since $r_{i}$ is
more compatible with $s_{i}^{*}$ than $r_{j}$ is with $s_{j}^{*}$,
$s_{i}^{'}$ must not have the weakly highest index at $y_{i}$. And
yet $r_{i}(y_{i})=s_{i}'$ contradiction.
\end{proof}

\subsection{Proof of Lemma \ref{lem:one_step_lem1}}
\begin{proof}
By way of contradiction, suppose there is some profile of moves by
$-i$, $(a_{h})_{h\in\mathcal{H}_{-i}}$, so that $h^{*}$ is off
the path of play in $(s_{i},(a_{h})_{h\in\mathcal{H}_{-i}})=(s_{i},a_{h^{*}},(a_{h})_{h\in\mathcal{H}_{-i}\backslash h^{*}})$.
Find a different action of $j$ on $h^{*}$, $a_{h^{*}}^{'}\ne a_{h^{*}}$.
Since $h^{*}$ is off the path of play, both $(s_{i},a_{h^{*}},(a_{h})_{h\in\mathcal{H}_{-i}\backslash h^{*}})$
and $(s_{i},a_{h^{*}}^{'},(a_{h})_{h\in\mathcal{H}_{-i}\backslash h^{*}})$
lead to the same payoff for $i$. But by Condition (1) in the definition
of factorability and the fact that $h^{*}\in F_{i}[s_{i}]$, we will
have found two $-i$ action profiles $s_{-i},s_{-i}^{'}$ in two different
blocks of $\Pi_{i}[s_{i}]$ with $u_{i}(s_{i},s_{-i})=u_{i}(s_{i},s_{-i}^{'})$.
This contradicts $\Pi_{i}[s_{i}]$ being the coarsest partition of
$\mathbb{S}_{-i}$ that makes $u_{i}(s_{i},\cdot)$ measurable.
\end{proof}

\subsection{Proof of Lemma \ref{lem:one_step_lem2}}
\begin{proof}
Since $i$'s payoff is not independent of $h^{*},$ there exist actions
$a_{h^{*}}\ne a_{h^{*}}^{'}$ on $h^{*}$ and a profile $a_{-h^{*}}$
of actions elsewhere in the game tree, so that $u_{i}(a_{h^{*}},a_{-h^{*}})\ne u_{i}(a_{h^{*}}^{'},a_{-h^{*}})$.
Consider the strategy $s_{i}$ for $i$ that matches $a_{-h^{*}}$
in terms of play on $i$'s information sets, so we may equivalently
write 
\[
u_{i}(s_{i},a_{h^{*}},(a_{h})_{h\in\mathcal{H}_{-i}\backslash h^{*}})\ne u_{i}(s_{i},a_{h^{*}}^{'},(a_{h})_{h\in\mathcal{H}_{-i}\backslash h^{*}}),
\]
 where $(a_{h})_{h\in\mathcal{H}_{-i}\backslash h^{*}}$ are the components
of $a_{-h^{*}}$ corresponding to information sets of $-i$. If $h^{*}\notin F_{i}[s_{i}],$
then by Condition (1) of factorability, $(a_{h^{*}},(a_{h})_{h\in\mathcal{H}_{-i}\backslash h^{*}})$
and $(a_{h^{*}}^{'},(a_{h})_{h\in\mathcal{H}_{-i}\backslash h^{*}})$
belong to the same block in $\Pi_{i}[s_{i}].$ Yet, they give different
payoffs to $i$, which contradicts that $i$'s payoff after $s_{i}$
must be measurable with respect to $\Pi_{i}[s_{i}].$
\end{proof}

\subsection{Proof of Proposition \ref{prop:one_step}}
\begin{proof}
Combining Lemmas \ref{lem:one_step_lem1} and \ref{lem:one_step_lem2}
implies there is an extensive-form strategy $s_{i}\in\mathbb{S}_{i}$
such that $h^{*}$ is on the path of play whenever $i$ chooses $s_{i}.$
Consider some strategy profile $(s_{i}^{\circ},s_{-i}^{\circ})$ where
$h^{*}$ is off the path. Then $i$ can unilaterally deviate to $s_{i}$,
and $h^{*}$ is on the path of $(s_{i},s_{-i}^{\circ})$. Furthermore,
$i$'s play differs on the new path relative to the old path on exactly
one information set, since $i$ plays at most once on any path. So
instead of deviating to $s_{i},$ $i$ can deviate to $s_{i}'$ that
matches $s_{i}$ in terms of this information set where $i$'s play
is modified, but otherwise is the same as $s_{i}^{\circ}$. So $h^{*}$
is also on the path of play for $(s_{i}'s_{-i}^{\circ}),$ where $s_{i}^{'}$
differs from $s_{i}^{\circ}$ only on one information set.
\end{proof}

\section{\label{sec:details_2learning_models}Index Compatibility of OPT and
WFP when $s_{i}^{*}\succsim s_{j}^{*}$}

In this section, we show that OPT and WFP are index compatible under
the conditions of Theorem \ref{thm:crossplayer_tremble_foundation}.
This conclusion, when combined with Proposition \ref{prop:index},
implies Theorem \ref{thm:crossplayer_tremble_foundation}.

Let the equivalence classes $\mathbb{E}$ be such that $(s_{i},u_{i}(s_{i},\tilde{s}_{-i}))\sim(\varphi(s_{i}),u_{j}(\varphi(s_{i}),\hat{s}_{-j})$
if and only if $\tilde{s}_{-i}|_{F_{i}[s_{i}]\cap\mathcal{H}_{-ij}}=\hat{s}_{-j}|_{F_{j}[\varphi(s_{i})]\cap\mathcal{H}_{-ij}}.$
Conditions on factorability and isomorphic factoring ensure that $i,j$
face $(\varphi,\mathbb{E})$-isomorphic learning problems. Indeed,
if $i$ and $j$ faced the same pure profile $\tilde{s},$ then $\tilde{s}_{-i}|_{F_{i}[s_{i}]\cap\mathcal{H}_{-ij}}=\tilde{s}_{-j}|_{F_{j}[\varphi(s_{i})]\cap\mathcal{H}_{-ij}}$
since $F_{i}[s_{i}]\cap\mathcal{H}_{-ij}=F_{j}[\varphi(s_{i})]\cap\mathcal{H}_{-ij}$
by isomorphic factoring.

\subsection{Weighted Fictitious Play}

To see that WFP satisfies index compatibility for $s_{i}^{*}$ and
$s_{j}^{*}$ under the conditions of Theorem \ref{thm:crossplayer_tremble_foundation},
let histories $y_{i},y_{j}$ and strategy $s_{i}^{'}\ne s_{i}^{*}$
be given with $y_{i,s_{i}^{*}}\sim y_{j,s_{j}^{*}}$, $y_{i,s_{i}^{'}}\sim y_{j,\varphi(s_{i}^{'})}$,
and $s_{j}^{*}$ having weakly the highest index for $j$. Construct
two totally mixed, independent behavior strategy profile, $\beta,\tilde{\beta}$
as follows. For each $s_{j}\in\mathbb{S}_{j},$ $\beta(h):=\alpha_{h}(\cdot;y_{j})$
for all $h\in F_{j}[s_{j}]$. (This is well-defined by Condition (2)
of factorability, as $F_{j}[s_{j}]\cap F_{j}[s_{j}']=\varnothing$
if $s_{j}\ne s_{j}^{'}.)$ For those $h\in\mathcal{H}\backslash\cup_{s_{j}\in\mathbb{S}_{j}}F_{j}[s_{j}],$
arbitrarily specify a strictly mixed action $\alpha_{h}\in\Delta(A_{h})$
for $\beta(h).$ Having constructed $\beta$ we turn to $\tilde{\beta}.$
For each $s_{i}\in\{s_{i}^{*},s_{i}^{'}\},$ $\tilde{\beta}(h):=\alpha_{h}(\cdot;y_{i})$
for all $h\in F_{i}[s_{i}]$. For all other $h\in\mathcal{H},$ let
$\tilde{\beta}(h):=\beta(h).$

From the definition of $y_{i,s_{i}^{*}}\sim y_{j,s_{j}^{*}},$ $\tilde{\beta}(h)=\beta(h)$
for all $h\in F_{i}[s_{i}]\cap\mathcal{H}_{-ij}.$ From the definition
of $y_{i,s_{i}^{'}}\sim y_{j,\varphi(s_{i}^{'})},$ $\tilde{\beta}(h)=\beta(h)$
for all $h\in F_{i}[s_{i}']\cap\mathcal{H}_{-ij}.$ Also, $\tilde{\beta}(h)=\beta(h)$
for all other $h\in\mathcal{H}_{-ij}$ by construction. So, $\tilde{\beta}$
and $\beta$ are totally mixed behavior strategy profiles that match
on the $-ij$ marginal, and they can be represented by $\tilde{\sigma},\sigma$
totally mixed strategy profiles (over $\mathbb{S})$ that match on
the $-ij$ marginal.

Since $j$'s payoff from each $s_{j}$ only depends on $-j$'s play
on $F_{j}[s_{j}]$ by Condition (1) of factorability, $\sum_{s\in\mathbb{S}}u_{j}(s_{j},s_{-j})\cdot\sigma(s)$
equals to the index that the weighted fictitious play agent assigns
to $s_{j}$ after history $y_{j}$. Since $s_{j}^{*}$ has the weakly
highest index, $\sum_{s\in\mathbb{S}}u_{j}(s_{j}^{*},s_{-j})\cdot\sigma(s)=\max_{s_{j}^{'}\in\mathbb{S}_{j}}\sum_{s\in\mathbb{S}}u_{j}(s_{j}^{'},s_{-j})\cdot\sigma(s)$.
From the definition of player compatibility, $s_{i}^{*}$ is strictly
optimal against $\tilde{\sigma}$, which in particular means $\sum_{s\in\mathbb{S}}u_{i}(s_{i}^{*},s_{-i})\cdot\tilde{\sigma}(s)>\sum_{s\in\mathbb{S}}u_{i}(s_{i}^{'},s_{-i})\cdot\tilde{\sigma}(s)$.
The RHS is $i$'s index for $s_{i}'$ after $y_{i},$ since $\tilde{\sigma}$
marginalized to every $h\in F_{i}[s_{i}']$ is $\alpha_{h}(\cdot;y_{i})$
by construction. This says $s_{i}^{'}$ does not have the weakly highest
index for $i$ after $y_{i}.$

Thus, WFP satisfies index compatibility for $s_{i}^{*}$ and $s_{j}^{*}$.

\subsection{\label{subsec:The-Heuristic-Gittins}The Gittins Index}

Write $V(\tau;s_{i},\nu_{s_{i}})$ for the value of the above auxiliary
problem under the (not necessarily optimal) stopping time $\tau$
in the definition of the Gittins index. The Gittins index of $s_{i}$
is $\sup_{\tau>0}V(\tau;s_{i},\nu_{s_{i}})$. We begin by linking
$V(\tau;s_{i},\tau_{s_{i}})$ to $i$'s stage-game payoff from playing
$s_{i}$. From belief $\nu_{s_{i}}$ and stopping time $\tau$, we
will construct the correlated profile $\alpha(\nu_{s_{i}},\tau)\in\Delta^{\circ}(\times_{h\in F_{i}[s_{i}]}A_{h})$,
so that $V(\tau;s_{i},\nu_{s_{i}})$ is equal to $i$'s expected payoff
when playing $s_{i}$ while opponents play according to this correlated
profile on the $s_{i}$-relevant information sets.
\begin{defn}
A full-support belief $\nu_{s_{i}}\in\times_{h\in F_{i}[s_{i}]}\Delta(\Delta(A_{h}))$
for player $i$ together with a (possibly random) stopping rule $\tau>0$
together induce a stochastic process $(\tilde{\boldsymbol{a}}_{(-i),t})_{t\ge1}$
over the space $\times_{h\in F_{i}[s_{i}]}A_{h}\cup\{\varnothing\}$,
where $\tilde{\boldsymbol{a}}_{(-i),t}\in\times_{h\in F_{i}[s_{i}]}A_{h}$
represents the opponents' actions observed in period $t$ if $\tau\ge t$,
and $\tilde{\boldsymbol{a}}_{(-i),t}=\varnothing$ if $\tau<t$. We
call $\tilde{\boldsymbol{a}}_{(-i),t}$ player $i$'s \emph{internal
history} at period $t$ and write $\mathbb{P}_{(-i)}$ for the distribution
over internal histories that the stochastic process induces.
\end{defn}
Internal histories live in the same space as player $i$'s actual
experience in the learning problem, represented as a history in $\mathbb{O}_{i}$.
The process over internal histories is $i$'s prediction about what
would happen in the auxiliary problem if they were to use $\tau.$

Enumerate all possible profiles of moves at information sets $F_{i}[s_{i}]$
as $\times_{h\in F_{i}[s_{i}]}A_{h}=\{\boldsymbol{a}_{(-i)}^{(1)},...,\boldsymbol{a}_{(-i)}^{(K)}\}$,
let $p_{t,k}:=\mathbb{P}_{(-i)}[\tilde{\boldsymbol{a}}_{(-i),t}=\boldsymbol{a}_{(-i)}^{(k)}]$
for $1\le k\le K$ be the probability under $\nu_{s_{i}}$of seeing
the profile of actions $\boldsymbol{a}_{(-i)}^{(k)}$ in period $t$
of the stochastic process over internal histories, $(\tilde{\boldsymbol{a}}_{(-i),t})_{t\ge0}$,
and let $p_{t,0}:=\mathbb{P}_{(-i)}[\tilde{\boldsymbol{a}}_{(-i),t}=\varnothing]$
be the probability of having stopped before period $t.$
\begin{defn}
The\emph{ synthetic correlated profile} at information sets in $F_{i}[s_{i}]$
is the element of $\Delta^{\circ}(\times_{h\in F_{i}[s_{i}]}A_{h})$
(i.e. a correlated random action) that assigns probability $\frac{\sum_{t=1}^{\infty}\beta^{t-1}p_{t,k}}{\sum_{t=1}^{\infty}\beta^{t-1}(1-p_{t,0})}$
to the profile of actions $\boldsymbol{a}_{(-i)}^{(k)}$. Denote this
profile by $\alpha(\nu_{s_{i}},\tau)$.
\end{defn}
Note that the synthetic correlated profile depends on the belief $\nu_{s_{i}}$
stopping rule $\tau,$ and effective discount factor $\beta$. Since
the belief $\nu_{s_{i}}$ has full support, there is always a positive
probability assigned to observing every possible profile of actions
on $F_{i}[s_{i}]$ in the first period, so the synthetic correlated
profile is totally mixed. The significance of the synthetic correlated
profile is that it gives an alternative expression for the value of
the auxiliary problem under stopping rule $\tau$.
\begin{lem}
\[
V(\tau;s_{i},\nu_{s_{i}})=u_{i}(s_{i},\alpha(\nu_{s_{i}},\tau))
\]
\end{lem}
The proof is the same as in \citet{fudenberg_he_2017} and is omitted.\footnote{Notice that even though $i$ starts with the belief that opponents
randomize independently at different information sets, and also holds
an independent prior belief, $V(\tau;s_{i},\nu_{s_{i}})$ may not
be the payoff of playing $s_{i}$ against a independent randomizations
by the opponent because of the endogenous correlation that we discussed
in the text.}

Consider now the situation where $i$ and $j$ share the same beliefs
about play of $-ij$ on the common information sets $F_{i}[s_{i}]\cap F_{j}[s_{j}]\subseteq\mathcal{H}_{-ij}$.
For any pure-strategy stopping time $\tau_{j}$ of $j$, we define
a random stopping rule of $i$, the \emph{mimicking stopping time}
for $\tau_{j}$. Lemma \ref{lem:mimic} will establish that the mimicking
stopping time generates a synthetic correlated profile that matches
the corresponding profile of $\tau_{j}$ on $F_{i}[s_{i}]\cap F_{j}[s_{j}]$.

Note that $\tau_{j}$ maps $j$'s internal histories to stopping decisions,
which do not live in the same space as $i$'s internal histories.
In particular, $\tau_{j}$ could make use of $i$'s play to decide
whether to stop. To mimic such a rule, $i$ makes use of external
histories, which include both the common component of $i$'s internal
history on $F_{i}[s_{i}]\cap F_{j}[s_{j}],$ as well as simulated
histories on $F_{j}[s_{j}]\backslash(F_{i}[s_{i}]\cap F_{j}[s_{j}]).$

For a given isomorphism $\varphi$ between $i$ and $j$ with $\varphi(s_{i})=s_{j}$
and $F_{i},F_{j},$ we may write $F_{i}[s_{i}]=F^{C}\cup\bar{F}^{-i}$
with $F^{C}\subseteq\mathcal{H}_{-ij}$ and $\bar{F}^{-i}\subseteq\mathcal{H}_{-i}$.
Similarly, we may write $F_{j}[s_{j}]=F^{C}\cup\bar{F}^{-j}$ with
$\bar{F}^{-j}\subseteq\mathcal{H}_{-j}$. (So, $F^{C}$ is the \textbf{c}ommon
information sets that are observed after both $s_{i}$ and $s_{j}$.)
Whenever $j$ plays $s_{j},$ they observes some $(\boldsymbol{a}_{(C)},\boldsymbol{a}_{(-j)})\in(\times_{h\in F^{C}}A_{h})\times(\times_{h\in\bar{F}^{-j}}A_{h})$,
where $\boldsymbol{a}_{(C)}$ is a profile of actions at information
sets in $F^{C}$ and $\boldsymbol{a}_{(-j)}$ is a profile of actions
at information sets in $\bar{F}^{-j}$. So a pure-strategy stopping
rule in the auxiliary problem defining $j$'s Gittins index for $s_{j}$
is a function $\tau_{j}:\cup_{t\ge1}[(\times_{h\in F^{C}}A_{h})\times(\times_{h\in\bar{F}^{-j}}A_{h})]^{t}\to\{0,1\}$
that maps finite histories in $\mathbb{O}_{j}$to stopping decisions,
where ``0'' means continue and ``1'' means stop.
\begin{defn}
Player $i$'s\emph{ mimicking stopping rule for $\tau_{j}$ }draws
$\alpha^{-j}\in\times_{h\in\bar{F}^{-j}}\Delta(A_{h})$ from $j$'s
belief $\nu_{s_{j}}$ on $\bar{F}^{-j}$, and then draws $(\boldsymbol{a}_{(-j),\ell})_{\ell\ge1}$
by independently generating $\boldsymbol{a}_{(-j),\ell}$ from $\alpha^{-j}$
each period. Conditional on $(\boldsymbol{a}_{(-j),\ell}),$ $i$
stops according to the rule $(\tau_{i}|(\boldsymbol{a}_{(-j),\ell}))((\boldsymbol{a}_{(C),\ell},\boldsymbol{a}_{(-i),\ell})_{\ell=1}^{t}):=\tau_{j}((\boldsymbol{a}_{(C),\ell},\boldsymbol{a}_{(-j),\ell})_{\ell=1}^{t}).$\footnote{Note this is a valid (stochastic) stopping time, as the event $\{\tau_{i}\le T\}$
only depends on $i$'s observations in $\mathbb{O}_{i}$ in the first
$T$ periods, plus some private randomizations of $i.$}
\end{defn}
That is, the mimicking stopping rule involves ex-ante randomization
across a family of pure-strategy stopping rules $\tau_{i}|(\boldsymbol{a}_{(-j),\ell})_{\ell=1}^{\infty}$,
indexed by $(\boldsymbol{a}_{(-j),\ell})_{\ell=1}^{\infty}$. First,
$i$ draws a behavior strategy on the information sets $\bar{F}^{-j}$
according to $j$'s belief about $-j$'s play there. Then, $i$ simulates
an infinite sequence $(\boldsymbol{a}_{(-j),\ell})_{\ell=1}^{\infty}$
of $i$'s play using this drawn behavior strategy and follows the
pure-strategy stopping rule $\tau_{i}|(\boldsymbol{a}_{(-j),\ell})_{\ell=1}^{\infty}.$

As in the definition of internal histories, the mimicking strategy
and $i$'s belief $\nu_{s_{i}}$ generates a stochastic process $(\tilde{\boldsymbol{a}}_{(-i),t},\tilde{\boldsymbol{a}}_{(C),t})_{t\ge1}$
of internal histories for $i$ (representing actions on $F_{i}[s_{i}]$
that $i$ anticipates seeing when they plays $s_{i}$). It also induces
a stochastic process $(\tilde{\boldsymbol{e}}_{(-j),t},\tilde{\boldsymbol{e}}_{(C),t})_{t\ge1}$
of ``external histories'' defined in the following way:
\begin{defn}
The stochastic process of \emph{external historie}s $(\tilde{\boldsymbol{e}}_{(-j),t},\tilde{\boldsymbol{e}}_{(C),t})_{t\ge1}$
is defined from the process of internal histories $(\tilde{\boldsymbol{a}}_{(-i),t},\tilde{\boldsymbol{a}}_{(C),t})_{t\ge1}$
that $\tau_{i}$ generates and given by: (i) if $\tau_{i}<t$, then
$(\tilde{\boldsymbol{e}}_{(-j),t},\tilde{\boldsymbol{e}}_{(C),t})=\varnothing;$
(ii) otherwise, $\tilde{\boldsymbol{e}}_{(C),t}=\tilde{\boldsymbol{a}}_{(C),t}$,
and $\tilde{\boldsymbol{e}}_{(-j),t}$ is the $t$-th element of the
infinite sequence $(\boldsymbol{a}_{(-j),\ell})_{\ell=1}^{\infty}$
that $i$ simulated before the first period of the auxiliary problem.
\end{defn}
Write $\mathbb{P}_{e}$ for the distribution over the sequence of
of external histories generated by $i$'s mimicking stopping time
for $\tau_{j},$ which is a function of $\tau_{j},\nu_{s_{j}}$, and
$\nu_{s_{i}}$.\footnote{To understand the distinction between internal and external histories,
note that the probability of $i$'s first-period internal history
satisfying $(\tilde{\boldsymbol{a}}_{(-i),1},\tilde{\boldsymbol{a}}_{(C),1})=(\bar{\boldsymbol{a}}_{(-i)},\bar{\boldsymbol{a}}_{(C)})$
for some fixed values $(\bar{\boldsymbol{a}}_{(-i)},\bar{\boldsymbol{a}}_{(C)})\in\times_{h\in F_{i}[s_{i}]}A_{h}$
is given by the probability that a mixed play $\alpha_{-i}$ on $F_{i}[s_{i}],$
drawn according to $i$'s belief $\nu_{s_{i}},$ would generate the
profile of actions $(\bar{\boldsymbol{a}}_{(-i)},\bar{\boldsymbol{a}}_{(C)}).$
On the other hand, the probability of $i$'s first-period external
history satisfying $(\tilde{\boldsymbol{e}}_{(-j),1},\tilde{\boldsymbol{e}}_{(C),1})=(\bar{\boldsymbol{a}}_{(-j)},\bar{\boldsymbol{a}}_{(C)})$
for some fixed values $(\bar{\boldsymbol{a}}_{(-j)},\bar{\boldsymbol{a}}_{(C)})\in\times_{h\in F_{j}[s_{j}]}A_{h}$
also depends on $j$'s belief $\nu_{s_{j}}$, for this belief determines
the distribution over $(\boldsymbol{a}_{(-j),\ell})_{\ell=1}^{\infty}$
drawn before the start of the auxiliary problem.}

When using the mimicking stopping time for $\tau_{j}$ in the auxiliary
problem, $i$ expects to see the same distribution of $-ij$'s play
before stopping as $j$ does when using $\tau_{j}$, on the information
sets in $F_{i}[s_{i}]\cap F_{j}[s_{j}]$. This is formalized in the
next lemma.
\begin{lem}
\label{lem:mimic} Suppose the game is isomorphcially factorable for
$i$ and $j$ with $\varphi(s_{i})=s_{j}$, and suppose $i$ holds
belief $\nu_{s_{i}}$ over play in $F_{i}[s_{i}]$ and $j$ holds
belief $\nu_{s_{j}}$ over play in $F_{j}[s_{j}],$ such that $\nu_{s_{i}}|_{F_{i}[s_{i}]\cap F_{j}[s_{j}]}=\nu_{s_{j}}|_{F_{i}[s_{i}]\cap F_{j}[s_{j}]}$,
that is the two sets of beliefs match when marginalized to the common
information sets in $\mathcal{H}_{-ij}$. Let $\tau_{i}$ be $i$'s
mimicking stopping time for $\tau_{j}$. Then, the synthetic correlated
profile $\alpha(\nu_{s_{j}},\tau_{j})$ marginalized to the information
sets of $-ij$ is the same as $\alpha(\nu_{s_{i}},\tau_{i})$ marginalized
to the same information sets.
\end{lem}
\begin{prop}
\label{prop:Gittins_matching_beliefs} Suppose the game is isomorphcially
factorable for $i$ and $j$ with $\varphi(s_{i})=s_{j}$, $\varphi(s_{i}^{'})=s_{j}^{'},$
where $s_{i}^{*}\ne s_{i}^{'}$. Suppose $i$ is more player-compatible
with $s_{i}^{*}$ than $j$ is with $s_{j}^{*}$. Suppose $i$ holds
belief $\nu_{s_{i}}\in\times_{h\in F_{i}[s_{i}]}\Delta(\Delta(A_{h}))$
about opponents' play after each $s_{i}$ and $j$ holds belief $\nu_{s_{j}}\in\times_{h\in F_{j}[s_{j}]}\Delta(\Delta(A_{h}))$
about opponents' play after each $s_{j},$ such that $\nu_{s_{i}^{*}}|_{F_{i}[s_{i}^{*}]\cap F_{j}[s_{j}^{*}]}=\nu_{s_{j}^{*}}|_{F_{i}[s_{i}^{*}]\cap F_{j}[s_{j}^{*}]}$
and $\nu_{s_{i}^{'}}|_{F_{i}[s_{i}^{'}]\cap F_{j}[s_{j}^{'}]}=\nu_{s_{j}^{'}}|_{F_{i}[s_{i}^{'}]\cap F_{j}[s_{j}^{'}]}$.
If $s_{j}^{*}$ has the weakly highest Gittins index for $j$ under
effective discount factor $0\le\delta\gamma<1$, then $s_{i}^{'}$
does not have the weakly highest Gittins index for $i$ under the
same effective discount factor.
\end{prop}
\begin{proof}
We begin by defining a collection of totally mixed correlated profiles
$(\alpha_{[s_{j}]})_{s_{j}\in\mathbb{S}_{j}}$ where $\alpha_{[s_{j}]}\in\Delta^{\circ}(\times_{h\in F_{j}[s_{j}]}A_{h})$.
For each $s_{j}\ne s_{j}^{'}$ the profile $\alpha_{[s_{j}]}$ is
the synthetic correlated profile $\alpha(\nu_{s_{j}},\tau_{s_{j}}^{*})$,
where $\tau_{s_{j}}^{*}$ is an optimal pure-strategy stopping time
in $j$'s auxiliary stopping problem involving $s_{j}.$ For $s_{j}=s_{j}^{'}$,
the correlated profile $\alpha_{[s_{j}^{'}]}$ is instead the synthetic
correlated profile associated with the mimicking stopping rule for
$\tau_{s_{i}^{'}}^{*}$, i.e. mimicking agent $i$'s pure-strategy
optimal stopping time in $i$'s auxiliary problem for $s_{i}^{'}.$

Next, define a profile of totally mixed correlated actions $(\alpha_{[s_{i}]})_{s_{i}\in\mathbb{S}_{i}}$
for $i$'s opponents on information sets $(F_{i}[s_{i}])_{s_{i}\in\mathbb{S}_{i}}$.
For each $s_{i}\notin\{s_{i}^{*},s_{i}^{'}\},$ just use the marginal
distribution of $\alpha_{[\varphi(s_{i})]}$ constructed before on
$F_{i}[s_{i}]\cap F_{j}[\varphi(s_{i})]$, then arbitrarily specify
play in $F_{i}[s_{i}]\backslash F_{j}[\varphi(s_{i})],$ if any. For
$s_{i}'$ the correlated profile is $\alpha(\nu_{s_{i}^{'}},\tau_{s_{i}^{'}}^{*})$,
i.e. the synthetic move associated with $i$'s optimal stopping rule
for $s_{i}^{'}.$ Finally, for $s_{i}^{*},$ the correlated profile
$\alpha_{[s_{i}^{*}]}$ is the synthetic correlated profile associated
with the mimicking stopping rule for $\tau_{s_{_{j}^{*}}}^{*}$.

From Lemma \ref{lem:mimic}, for every $s_{i}$, the profiles of correlated
actions $\alpha_{[s_{i}]}$ and $\alpha_{[\varphi(s_{i})]}$ agree
when marginalized to the information sets $F_{i}[s_{i}]\cap F_{j}[\varphi(s_{i})]$.
Therefore, $(\alpha_{[s_{i}]})_{s_{i}\in\mathbb{S}_{i}}$ and $(\alpha_{[s_{j}]})_{s_{j}\in\mathbb{S}_{j}}$
can be completed into two totally mixed correlated strategy profiles,
$\tilde{\sigma}$ and $\sigma$ (over $\mathbb{S})$, such that $\tilde{\sigma}|_{F_{i}[s_{i}]\cap F_{j}[\varphi(s_{i})]}=\sigma|_{F_{i}[s_{i}]\cap F_{j}[\varphi(s_{i})]}$
for every $s_{i}$. For each $s_{j}\ne s_{j}^{'},$ the Gittins index
of $s_{j}$ for $j$ is $u_{j}(s_{j},\sigma_{s_{j}}).$ Also, since
$\alpha_{[s_{j}^{'}]}$ is the mixed profile associated with the suboptimal
mimicking stopping time, $u_{j}(s_{j}^{'},\sigma_{s_{j}^{'}})$ is
no larger than the Gittins index of $s_{j}^{'}$ for $j.$ By the
hypothesis that $s_{j}^{*}$ has the weakly highest Gittins index
for $j$,$u_{j}(s_{j}^{*},\sigma_{s_{j}^{*}})\ge\max_{s_{j}\ne s_{j}^{*}}u_{j}(s_{j},\sigma_{s_{j}}).$
By the definition of strong player compatibility, we must also have
$u_{i}(s_{i}^{*},\sigma_{s_{i}^{*}})>\max_{s_{i}\ne s_{i}^{*}}u_{i}(s_{i},\sigma_{s_{i}}),$
so in particular $u_{i}(s_{i}^{*},\sigma)>u_{i}(s_{i}'\sigma_{s_{i}^{'}}).$
But $u_{i}(s_{i}^{*},\sigma_{s_{i}^{*}})$ is no larger than the Gittins
index of $s_{i}^{*},$ for $\alpha_{[s_{i}^{*}]}$ is the synthetic
strategy associated with a suboptimal mimicking stopping time. As
$u_{i}(s_{i}'\sigma_{s_{i}^{'}})$ is equal to the Gittins index of
$s_{i}'$ this shows $s_{i}^{'}$ cannot have even weakly the highest
Gittins index at this belief,for $s_{i}^{*}$ already has a strictly
higher Gittins index than $s_{i}^{'}$ does.
\end{proof}
To see that OPT is index compatible for $s_{i}^{*},s_{j}^{*}$ under
the conditions of Theorem \ref{thm:crossplayer_tremble_foundation},
let histories $y_{i},y_{j}$ and strategy $s_{i}^{'}\ne s_{i}^{*}$
be given with $y_{i,s_{i}^{*}}\sim y_{j,s_{j}^{*}}$, $y_{i,s_{i}^{'}}\sim y_{j,\varphi(s_{i}^{'})}$.
Since $g_{i},g_{j}$ are equivalent priors, $i,j$'s posterior beliefs
match on every $F\in F_{i}[s_{i}]\cap F_{j}[\varphi(s_{i})],$ for
$s_{i}\in\{s_{i}^{*},s_{i}^{'}\}$. After such histories, if $s_{j}^{*}$
has weakly the highest Gittins index for $j$, we use the hypothesis
of player compatibility and Proposition \ref{prop:Gittins_matching_beliefs}
to see that $s_{i}^{'}$ does not have the weakly highest Gittins
index for $i.$

\subsection{\label{subsec:Proof-of-Mimic} Proof of Lemma \ref{lem:mimic}}
\begin{proof}
Let $(\tilde{\boldsymbol{a}}_{(-i),t},\tilde{\boldsymbol{a}}_{(C),t})_{t\ge1}$
and $(\tilde{\boldsymbol{e}}_{(-j),t},\tilde{\boldsymbol{e}}_{(C),t})_{t\ge1}$
be the stochastic processes of internal and external histories for
$\tau_{i}$, with distributions $\mathbb{P}_{-i}$ and $\mathbb{P}_{e}$.
Enumerate possible profiles of actions on $F^{C}$ as $\times_{h\in F^{C}}A_{h}=\{\boldsymbol{a}_{(C)}^{(1)},...,\boldsymbol{a}_{(C)}^{(K_{C})}\},$
possible profiles of actions on $\bar{F}^{-i}$ as $\times_{h\in\bar{F}^{-i}}A_{h}=\{\boldsymbol{a}_{(-i)}^{(1)},...,\boldsymbol{a}_{(-i)}^{(K_{-i})}\},$
and possible profiles of actions on $\bar{F}^{-j}$ as $\times_{h\in\bar{F}^{-j}}A_{h}=\{\boldsymbol{a}_{(-j)}^{(1)},...,\boldsymbol{a}_{(-j)}^{(K_{-j})}\}$.

Write $p_{t,(k_{-i},k_{C})}:=\mathbb{P}_{-i}[(\tilde{\boldsymbol{a}}_{(-i),t},\tilde{\boldsymbol{a}}_{(C),t})=(\boldsymbol{a}_{(-i)}^{(k_{-i})},\boldsymbol{a}_{(C)}^{(k_{C})})]$
for $k_{-i}\in\{1,...,K_{-i}\}$ and $k_{C}\in\{1,...,K_{C}\}.$ Also
write $q_{t,(k_{-j},k_{C})}:=\mathbb{P}_{e}[(\tilde{\boldsymbol{e}}_{(-j),t},\tilde{\boldsymbol{e}}_{(C),t})=(\boldsymbol{a}_{(-j)}^{(k_{-j})},\boldsymbol{a}_{(C)}^{(k_{C})})]$
for $k_{-j}\in\{1,...,K_{-j}\}$ and $k_{C}\in\{1,...,K_{C}\}.$ Let
$p_{t,(0,0)}=q_{t,(0,0)}:=\mathbb{P}_{-i}[\tau_{i}<t]=\mathbb{P}_{e}[\tau_{i}<t]$
be the probability of having stopped before period $t.$

The distribution of external histories that $i$ expects to observe
before stopping under belief $\nu_{s_{i}}$when using the mimicking
stopping rule $\tau_{i}$ is the same as the distribution of internal
histories that $j$ expects to observe when using stopping rule $\tau_{j}$
under belief $\nu_{s_{j}}$, because $i$ simulates the data-generating
process on $\bar{F}^{-j}$ by drawing a mixed action $\alpha^{-j}$
according to $j$'s belief $\nu_{s_{j}}|_{\bar{F}^{-j}}$ and $\nu_{s_{i}}|_{F^{C}}=\nu_{s_{j}}|_{F^{C}}$.
Thus for every $k_{-j}\in\{1,...,K_{-j}\}$ and every $k_{C}\in\{1,...,K_{C}\},$
\[
\frac{\sum_{t=1}^{\infty}(\delta\gamma)^{t-1}q_{t,(k_{-j},k_{C})}}{\sum_{t=1}^{\infty}(\delta\gamma)^{t-1}(1-q_{t,(0,0)})}=\alpha(\nu_{s_{j}},\tau_{j})(\boldsymbol{a}_{(-j)}^{(k_{-j})},\boldsymbol{a}_{(C)}^{(k_{C})}).
\]
For a fixed $\bar{k}_{C}\in\{1,...,K_{C}\}$, summing across $k_{-j}$
gives 
\[
\frac{\sum_{t=1}^{\infty}(\delta\gamma)^{t-1}\sum_{k_{-j}=1}^{K_{-j}}q_{t,(k_{-j},\bar{k}_{C})}}{\sum_{t=1}^{\infty}(\delta\gamma)^{t-1}(1-q_{t,(0,0)})}=\alpha(\nu_{s_{j}},\tau_{j})(\boldsymbol{a}_{(C)}^{(\bar{k}_{C})}).
\]
By definition, the processes $(\tilde{\boldsymbol{a}}_{(-i),t},\tilde{\boldsymbol{a}}_{(C),t})_{t\ge0}$
and $(\tilde{\boldsymbol{e}}_{(-j),t},\tilde{\boldsymbol{e}}_{(C),t})_{t\ge0}$
have the same marginal distribution on the second dimension: 
\[
\sum_{k_{-j}=1}^{K_{-j}}q_{t,(k_{-j},\bar{k}_{C})}=\mathbb{P}_{-i}[\tilde{\boldsymbol{a}}_{(C),t}=\boldsymbol{a}_{(C)}^{(\bar{k}_{C})}]=\sum_{k_{-i}=1}^{K_{-i}}p_{t,(k_{-i},\bar{k}_{C})}.
\]
Making this substitution and using the fact that $p_{t,(0,0)}=q_{t,(0,0)}$,
\[
\frac{\sum_{t=1}^{\infty}(\delta\gamma)^{t-1}\sum_{k_{-i}=1}^{K_{-i}}p_{t,(k_{-i},\bar{k}_{C})}}{\sum_{t=1}^{\infty}(\delta\gamma)^{t-1}(1-p_{t,(0,0)})}=\alpha(\nu_{s_{j}},\tau_{j})(\boldsymbol{a}_{(C)}^{(\bar{k}_{C})}).
\]
But by the definition of synthetic correlated profile, the LHS is
$\sum_{k_{-i}=1}^{K_{-i}}\alpha(\nu_{s_{i}},\tau_{i})(\boldsymbol{a}_{(-i)}^{(k_{-i})},\boldsymbol{a}_{(C)}^{(\bar{k}_{C})})=\alpha(\nu_{s_{i}},\tau_{i})(\boldsymbol{a}_{(C)}^{(\bar{k}_{C})})$.

Since the choice of $\boldsymbol{a}_{(C)}^{(\bar{k}_{C})}\in\times_{h\in F^{C}}A_{h}$
was arbitrary, we have shown that the synthetic profile $\alpha(\nu_{s_{j}},\tau_{j})$
of the original stopping rule $\tau_{j}$ and the one associated with
the mimicking strategy of $i,$ $\alpha(\nu_{s_{i}},\tau_{i}),$ coincide
on $F^{C}$.
\end{proof}
\newpage{}
\begin{center}
\textbf{\Large{}Online Appendix}
\par\end{center}
\section{\label{sec:Proofs-Omitted-from}Proofs Omitted from the Appendix}

\subsection{Proof of Proposition \ref{prop:transitive}}
\begin{proof}
Suppose $s_{k}^{*}$ is weakly optimal for $k$ against some totally
mixed correlated profile $\sigma^{(k)}$. We show that $s_{i}^{*}$
is strictly optimal for $i$ against any totally mixed and correlated
$\sigma^{(i)}$ with the property that $\text{marg}_{-ik}(\sigma^{(k)})=\text{marg}_{-ik}(\sigma^{(i)})$.

To do this, we first modify $\sigma^{(i)}$ into a new totally profile
by copying how the action of $i$ correlates with the actions of $-(ik)$
in $\sigma^{(k)}$. For each $s_{-ik}\in\mathbb{S}_{-ik}$ and $s_{i}\in\mathbb{S}_{i}$,
$\sigma^{(k)}(s_{i},s_{-ik})>0$ since $\text{marg}_{-k}(\sigma^{(k)})\in\Delta^{\circ}(\mathbb{S}_{-k})$.
So write $p(s_{i}\mid s_{-ik}):=\frac{\sigma^{(k)}(s_{i},s_{-ik})}{\sum_{s_{i}^{'}\in\mathbb{S}_{i}}\sigma^{(k)}(s_{i}'s_{-ik})}>0$
as the conditional probability that $i$ plays $s_{i}$ given $-ik$
play $s_{-ik},$ in the profile $\sigma^{(k)}$. Now construct the
profile $\hat{\hat{\sigma}}\in\Delta^{\circ}(\mathbb{S}),$ where
\[
\hat{\hat{\sigma}}(s_{i},s_{-ik},s_{k}):=p(s_{i}\mid s_{-ik})\cdot\sigma^{(i)}(s_{-ik},s_{k}).
\]
 Profile $\hat{\hat{\sigma}}$ has the property that $\text{marg}_{-jk}(\hat{\hat{\sigma}})=\text{marg}_{-jk}(\sigma^{(k)})$.
To see this, note first that because $\hat{\hat{\sigma}}$ and $\sigma^{(k)}$
agree on the $-(ijk)$ marginal $\text{marg}_{-ik}(\sigma^{(k)})=\text{marg}_{-ik}(\sigma^{(i)})$.
Also, by construction, the conditional distribution of $i$'s action
given profile of $(-ijk)$'s actions is the same.

From the hypothesis that $s_{j}^{*}\succsim s_{k}^{*},$ we get $j$
finds $s_{j}^{*}$ strictly optimal against $\hat{\hat{\sigma}}$.

But at the same time, $\text{marg}_{-i}(\hat{\hat{\sigma}})=\text{marg}_{-i}(\sigma^{(i)})$
by construction, so this implies also $\text{marg}_{-ij}(\hat{\hat{\sigma}})=\text{marg}_{-ij}(\sigma^{(i)})$.
From $s_{i}^{*}\succsim s_{j}^{*}$, and the conclusion that $j$
finds $s_{j}^{*}$ strictly optimal against $\hat{\hat{\sigma}}$
just obtained, we get $i$ finds $s_{i}^{*}$ strictly optimal against
$\sigma^{(i)}$ as desired.
\end{proof}

\subsection{\label{subsec:Proof-of-Asymm} Proof of Proposition \ref{prop:asymm}}
\begin{proof}
Suppose that $s_{i}^{*}\succsim s_{j}^{*}$ and that neither (ii)
nor (iii) holds. We show that these assumptions imply $s_{j}^{*}\not\succsim s_{i}^{*}$.

Partition the set $\Delta^{\circ}(\mathbb{S})$ into three subsets,
$\Pi^{+}\cup\Pi^{0}\cup\Pi^{-}$, with $\Pi^{+}$ consisting of $\sigma\in\Delta^{\circ}(\mathbb{S})$
that make $s_{j}^{*}$ strictly better than the best alternative pure
strategy, $\Pi^{0}$ the elements of $\Delta^{\circ}(\mathbb{S})$
that make $s_{j}^{*}$ indifferent to the best alternative, and $\Pi^{-}$
the elements that make $s_{j}^{*}$ strictly worse. (These sets are
well defined because $|\mathbb{S}_{j}|\ge2,$ so $j$ has at least
one alternative pure strategy to $s_{j}^{*}$.) If $\Pi^{0}$ is non-empty,
then there is some $\sigma\in\Pi^{0}$ such that $\sum_{s\in\mathbb{S}}u_{j}(s_{j}^{*},s_{-j})\sigma(s)=\max_{s_{j}^{'}\in\mathbb{S}_{j}}\sum_{s\in\mathbb{S}}u_{j}(s_{j}^{'},s_{-j})\sigma(s).$
Because $s_{i}^{*}\succsim s_{j}^{*}$, $\sum_{s\in\mathbb{S}}u_{i}(s_{i}^{*},s_{-i})\hat{\sigma}(s)>\max_{s_{i}^{'}\in\mathbb{S}_{i}\backslash\{s_{i}^{*}\}}\sum_{s\in\mathbb{S}}u_{i}(s_{i}'s_{-i})\hat{\sigma}(s)$
for every $\hat{\sigma}\in\Delta^{\circ}(\mathbb{S})$ such that $\text{marg}_{-ij}(\sigma)=\text{marg}_{-ij}(\hat{\sigma})$.
Since at least one such $\hat{\sigma}$ exists, we do not have $s_{j}^{*}\succsim s_{i}^{*}$.

Also, if both $\Pi^{+}$ and $\Pi^{-}$ are non-empty, then $\Pi^{0}$
is non-empty. This is because both $\sigma\mapsto\sum_{s\in\mathbb{S}}u_{j}(s_{j}^{*},s_{-j})\sigma(s)$
and $\sigma\mapsto\max_{s_{j}^{'}\in\mathbb{S}_{j}\backslash\{s_{j}^{*}\}}\sum_{s\in\mathbb{S}}u_{j}(s_{j}^{'},s_{-j})\sigma(s)$
are continuous functions. If $\sum_{s\in\mathbb{S}}u_{j}(s_{j}^{*},s_{-j})\sigma(s)-\max_{s_{j}^{'}\in\mathbb{S}_{j}\backslash\{s_{j}^{*}\}}\sum_{s\in\mathbb{S}}u_{j}(s_{j}^{'},s_{-j})\sigma(s)>0$
and also $\sum_{s\in\mathbb{S}}u_{j}(s_{j}^{*},s_{-j})\tilde{\sigma}(s)-\max_{s_{j}^{'}\in\mathbb{S}_{j}\backslash\{s_{j}^{*}\}}\sum_{s\in\mathbb{S}}u_{j}(s_{j}^{'},s_{-j})\tilde{\sigma}(s)<0$,
then some mixture between $\sigma$ and $\tilde{\sigma}$ must belong
to $\Pi^{0}$.

So we have shown that if either $\Pi^{0}$ is non-empty or both $\Pi^{+}$
and $\Pi^{-}$ are non-empty, then $s_{j}^{*}\not\succsim s_{i}^{*}$.

If only $\Pi^{+}$ is non-empty, then $s_{j}^{*}$ is strictly interior
dominant for $j$. Together with $s_{i}^{*}\succsim s_{j}^{*},$ this
would imply that $s_{i}^{*}$ is strictly interior dominant for $i$,
contradicting the assumption that (iii) does not hold.

Finally suppose that only $\Pi^{-}$ is non-empty, so that for every
$\sigma\in\Delta^{\circ}(\mathbb{S})$ there exists a strictly better
pure response than $s_{j}^{*}$ against $\sigma_{-j}$. Then, from
Lemma 4 of \citet{pearce1984rationalizable}, there is s a mixed strategy
$\sigma_{j}$ for $j$ that weakly dominates $s_{j}^{*}$ against
all correlated strategy distributions. This $\sigma_{j}$ strictly
dominates $s_{j}^{*}$ against strategy distributions in $\Delta^{\circ}(\mathbb{S}_{-j}),$
so $s_{j}^{*}$ is strictly interior dominated for $j$. Since (ii)
does not hold, there is a $\sigma_{-i}\in\Delta^{\circ}(\mathbb{S}_{-i})$
against which $s_{i}^{*}$ is a weak best response. Then, the fact
that $s_{j}^{*}$ is not a strict best response against any $\sigma_{-j}\in\Delta^{\circ}(\mathbb{S}_{-j})$
means $s_{j}^{*}\not\succsim s_{i}^{*}.$
\end{proof}

\subsection{Proof of Lemma \ref{lem:epsilon_eqm_translate}}
\begin{proof}
We prove the first statement by contraposition. If $\mathfrak{\mathscr{C}}(\bar{\sigma})$
is not an $\mathfrak{\mathscr{C}}(\boldsymbol{\bar{\epsilon}})$-equilibrium
in the base game, then some $i$ assigns more than the required weight
to some $s_{i}^{'}\in\mathbb{S}_{i}$ that does not best respond to
$\mathfrak{\mathscr{C}}(\bar{\sigma})_{-i}.$ This means no $(s_{i}',n_{i})\in\bar{\mathbb{S}}_{i}$
best responds to $\bar{\sigma}_{-i}$, since all copies of a strategy
are payoff equivalent. Since $\mathfrak{\mathscr{C}}(\bar{\sigma})$
and $\mathfrak{\mathscr{C}}(\boldsymbol{\bar{\epsilon}})$ are defined
by adding up the respective extended-game probabilities, $\mathfrak{\mathscr{C}}(\bar{\sigma})_{i}(s_{i}^{'})>\mathfrak{\mathscr{C}}(\boldsymbol{\bar{\epsilon}})(s_{i}^{'})$
means $\sum_{n_{i}}\bar{\sigma}_{i}(s_{i}',n_{i})>\sum_{n_{i}}\bar{\boldsymbol{\epsilon}}(s_{i}',n_{i})$.
So for at least one $n_{i}^{'},$ $\bar{\sigma}_{i}(s_{i}',n_{i}^{'})>\bar{\boldsymbol{\epsilon}}(s_{i}',n_{i}^{'})$,
that is $\bar{\sigma}_{i}$ assigns more than required weight to the
non best response $(s_{i}',n_{i}^{'})\in\bar{\mathbb{S}}_{i}$. We
conclude $\bar{\sigma}$ is not an $\bar{\boldsymbol{\epsilon}}$-equilibrium,
as desired.

Again by contraposition, suppose $\mathfrak{\mathscr{E}}(\sigma)$
is not an $\mathfrak{\mathscr{E}}(\boldsymbol{\epsilon})$-equilibrium
in the extended game. This means some $i$ assigns more than the required
weight to some $(s_{i}',n_{i}^{'})\in\bar{\mathbb{S}}_{i}$ that does
not best respond to $\mathfrak{\mathscr{E}}(\sigma)_{-i}.$ This implies
$s_{i}^{'}$ does not best respond to $\sigma_{-i}.$ By the definition
of $\mathfrak{\mathscr{E}}(\boldsymbol{\epsilon})$ and $\mathfrak{\mathscr{E}}(\sigma)$,
if $\mathfrak{\mathscr{E}}(\sigma)_{i}(s_{i}',n_{i}^{'})>\mathfrak{\mathscr{E}}(\boldsymbol{\epsilon})(s_{i}',n_{i}^{'})$,
then also $\mathfrak{\mathscr{E}}(\sigma)_{i}(s_{i}',n_{i})>\mathfrak{\mathscr{E}}(\boldsymbol{\epsilon})(s_{i}',n_{i})$
for every $n_{i}$ such that $(s_{i}',n_{i})\in\bar{\mathbb{S}}_{i}$.
Therefore, we also have $\sigma_{i}(s_{i}^{'})>\boldsymbol{\epsilon}(s_{i}^{'})$,
so $\sigma$ is not an $\boldsymbol{\epsilon}$-equilibrium in the
base game as desired.
\end{proof}

\subsection{Proof of Proposition \ref{prop:PCE_translate}}
\begin{proof}
Suppose $\bar{\sigma}^{*}$ is a PCE in the extended game. So, we
have $\bar{\sigma}^{(t)}\to\bar{\sigma}^{*}$ where each $\bar{\sigma}^{(t)}$
is an $\boldsymbol{\bar{\epsilon}}^{(t)}$-PCE, and each $\boldsymbol{\bar{\epsilon}}^{(t)}$
is player-compatible (in the extended game sense). This means each
$\mathfrak{\mathscr{C}}(\boldsymbol{\bar{\epsilon}}^{(t)})$ is player
compatible in the base game sense, and furthermore each $\mathfrak{\mathscr{C}}(\bar{\sigma}^{(t)})$
is an $\mathfrak{\mathscr{C}}(\boldsymbol{\bar{\epsilon}}^{(t)})$-equilibrium
(by Lemma \ref{lem:epsilon_eqm_translate}), hence an $\mathfrak{\mathscr{C}}(\boldsymbol{\bar{\epsilon}}^{(t)})$-PCE.
Since $\boldsymbol{\bar{\epsilon}}^{(t)}\to\boldsymbol{0}$, $\mathfrak{\mathscr{C}}(\boldsymbol{\bar{\epsilon}}^{(t)})\to\boldsymbol{0}$
as well. Since $\bar{\sigma}^{(t)}\to\bar{\sigma}^{*},$ $\mathfrak{\mathscr{C}}(\bar{\sigma}^{(t)})\to\mathfrak{\mathscr{C}}(\bar{\sigma}^{*}).$
We have shown $\mathfrak{\mathscr{C}}(\bar{\sigma}^{*})$ is a PCE
in the base game.

The proof of the other statement is exactly analogous.
\end{proof}

\subsection{Proof of Proposition \ref{prop:learning_with_duplicates}}
\begin{proof}
We have $r_{i}=\text{OPT}_{i}$ if and only if for every $\tilde{y}_{i}\in\tilde{Y}_{i},$
$r_{i}(\psi(\tilde{y}_{i}))$ has the (weakly) higher Gittins index.
Since $r_{i},\tilde{r}_{i}$ are equivalent up to duplicates, this
means for any $\tilde{y}_{i}\in\tilde{Y}_{i},$ $\tilde{r}_{i}(\tilde{y}_{i})$
either puts probability 1 on \textbf{Out} or probability 1 on \textbf{In}
and \textbf{In-d}. Since \textbf{In} and \textbf{In-d} can be viewed
as two identical ways of pulling the risky arm in a two-armed bandit
with one safe arm and one risky arm, $\tilde{r}_{i}$ is optimal if
and only if $\tilde{r}_{i}(\tilde{y}_{i})$ assigns positive probability
1 to \textbf{In} and \textbf{In-d }when the risky arm has a (weakly)
higher Gittins index than the safe one. These two statements are equivalent
when $\tilde{r}_{i},r_{i}$ are equivalent up to duplicates, since
the Gittins index of the risky arm is the same under $\tilde{y}_{i}$
and $\psi(\tilde{y}_{i})$. Similarly, $r_{i}=\text{WFP}_{i}$ if
and only if for every $\tilde{y}_{i}\in\tilde{Y}_{i},$ $r_{i}(\psi(\tilde{y}_{i}))$
has the (weakly) higher ``WFP'' index, defined as the one-period
expected payoff of playing a certain strategy against the weighted
fictitious play conjecture of $-i$'s play. These indices are the
same after history $\tilde{y}_{i}$ in the extended stage game and
after $\psi(\tilde{y}_{i})$ in the original stage game.

Finally, let $X_{i}^{t}$ be the random variable representing $i$'s
play in period $t$ in the base game under policy $r_{i}$ and social
distribution $\sigma_{-i}.$ Let $\tilde{X}_{i}^{t}$ be the random
variable representing $i$'s play in period $t$ in the extended game
under policy $\tilde{r}_{i}$ and social distribution $\tilde{\sigma}_{-i}.$
Because $r_{i},\tilde{r}_{i}$ are equivalent up to duplicates to
the empty history, $\mathbb{P}_{r_{i},\sigma_{-i}}[X_{i}^{1}=\text{\textbf{Out}}]=\mathbb{P}_{\tilde{r}_{i},\tilde{\sigma}_{-i}}[\tilde{X}_{i}^{1}=\text{\textbf{Out}}]$.
Since $\sigma_{-i}$ and $\tilde{\sigma}_{-i}$ are $-i$ equivalent,
$(r_{i},\sigma_{-i})$ and $(\tilde{r}_{i},\tilde{\sigma}_{-i})$
generate the same distribution over length-1 histories (up to duplicates),
i.e. $\mathbb{P}_{r_{i},\sigma_{-i}}[y_{i}]=\mathbb{P}_{\tilde{r}_{i},\tilde{\sigma}_{-i}}[\psi^{-1}(y_{i})]$
for all $y_{i}\in(\{\text{\textbf{In}},\text{\textbf{Out}}\}\times\mathbb{R})$.
By induction suppose $\mathbb{P}_{r_{i},\sigma_{-i}}[y_{i}]=\mathbb{P}_{\tilde{r}_{i},\tilde{\sigma}_{-i}}[\psi^{-1}(y_{i})]$
for all $y_{i}\in(\{\text{\textbf{In}},\text{\textbf{Out}}\}\times\mathbb{R})^{t}$,
for some $t\ge1.$ If $r_{i}(y_{i})=\text{\text{\textbf{Out}}}$,
then using the fact that $r_{i},\tilde{r}_{i}$ are equivalent up
to duplicates, $\tilde{r}_{i}(\tilde{y}_{i})(\text{\text{\textbf{Out}}})=1$
for all $\tilde{y}_{i}\in\psi^{-1}(y_{i}).$ Thus, for all $x\in\mathbb{R},$
by the inductive hypothesis $\mathbb{P}_{r_{i},\sigma_{-i}}[(y_{i},\text{\text{\textbf{Out}}},x)]=\mathbb{P}_{\tilde{r}_{i},\tilde{\sigma}_{-i}}[\psi^{-1}(y_{i})\times(\text{\text{\textbf{Out}}},x)]$,
and $\mathbb{P}_{r_{i},\sigma_{-i}}[(y_{i},\text{\text{\textbf{In}}},x)]=\mathbb{P}_{\tilde{r}_{i},\tilde{\sigma}_{-i}}[\psi^{-1}(y_{i})\times(\text{\text{\textbf{In}}},x)]=\mathbb{P}_{\tilde{r}_{i},\tilde{\sigma}_{-i}}[\psi^{-1}(y_{i})\times(\text{\text{\textbf{In-d}}},x)]=0.$
On the other hand, if $r_{i}(y_{i})=\text{\text{\textbf{In}}}$, then
using the fact that $r_{i},\tilde{r}_{i}$ are equivalent up to duplicates,
$\tilde{r}_{i}(\tilde{y}_{i})(\text{\text{\textbf{In}}})+\tilde{r}_{i}(\tilde{y}_{i})(\text{\text{\textbf{In-d}}})=1$
for all $\tilde{y}_{i}\in\psi^{-1}(y_{i}).$ Thus, for all $x\in\mathbb{R},$
by the inductive hypothesis, $\mathbb{P}_{r_{i},\sigma_{-i}}[(y_{i},\text{\text{\textbf{Out}}},x)]=\mathbb{P}_{\tilde{r}_{i},\tilde{\sigma}_{-i}}[\psi^{-1}(y_{i})\times(\text{\text{\textbf{Out}}},x)]=0$,
and $\mathbb{P}_{r_{i},\sigma_{-i}}[(y_{i},\text{\text{\textbf{In}}},x)]=\mathbb{P}_{\tilde{r}_{i},\tilde{\sigma}_{-i}}[\psi^{-1}(y_{i})\times(\text{\text{\textbf{In}}},x)]+\mathbb{P}_{\tilde{r}_{i},\tilde{\sigma}_{-i}}[\psi^{-1}(y_{i})\times(\text{\text{\textbf{In-d}}},x)].$
In either case, we get $\mathbb{P}_{r_{i},\sigma_{-i}}[y_{i}]=\mathbb{P}_{\tilde{r}_{i},\tilde{\sigma}_{-i}}[\psi^{-1}(y_{i})]$
for all $y_{i}\in(\{\text{\textbf{In}},\text{\textbf{Out}}\}\times\mathbb{R})^{t+1}$,
and also $\mathbb{P}_{r_{i},\sigma_{-i}}[X_{i}^{t}=\text{\textbf{Out}}]=\mathbb{P}_{\tilde{r}_{i},\tilde{\sigma}_{-i}}[\tilde{X}_{i}^{t}=\text{\textbf{Out}}]$.
By induction we get $\mathbb{P}_{r_{i},\sigma_{-i}}[X_{i}^{t}=\text{\textbf{Out}}]=\mathbb{P}_{\tilde{r}_{i},\tilde{\sigma}_{-i}}[\tilde{X}_{i}^{t}=\text{\textbf{Out}}]$
for every $t\ge1,$ thus $\phi_{i}(\text{\textbf{In}};r_{i},\sigma_{-i})=\phi_{i}(\text{\textbf{In}};\tilde{r}_{i},\tilde{\sigma}_{-i})+\phi_{i}(\text{\textbf{In-d}};\tilde{r}_{i},\tilde{\sigma}_{-i})$.
\end{proof}

\section{\label{sec:Refinements-in-the}Refinements in the Link-Formation
Game}
\begin{prop}
Each of the following refinements selects the same subset of pure
Nash equilibria when applied to the anti-monotonic and co-monotonic
versions of the link-formation game: extended proper equilibrium,
proper equilibrium, trembling-hand perfect equilibrium, $p$-dominance,
Pareto efficiency, and strategic stability. Pairwise stability does
not apply to the link-formation game. Finally, the link-formation
game is not a potential game.
\end{prop}
\begin{proof}
\textbf{Step 1. Extended proper equilibrium, proper equilibrium, and
trembling-hand perfect equilibrium allow the ``no links'' equilibrium
in both versions of the game. }For $(q_{i})$ anti-monotonic with
$(c_{i}),$ for each $\epsilon>0$ let N1 and S1 play \textbf{Active}
with probability $\epsilon^{2}$, N2 and S2 play \textbf{Active} with
probability $\epsilon$. For small enough $\epsilon$, the expected
payoff of \textbf{Active} for player $i$ is approximately $(10-c_{i})\epsilon$
since terms with higher order $\epsilon$ are negligible. It is clear
that this payoff is negative for small $\epsilon$ for every player
$i$, and that under the utility re-scalings $\beta_{N1}=\beta_{S1}=10,$
$\beta_{N2}=\beta_{S2}=1,$ the loss to playing \textbf{Active} smaller
for N2 and S2 than for N1 and S1. So this strategy profile is a $(\boldsymbol{\beta},\epsilon)$-extended
proper equilibrium. Taking $\epsilon\to0$, we arrive at the equilibrium
where each player chooses \textbf{Inactive} with probability 1.

For the version with $(q_{i})$ co-monotonic with $(c_{i}),$ consider
the same strategies without re-scalings, i.e. $\boldsymbol{\beta}=\boldsymbol{1}$.
Then already the loss to playing \textbf{Active} smaller for N2 and
S2 than for N1 and S1, making the strategy profile a $(\boldsymbol{1},\epsilon)$-extended
proper equilibrium.

These arguments show that the ``no links'' equilibrium is an extended
proper equilibrium in both versions of the game. Every extended proper
equilibrium is also proper and trembling-hand perfect, which completes
the step.

\textbf{Step 2. $p-$dominance eliminates the ``no links'' equilibrium
in both versions of the game.} Regardless of whether $(q_{i})$ are
co-monotonic or anti-monotonic with $(c_{i})$, under the belief that
all other players choose \textbf{Active} with probability $p$ for
$p\in(0,1)$, the expected payoff of playing \textbf{Active} (due
to additivity across links) is $(1-p)\cdot0+p\cdot(10+30-2c_{i})>0$
for any $c_{i}\in\{14,19\}$.

\textbf{Step 3. Pareto eliminates the ``no links'' equilibrium in
both versions of the game.} It is immediate that the no-links equilibrium
outcome is Pareto dominated by the all-links equilibrium outcome under
both parameter specifications, so Pareto efficiency would rule it
out whether $(c_{i})$ is anti-monotonic or co-monotonic with $(q_{i})$.

\textbf{Step 4. Strategic stability} \citep{kohlberg_strategic_1986}\textbf{
eliminates the ``no links'' equilibrium in both versions of the
game.} First suppose the $(c_{i})$ are anti-monotonic with $(q_{i}).$
Let $\eta=1/100$ and let $\epsilon^{'}>0$ be given. Define $\epsilon_{N1}(\text{\textbf{Active}})=\epsilon_{S1}(\text{\textbf{Active}})=2\epsilon^{'},$
$\epsilon_{N2}(\text{\textbf{Active}})=\epsilon_{S2}(\text{\textbf{Active}})=\epsilon^{'}$
and $\epsilon_{i}(\text{\textbf{Inactive}})=\epsilon^{'}$ for all
players $i$. When each $i$ is constrained to play $s_{i}$ with
probability at least $\epsilon_{i}(s_{i}),$ the only Nash equilibrium
is for each player to choose \textbf{Active} with probability $1-\epsilon^{'}.$
(To see this, consider N2's play in any such equilibrium $\sigma.$
If N2 weakly prefers \textbf{Active}, then N1 must strictly prefer
it, so $\sigma_{N1}(\mathbf{Active})=1-\epsilon^{'}\ge\sigma_{N2}(\mathbf{Active}).$
On the other hand, if N2 strictly prefers \textbf{Inactive}, then
$\sigma_{N2}(\mathbf{Active})=\epsilon^{'}<2\epsilon^{'}\le\sigma_{N1}(\mathbf{Active})$.
In either case, $\sigma_{N1}(\mathbf{Active})\ge\sigma_{N2}(\mathbf{Active})$.)
When both North players choose \textbf{Active} with probability $1-\epsilon^{'}$,
each South player has \textbf{Active} as their strict best response,
so $\sigma_{S1}(\mathbf{Active})=\sigma_{S2}(\mathbf{Active})=1-\epsilon^{'}$.
Against such a profile of South players, each North player has \textbf{Active}
as their strict best response, so $\sigma_{N1}(\mathbf{Active})=\sigma_{N2}(\mathbf{Active})=1-\epsilon^{'}$.

Now suppose the $(c_{i})$ are co-monotonic with $(q_{i})$. Again
let $\eta=1/100$ and let $\epsilon^{'}>0$ be given. Define $\epsilon_{N1}(\text{\textbf{Active}})=\epsilon_{S1}(\text{\textbf{Active}})=\epsilon^{'},$
$\epsilon_{N2}(\text{\textbf{Active}})=\epsilon^{'}/1000,$ $\epsilon_{S2}(\text{\textbf{Active}})=\epsilon^{'}$
and $\epsilon_{i}(\text{\textbf{Inactive}})=\epsilon^{'}$ for all
players $i$. Suppose by way of contradiction there is a Nash equilibrium
$\sigma$ of the constrained game which is $\eta$-close to the \textbf{Inactive}
equilibrium. In such an equilibrium, N2 must strictly prefer \textbf{Inactive},
otherwise N1 strictly prefers \textbf{Active} so $\sigma$ could not
be $\eta$-close to the \textbf{Inactive} equilibrium. Similar argument
shows that S2 must strictly prefer \textbf{Inactive}. This shows N2
and S2 must play \textbf{Active} with the minimum possible probability,
that is $\sigma_{N2}(\text{\textbf{Active}})=\epsilon^{'}/1000$ and
$\sigma_{S2}(\text{\textbf{Active}})=\epsilon^{'}$ . This implies
that, even if $\sigma_{N1}(\text{\textbf{Active}})$ were at its minimum
possible level of $\epsilon^{'}$, S1 would still strictly prefer
playing \textbf{Inactive} because S1 is 1000 times as likely to link
with the low-quality opponent as the high-quality opponent. This shows
$\sigma_{S1}(\text{\textbf{Active}})=\epsilon^{'}$. But when $\sigma_{S1}(\text{\textbf{Active}})=\sigma_{S2}(\text{\textbf{Active}})=\epsilon^{'}$,
$N1$ strictly prefers playing \textbf{Active}, so $\sigma_{N1}(\text{\textbf{Active}})=1-\epsilon^{'}$.
This contradicts $\sigma$ being $\eta$-close to the no-links equilibrium.

\textbf{Step 5. Pairwise stability }\citep{jackson1996strategic}\textbf{
does not apply to this game}. This is because each player chooses
between either linking with every player on the opposite side who
plays \textbf{Active}, or linking with no one. A player cannot selectively
cut off one of their links while preserving the other.

\textbf{Step 6. The game does not have an ordinal potential, so refinements
of potential games }\citep{monderer1996potential}\textbf{ do not
apply}. To see that this is not a potential game, consider the anti-monotonic
parameterization. Suppose a potential $P$ of the form $P(a_{N1},a_{N2},a_{S1},a_{S2})$
exists, where $a_{i}=1$ corresponds to $i$ choosing \textbf{Active},
$a_{i}=0$ corresponds to $i$ choosing \textbf{Inactive}. We must
have 
\[
P(0,0,0,0)=P(1,0,0,0)=P(0,0,0,1),
\]
since a unilateral deviation by one player from the \textbf{Inactive}
equilibrium does not change any player's payoffs. But notice that
$u_{N1}(1,0,0,1)-u_{N1}(0,0,0,1)=10-14=-4,$ while $u_{S2}(1,0,0,1)-u_{S2}(1,0,0,0)=30-19=11$.
If the game has an ordinal potential, then both of these expressions
must have the same sign as $P(1,0,0,1)-P(1,0,0,0)=P(1,0,0,1)-P(0,0,0,1)$,
which is not true. A similar argument shows the co-monotonic parameterization
does not have a potential either.
\end{proof}

\end{document}